\documentclass[letter paper, 11pt]{article}
\usepackage[margin = 1in]{geometry}
\usepackage{amssymb,amsmath,amsthm}
\usepackage[colorlinks,citecolor=blue,linkcolor=blue,urlcolor=red]{hyperref}
\usepackage{lineno}
\usepackage{color}
\usepackage{enumitem}
\usepackage[noend]{algpseudocode}
\usepackage{thm-restate}
\usepackage{algorithm}
\usepackage[T1]{fontenc}
\usepackage{physics}

\newcommand*\patchAmsMathEnvironmentForLineno[1]{%
  \expandafter\let\csname old#1\expandafter\endcsname\csname #1\endcsname
  \expandafter\let\csname oldend#1\expandafter\endcsname\csname end#1\endcsname
  \renewenvironment{#1}%
     {\linenomath\csname old#1\endcsname}%
     {\csname oldend#1\endcsname\endlinenomath}}%
\newcommand*\patchBothAmsMathEnvironmentsForLineno[1]{%
  \patchAmsMathEnvironmentForLineno{#1}%
  \patchAmsMathEnvironmentForLineno{#1*}}%
\AtBeginDocument{%
\patchBothAmsMathEnvironmentsForLineno{equation}%
\patchBothAmsMathEnvironmentsForLineno{align}%
\patchBothAmsMathEnvironmentsForLineno{flalign}%
\patchBothAmsMathEnvironmentsForLineno{alignat}%
\patchBothAmsMathEnvironmentsForLineno{gather}%
\patchBothAmsMathEnvironmentsForLineno{multline}%
}

\newtheorem{theorem}{Theorem}
\newtheorem{definition}[theorem]{Definition}
\newtheorem{corollary}[theorem]{Corollary}
\newtheorem{fact}[theorem]{Fact}
\newtheorem{claim}[theorem]{Claim}
\newtheorem{lemma}[theorem]{Lemma}
\newtheorem{observation}[theorem]{Observation}
\newtheorem*{corollary*}{Corollary}

\newbox\ProofSym
\setbox\ProofSym=\hbox{%
\unitlength=0.18ex%
\begin{picture}(10,10)
\put(0,0){\framebox(9,9){}}
\put(0,3){\framebox(6,6){}}
\end{picture}}

\renewcommand{\vec}{\vb*}
\renewcommand{\det}{\mathrm{det}}

\title{
   Dense Subset Sum in Multi-Dimension
}

\author{
  Lin Chen, Tingwei Hu, Yuchen Mao, Guochuan Zhang\footnote{\{linchen198662,tingweihu,maoyc,zgc\}@zju.edu.cn. Zhejiang University.}
}

\date{\today}

\begin{document}
\maketitle

\begin{abstract}
    We study the additive structure of dense subset sum in multi-dimension, and use the structure to develop efficient algorithms for the dense subset sum problem. More precisely, given a set $A$ of $n$ vectors in the $d$-dimensional hyperrectangle $[N_1]\times [N_2]\times\cdots\times [N_d]$, we study the structure of $\mathcal{S}(A)$, which is the set of all subset sums of $A$, and then utilize the structure to develop efficient algorithms for determining whether $\vec{t}\in \mathcal{S}(A)$ for a given $\vec{t}$. We focus on the dense regime of the problem where $n \gg \sqrt{\Phi}$ and $\Phi = N_1 \times \cdots \times N_d$.

    The problem is well understood in the 1-dimensional case. It is known that if $n \gg \sqrt{\Phi}$, then $\mathcal{S}(A)$ contains an arithmetic progression of length $\Phi$~[S{\'a}rk{\"o}zy '94, Szemer{\'e}di \& Vu '06]. If, in addition, no prime divides most elements of $A$, then $\mathcal{S}(A)$ contains all integers in the range $[o(1)\sigma(A), (1-o(1))\sigma(A)]$, where $\sigma(A) = \sum_{a \in A}a$~[Lev '03]. Using this combinatorics result, it has been shown that under the sole condition that $n \gg \sqrt{\Phi}$, one can determine whether $\vec{t}\in \mathcal{S}(A)$ in $\widetilde{O}(n)$ time for any $t \in [o(1)\sigma(A), (1-o(1))\sigma(A)]$ [Galil \& Margalit '91, Bringmann \& Wellnitz '21].

    For multi-dimension, the problem is far from clear. It is only known that if no nontrivial lattice (lattice other than $\mathbb{Z}^d$) contains most elements of $A$, then $\mathcal{S}(A)$ contains all integer points within some ellipsoid, provided that $n \gg \Phi^{2/3}$ for $d = 2$~[Freiman '96, Plagne '99], and $n \gg \Phi^{\frac{d-1}{d}}$ for $d\geq 3$~[Chaimovich '91]. 

    We improve upon these combinatorics results by showing that for any constant $d\geq 1$, if $n \gg \sqrt{\Phi}$, then $\mathcal{S}(A)$ contains a long generalized progression in multi-dimension. If we further have that no nontrivial lattice can contain the majority of $A$, then $\mathcal{S}(A)$ contains all the integer points in the zonotope $\{x_1\vec{a}_1 + \cdots + x_n\vec{a}_n: o(1)\leq x_j \leq 1-o(1),  x_j \in \mathbb{R}\}$. Compared to the previous results, our result significantly reduces the density threshold and enlarges the region inside which all the integer points belong to $\mathcal{S}(A)$. Indeed, it matches the bound for the 1-dimensional case.

    Using our combinatorics result, we also develop an $\widetilde{O}(n)$-time algorithm for the dense subset sum problem in multi-dimension.
\end{abstract}

\clearpage

\tableofcontents

\clearpage

\section{Introduction}
We study dense subset sum in multi-dimension. Let $[N] :=\{0,1,\cdots,N-1\}$. Let $A \subseteq [N]^d$ be a set of $n$ vectors. Let $\mathcal{S}(A)=\{\sum_{j=1}^nz_j\vec{a}_j : z_j\in\{0,1\}\}$ be the set of all subset sums of $A$.  It is easy to see that when $n$ is very small, $\mathcal{S}(A)$ may consist of only some sporadic points\footnote{Points and vectors are used interchangeably in this paper.}. However, when $n$ is sufficiently large, $\mathcal{S}(A)$ is likely to contain all the integer points within a certain well-structured area. Our goal is to characterize such an area, and exploit this structural information to develop efficient algorithms for the subset sum problem, that is, to determine whether a given $\vec{t}$ belongs to $\mathcal{S}(A)$. 

Our research is motivated by the study of the dense subset sum problem in one dimension.  In additive combinatorics, a classic result states that if $n \gg \sqrt{N}$, then $\mathcal{S}(A)$ admits an arithmetic progression of length at least $N$~\cite{Sar94,Lev03,GM91,SV06b,Fre93}. We say that $A$ is spread if no integer $q>1$ divides most of the integers in $A$ (that is, $A$ does not concentrate on $q\mathbb{Z}$ for some $q > 1$). 
If $n\gg \sqrt{N}$ and $A$ is spread, then Lev~\cite{Lev03} showed that $\mathcal{S}(A)$ contains all the integers within the interval $[o(1)\sigma(A), (1-o(1))\sigma(A)]$, where $\sigma(A) = \sum_{a \in A}a$. 
Such a structural characterization allows dense instances of the subset sum problem to be solved faster~\cite{CFG89,GM91,BW21}.   
Recently, Bringmann and Wellnitz~\cite{BW21} made a comprehensive study of dense subset sum from a computational perspective. They showed that when $n\gg \sqrt{N}$ and $t \in [o(1)\sigma(A), (1-o(1))\sigma(A)]$, even if $A$ is not spread, determining whether $t\in \mathcal{S}(A)$ takes only $\widetilde{O}(n)$ time.\footnote{Throughout this paper, $\widetilde{\Omega}(f)$, $\widetilde{O}(f)$ hide polylogarithmic factors in $f$, $n$, $N$, and $d$.}

When $d\ge 2$, the situation becomes much more complicated and much less is known. Given a set of $d$ linearly independent vectors $B=\{\vec{b}_1,\vec{b}_2,\cdots,\vec{b}_d\}$ in $\mathbb{Z}^d$, we denote by $\mathcal{L}(B)$ the lattice generated by $\vec{b}_1, \ldots, \vec{b}_d$, i.e., 
$\mathcal{L}(B)=\{B\vec{x}: \vec{x}\in\mathbb{Z}^d\}$. Extending the notions from the 1-dimensional case, we say $A$ is spread if $|A\setminus \Gamma|$ is large for any lattice $\Gamma \neq \mathbb{Z}^d$. In 1991, Freiman~\cite{Fre96} showed that for $A\subseteq [N_1]\times [N_2]$, if $n\gg \Phi^{2/3}$, where $\Phi=N_1N_2$, and $A$ is spread, then there exists a moderately large ellipsoid such that all the integer points inside it belong to $\mathcal{S}(A)$. Later, Plagne~\cite{Pla99} made an improvement by enlarging the ellipsoid by some polylogarithmic factor. Chaimovich~\cite{Cha91} further extended the result to the case of $d\ge 3$. He showed that for $A\subseteq [N_1]\times [N_2]\times\cdots\times[N_d]$, if $n\gg \Phi^{\frac{d-1}{d}}$, where $\Phi=\prod_j N_j$, and $A$ is spread, then there also exists a moderately large ellipsoid such that all the integer points inside it belong to $\mathcal{S}(A)$. It is noted explicitly in~\cite{Cha91} that the result is only valid for $d\ge 3$. 
Chaimovich~\cite{Cha91} also presented an algorithm showing that when $n\gg \Phi^{\frac{d-1}{d}}$ and $\vec{t}$ lies in some moderately large ellipsoid, then -- even if $A$ is not spread -- one can decide whether $\vec{t}\in \mathcal{S}(A)$ in $\widetilde{O}(n^{2+\frac{1}{d-1}})$-time.

For $d=1$, there is evidence that the density threshold of $n \gg \sqrt{N}$ is likely the best one can hope for. Szemer{\'e}di and Vu~\cite{SV06b} constructed a spread set $A$ of size $O(\sqrt{N})$ such that $|\mathcal{S}(A)|=\Omega(N^{3/2})$ but the longest arithmetic progression in $\mathcal{S}(A)$ has length $O(N^{3/4})$. This suggests that the density threshold for $d=1$ is sharp.

For $d\geq 2$, however, things are far from clear in two aspects: first, the density threshold is much higher than that for $d=1$, and second, the ellipsoid given by Freiman~\cite{Fre96}, Plagne~\cite{Pla99} and Chaimovich~\cite{Cha91} is a very small area that covers, as we will explain later, only an $o(1)$ fraction of the points in $\mathcal{S}(A)$.  Hence, it is natural to ask the following question.

\begin{quote}
    For $d \geq 2$, is it possible to reduce the density threshold and enlarge the area inside which all the integer points belong to $\mathcal{S}(A)$? In particular, can we obtain a bound that matches the 1-dimensional case?
\end{quote}

We shall provide an affirmative answer to the above question. We remark that our approach is fully combinatorial, while all the prior results for $d \geq 2$ are based on analytic number theory. (For $d =1$, both the combinatorial approach~\cite{Sar94,Lev03,GM91,SV06b} and the analytical approach~\cite{AF88} have been explored.) This is partly because the analytical approach can be extended to higher dimensions in a more systematic way, while new tools are needed to extend the combinatorial approach to even $d=2$. But the analytic approach has certain limitations that prevent it from obtaining a sharp bound even for $d = 1$, as we will explain below.

The analytical approach addresses the decision problem by showing that when $\vec{t}$ lies in some ellipsoid centered at the average of $A$,  the number of feasible solutions to subset sum can be approximated by a Gaussian distribution. One way to interpret such a result is to view the subset sum problem from a different perspective: suppose we are given $n$ independent random variables $X_i$, where $\Pr(X_i=\vec{a}_i)=1/2$ and $\Pr(X_i=0)=1/2$. The number of solutions to $\sum_i\vec{a}_ix_i=\vec{t}$ can be calculated via the probability $\Pr(\sum_i X_i=\vec{t})$. While the  central limit theorem (CLT) is not applicable as $X_i$ depends on $n$, it turns out surprisingly that for large $n$ this probability can still be approximated by a Gaussian distribution. This elegant idea was first introduced by Alon and Freiman~\cite{AF88}, who considered the 1-dimensional dense subset sum, and was later extended to the multidimensional case in subsequent works. However, Gaussian approximation for subset sum also has a fundamental limitation -- $\vec{t}$ has to be sufficiently close to the expectation $\sum_i \vec{a}_i/2$, otherwise the approximation error will dominate. For $d=1$, the analytical approach by Alon and Freiman~\cite{AF88} requires that $n\gg N^{2/3}$ and that $A$ is spread, and only guarantees the existence of a feasible solution within an interval of length $O(\sqrt{\sum_i a_i^2})$. The case where $d\geq 2$ faces essentially the same issue.

\subsection{Our Results}\label{sec:res}

To ease our discussion, we consider the case where $A$ is a set of vectors from $[N]^d$ and present our results informally in this section. Formal theorems for the case where $A$ is a multiset of vectors from $[N_1] \times \cdots \times [N_d]$  can be found in Section~\ref{sec:extension}. We use $\widetilde{O}_d(\cdot)$ to further hide factors depending only on $d$. Recall that $\mathcal{L}(A)$ is the lattice generated by the vectors in $A$. At this stage, one may assume that the lattice $Z: = \mathbb{Z}^d$ and $\det(Z) = 1$ in all the following theorems and corollaries in this section.

\paragraph{Long Progressions in Dense Subset Sums}Our first result is the existence of long progressions in dense subset sums.
\begin{theorem}[Informal Version of Theorem~\ref{thm:ap-standard-basis-extended}]\label{thm:ap-standard-basis-informal}
    Let $A\subseteq [N]^d$ be a set of vectors, where $Z: = \mathcal{L}(A)$ is full-dimensional. Suppose that $N$ is sufficiently large and that the following conditions hold.
    \begin{itemize}
        \item (Density Condition) $|A| \gg N^{d/2}$.

        \item (Spreadness Condition) $|A \setminus \Lambda| \gg \frac{N^d}{|A| \cdot \det(Z)}$ for any sublattice $\Lambda \subsetneq Z$.

        \item (Nondegeneracy Condition) $|A \setminus \mathcal{V}| \gg \frac{N^d}{|A|}$ for any lower-dimensional subspace $\mathcal{V} \subseteq \mathbb{R}^d$.
    \end{itemize}
    Then 
    \[
        \mathcal{S}(A) \supseteq \{ \vec{v_0} + z_1s\vec{e}_1+ \cdots+z_ds\vec{e}_d : 0 \leq z_j \leq N, z_j \in \mathbb{Z}\},
    \]
    where $\vec{v}_0\in \mathbb{Z}^d$, $s := \det(Z)$, and $\{\vec{e}_1,\cdots,\vec{e}_d\}$ is the standard basis.
\end{theorem}

We claim that the step size $s = \det(Z) \leq \frac{d^dN^d}{|A|}$ by a standard lattice computation. We also remark that the spreadness condition in the above theorem can actually be removed in the sense that given a dense and nondegenerate set, we can always extract a subset that is dense, spread, and nondegenerate. 

\begin{corollary*}
    Let $A\subseteq [N]^d$ be a set of vectors, where $Z: = \mathcal{L}(A)$ is full-dimensional. Suppose that $N$ is sufficiently large and that the following conditions hold.
    \begin{itemize}
        \item (Density Condition) $|A| \gg N^{d/2}$.

        \item (Nondegeneracy Condition) $|A \setminus \mathcal{V}| \gg \frac{N^d}{|A|}$ for any lower-dimensional subspace $\mathcal{V} \subseteq \mathbb{R}^d$.
    \end{itemize}
    Then 
    \[
        \mathcal{S}(A) \supseteq \{ \vec{v_0} + z_1s\vec{e}_1+ \cdots+z_ds\vec{e}_d : 0 \leq z_j \leq N, z_j \in \mathbb{Z}\},
    \]
    where $\vec{v}_0\in \mathbb{Z}^d$, $s \leq \widetilde{O}_d(\frac{N^d}{|A|})$, and $\{\vec{e}_1,\cdots,\vec{e}_d\}$ is the standard basis.
\end{corollary*}

Note that when $d=1$, the nondegeneracy condition holds automatically, and the corollary matches (up to polylogarithmic factors) the previous theorem for the 1-dimensional case~\cite{Sar94,Lev03,GM91,SV06b,Fre93}.

\paragraph{A John-Type Theorem for Dense Subset Sum} Note that the vectors in $\mathcal{S}(A)$ must belong to the zonotope
\(
    \{x_1\vec{a}_1 + \cdots + x_n\vec{a}_n : 0 \leq x_i \leq 1\}.
\)
Our second result shows that almost all the integer points inside this zonotope belong to $\mathcal{S}(A)$ when $A$ is dense, spread, and nondegenerate.

\begin{theorem}[Informal version of Theorem~\ref{thm:john-type-extended}]\label{thm:john-type-informal}
     Let $A\subseteq [N]^d$ be a set of vectors, where $Z: = \mathcal{L}(A)$ is full-dimensional. Suppose that $N$ is sufficiently large and that the following conditions hold.
    \begin{itemize}
        \item (Density Condition) $|A| \gg \sqrt{\frac{N^d}{\det(Z)}}$.

        \item (Spreadness Condition) $|A \setminus \Lambda| \gg \frac{N^d}{|A| \cdot \det(Z)}$ for any sublattice $\Lambda \subsetneq Z$.

        \item (Nondegeneracy Condition) for some $n'\gg \frac{N^d}{|A| \cdot \det(Z)}$, it holds that $|A \setminus \mathcal{V}| \geq n'$ for any lower-dimensional subspace $\mathcal{V} \subseteq \mathbb{R}^d$.
    \end{itemize}
    Then 
    \[
        Z \cap \{x_1\vec{a}_1 + \cdots + x_n\vec{a}_n : \varepsilon \leq x_i \leq 1 - \varepsilon\} \subseteq \mathcal{S}(A) ,    
    \]
    where $\varepsilon \leq \widetilde{O}_d(\frac{N^d}{n' \cdot |A| \cdot \det(Z)})$.
\end{theorem} 

Note that the nondegeneracy condition is weaker than the spreadness condition, in the sense that if a set is highly spread, then it must be highly  nondegenerate. Therefore, we have the following corollary.

\begin{corollary*}
     Let $A\subseteq [N]^d$ be a set of vectors, where $Z: = \mathcal{L}(A)$ is full-dimensional. Suppose that $N$ is sufficiently large and that the following conditions hold.
    \begin{itemize}
        \item (Density Condition) $|A| \gg \sqrt{\frac{N^d}{\det(Z)}}$.

        \item (Spreadness Condition) for some $n'\gg \frac{N^d}{|A| \cdot \det(Z)}$, it holds that $|A \setminus \Lambda| \geq n'$ for any sublattice $\Lambda \subsetneq Z$.
    \end{itemize}
    Then 
    \[
        Z \cap \{x_1\vec{a}_1 + \cdots + x_n\vec{a}_n : \varepsilon \leq x_i \leq 1 - \varepsilon\} \subseteq \mathcal{S}(A) ,    
    \]
    where $\varepsilon \leq \widetilde{O}_d(\frac{N^d}{n' \cdot |A| \cdot \det(Z)})$.
\end{corollary*}

Theorem~\ref{thm:john-type-informal} implies that for multidimensional dense subset sum, its natural linear programming relaxation essentially captures $\mathcal{S}(A)$, except for the points that are sufficiently close to the boundary of the zonotope. 

When $d = 1$, the region is essentially the same as $[\varepsilon \sigma(A), (1 - \varepsilon)\sigma(A)]$, so Theorem~\ref{thm:john-type-informal} matches prior results for the 1-dimensional case. (Note that in the 1-dimensional case, the nondegenerate condition holds automatically, and one can assume that $n' = |A|$.)

When $d \geq 2$, our result gives a substantial improvement upon prior results in the following aspects. 
\begin{itemize}
    \item Density threshold: 
    Prior results require $n \gg \Phi^{2/3}$ for $d=2$~\cite{Fre96,Pla99}, and $n \gg \Phi^{\frac{d-1}{d}}$ for $d\ge 3$~\cite{Cha91}. Our result unifies the threshold to $n\gg \sqrt{\Phi}$ for every positive constant $d$. This density threshold is sharp (up to some polylogarithmic factors). See Appendix~\ref{apx:lb} for details.

    \item Structured area: our result gives a convex region that captures a $1 - o(1)$ fraction of $\mathcal{S}(A)$, while prior results characterize only a small area of $\mathcal{S}(A)$. To see the latter, consider $A = \{(a_i, b_i)\}_i$. Freiman~\cite{Fre96} showed that when density and spreadness conditions hold, $\mathcal{S}(A)$ contains all integer points $\vec{t}$ within an ellipsoid given by $|(\frac{\sigma(A)}{2}-\vec{t})M^{-1}(\frac{\sigma(A)}{2}-\vec{t})^{\top}|=O(1)$, where 
    \[ 
        M= \frac{1}{4}
        \begin{bmatrix}
            \sum_i a_i^2       & \sum_i a_ib_i \\
            \sum_i a_ib_i       & \sum_i b_i^2
        \end{bmatrix} 
    \]
    Consider the following specific set $A := A_1 \cup A_2$, where
    \begin{align*}
        A_1 &=: \{(a,a+b): \frac{3}{2}N\leq a\leq 2N, 1\leq b\leq \frac{1}{2}N^{1/3}\};\\
        A_2 &=: \{(a,a+b): N\leq a\leq \frac{3}{2}N, \frac{1}{2}N^{1/3}\leq b\leq N^{1/3}\}.
    \end{align*}
    One can verify that the number of integer points within the ellipsoid is $\widetilde{O}(N^{8/3})$. However, $\mathcal{S}(A)$ contains $\Omega(N^4)$ points. So the ellipsoid captures only an $O(N^{-4/3})$ fraction of all the points. In a follow-up paper, Plagne~\cite{Pla99} slightly enlarged the ellipsoid by a logarithmic factor, which hardly changes the situation.
\end{itemize}

\paragraph{A Linear-Time Algorithm for Dense Subset Sum}
\begin{theorem}[Informal version of Theorem~\ref{thm:dense-alg-extended}]\label{thm:dense-alg-informal}
    Let $A\subseteq [N]^d$ be a set of vectors, where $Z: = \mathcal{L}(A)$ is full-dimensional. Suppose that $N$ is sufficiently large and that the following conditions hold.
    \begin{itemize}
        \item (Density Condition) $|A| \gg \sqrt{\frac{N^d}{\det(Z)}}$.

        \item (Nondegeneracy Condition) for some $n' \gg \frac{N^d}{|A| \cdot \det(Z)}$, it holds that $|A \setminus \mathcal{V}| \geq n' $ for any lower-dimensional subspace $\mathcal{V} \subseteq \mathbb{R}^d$.
    \end{itemize}
    Then for any
    \[
        \vec{t}\in \{x_1\vec{a}_1 + \cdots + x_n\vec{a}_n : \varepsilon \leq x_i \leq 1 - \varepsilon\} ,
    \]
    where $\varepsilon \leq \widetilde{O}_d(\frac{N^d}{n' \cdot |A| \cdot \det(Z)})$, we can determine whether $\vec{t} \in \mathcal{S}(A)$ in $\widetilde{O}_d(n)$ time with error probability at most $N^{-\Omega(d)}$.
\end{theorem}
We also show that when $n \gg \sqrt{N^d/\det(Z)}$, in $\widetilde{O}_d(n)$ time, one can determine whether a set $A$ satisfies the nondegeneracy condition, as well as whether $\vec{t}\in \{x_1\vec{a}_1 + \cdots + x_n\vec{a}_n : \varepsilon \leq x_i \leq 1 - \varepsilon\}$. See Appendix~\ref{app:decide-nondegen} and Appendix~\ref{apx:zonotope} for details.

Compared with the algorithms for the 1-dimensional case~\cite{GM91,BW21}, our theorem additionally requires $A$ to be nondegenerate. But when $d=1$, this condition holds automatically, and one can assume that $n' = |A|$. So our theorem matches the previous algorithms for the 1-dimensional case. In particular, the bound on $\varepsilon$ matches (up to polylogarithmic factors) that in Bringmann and Wellnitz's algorithm, and is sharp for the set case under Strong Exponential Time Hypothesis~\cite{BW21}. 

For $d \geq 2$, this requirement seems to be unavoidable. Basically, if a majority of $A$ degenerates to a lower-dimensional subspace, then it is not dense enough along some direction in the original space, and therefore, the techniques for the dense case cannot be applied. We claim that removing this requirement is as hard as developing an $\widetilde{O}(N^{1 + o(1)})$ algorithm for the 1-dimensional subset sum, which would require techniques beyond those for the dense case. (So far the best-known algorithm has a running time $\widetilde{O}(N^{5/4})$ due to Chen, Lian, Mao, and Zhang~\cite{CLMZ24a}.)

To see the claim, consider an arbitrary set $A \subseteq [N]$ of integers. We shall embed it into a 2-dimensional dense case. Let $A'\subseteq [N^{1+\delta}]\times [N]$ for some $\delta =o(1)$. $A'$ consists of two subsets $A_x$ and $A_y$, with $A_x=\{(x,0):x\in [N^{1+\delta}]\}$ being on the $x$-axis, and $A_y=\{(0,y):y\in A\}$ on the $y$-axis. Let $\sigma_x=\sum_{x:(x,0)\in A_x}x$ and $\sigma_y=\sum_{y:(0,y)\in A_y} y$. It is easy to see that $A'$ is dense. Given any $\vec{t}=(t_x,t_y)\in [o(1)\sigma_x,(1-o(1))\sigma_x]\times [o(1)\sigma_y,(1-o(1))\sigma_y]$, determining whether $\vec{t}\in\mathcal{S}(A')$ is equivalent to determining whether $t_y\in\mathcal{S}(A)$.

\subsection{Other Related Work}
Subset Sum is a fundamental NP-hard problem in theoretical computer science~\cite{Kar72}. There have been extensive studies on pseudo-polynomial time algorithms for Subset Sum~\cite{Bel57,Pis99,Pis03,KX17,KX18,Bri17,JW19,PRW21,Cha99,CFG89,GM91,BW21,CLMZ24a}. 
    
Multidimensional Subset Sum is related to the lattice basis computation problem. Specifically, our algorithm relies on efficient algorithms for computing Hermite normal form of given vectors, which has been studied extensively in~\cite{SL96,KB79,CC82,Ili89,LS22,KR25}. Given $n$ column vectors in $\mathbb{R}^d$, the algorithm of Storjohann and Labahn~\cite{SL96} computes the Hermite normal form of the $d\times n$ matrix $B$ formed by these vectors in $\widetilde{O}(nd^{\omega}\log\|B\|)$ time, where $\omega$ denotes the exponent for matrix multiplication and is currently $\omega\le 2.371552$~\cite{WXZ24}. For our purpose, $d$ is a fixed constant; thus, the running time is near-linear in $n$ and serves our needs.
    
Multidimensional Subset Sum is closely related to integer (linear) programming. 
Consider an arbitrary integer program $\min \{\vec{c} \vec{x}: H\vec{x}=\vec{t}, \vec{\ell}\le \vec{x}\le \vec{u}\}$ where $H\in \mathbb{Z}^{d\times n}$, and let $N$ denote the upper bound on the absolute value of all entries in the constraint matrix $H$.
There are two classes of algorithms for integer programming. The first class has a running time exponential in the number of variables $n$. Lenstra~\cite{Len83} presented a $2^{O(n^3)}\cdot\text{poly}(|I|)$-time algorithm, where $\textnormal{poly}(|I|)$ denotes a polynomial in the input length $|I|$. The running time has been improved in several subsequent works~\cite{Kan87, DPV11,Dad12,RR23}. Very recently, Reis and Rothvoss~\cite{RR23} proved that integer programming can be solved in $(\log n)^{O(n)}\cdot\text{poly}(|I|)$ time. The second class has a running time depending on $d$ and $N$.   Papadimitriou~\cite{Pap81} provided an algorithm running in $n^{O(d)}(m\cdot\max\{N,\|\vec{t}\|_{\infty}\})^{O(d^2)}$ time. The running time was later improved by Eisenbrand and Weismantel~\cite{EW19}, and Jansen and Rohwedder~\cite{JR23}. 
    
In recent years, there has been extensive algorithmic research in subset-sum-related optimization problems using additive combinatorics results, see, e.g.~\cite{PRW21,DJM23,RW24,ABF23,CLMZ25,CLMZ24c,CLMZ24b,CLMZ24a,BFN25,Jin25,Jin26,Bri24,Jin24,BN21}.

In this paper, we study the decision version of Multidimensional Dense Subset Sum, and our algorithmic result extends the dense subset sum algorithm by Bringmann and Wellnitz~\cite{BW21} to higher dimensions. It is worth mentioning that very recently, Chen, Mao and Zhang~\cite{CMZ25} studied the search version of dense Subset Sum and showed that in near-linear time it is also possible to return a feasible solution if the answer is ``yes''. Unfortunately, our algorithm in this paper cannot address the search version of Multidimensional Dense Subset Sum. However, since our method is combinatorial, we are optimistic that a feasible solution can be found in near-linear time.

\section{Preliminaries}\label{sec:pre}
\subsection{Notation}
All logarithms in this paper are base 2.

We use $\mathbb{R}$ to denote the set of all reals and $\mathbb{Z}$ to denote the set of all integers. The set of integers from $0$ to $N-1$ is written as $[N]$. That is,
\(
    [N] := \{0, 1, \ldots, N-1\}.
\)
Throughout the paper, we always assume that $N$ is sufficiently large. By sufficiently large, we mean that there exists a universal constant $C$ such that $N \geq C$.

We use $\vec{e}_1, \vec{e}_2, \ldots, \vec{e}_d$ to denote the standard basis of $\mathbb{R}^d$, $\vec{0}$ to denote the vector $(0, \ldots, 0)$, and $\vec{1}$ to denote $(1, \ldots, 1)$. The Euclidean norm of a vector $\vec{a}$ is represented by $\|\vec{a}\|_2$, and the inner product of two vectors $\vec{a}$ and $\vec{b}$ is represented by $\vec{a} \cdot \vec{b}$.

Let $P \subseteq \mathbb{R}^d$ be a set of vectors. 
For any vector $\vec{v}$, we use $P + \vec{v}$ to denote the translate of $P$ along $\vec{v}$. That is, 
\(
    P +  \vec{v} := \{\vec{p} + \vec{v}: \vec{p} \in P\}. 
\)
Let $Q$ be another set of vectors. The sumset of $P$ and $Q$ is defined by the following formula.
\[
    P + Q := \{\vec{p} + \vec{q} : \vec{p} \in P, \vec{q} \in Q\}
\]
The $k$-fold sumset of $P$ is denoted as $kP$. That is, $kP = P$ if $k = 1$, and $kP = (k-1)P + P$ for $k \geq 2$.
We remark that when $P$ is convex, $kP$ can also be viewed as a set obtained by stretching the vectors in $P$ by a factor of $k$. We slightly abuse the notation. When $P$ is convex, for any real $a \geq 0$, we use $aP$ to denote the set obtained by stretching the vectors in $P$ by a factor of $a$. That is,
\(
    a P := \{a\vec{p} : \vec{p} \in P\}.
\) 

Let $A = \{\vec{a}_1, \ldots, \vec{a}_n\}$ be a finite set of vectors. We define $\sigma(A) := \sum_{i}\vec{a}_i$ to be the sum of $A$. We define two sets $\mathcal{S}(A)$ and $\mathcal{P}(A)$ as follows. 
\begin{align*}
    \mathcal{S}(A) &:= \{x_1 \vec{a}_1 + \cdots + x_n\vec{a}_n : x_j \in \{0,1\}\}\\
    \mathcal{P}(A) &:= \{x_1 \vec{a}_1 + \cdots + x_n\vec{a}_n :  |x_j| \leq 1/2\}
\end{align*}
Basically, $\mathcal{S}(A)$ is the set of all subset sums of $A$, and $\mathcal{P}(A)$ is a zonotope generated by the vectors in $A$. One can make the following two easy observations.
\begin{itemize}
    \item $\mathcal{S}(A) \subseteq \mathcal{P}(A) + \frac{\sigma(A)}{2} \subseteq 2\mathcal{P}(A)$.

    \item For any partition $A_1 \cup A_2$ of $A$, we have
    \(
        \mathcal{S}(A) = \mathcal{S}(A_1) + \mathcal{S}(A_2)
    \)
    and 
    \(
        \mathcal{P}(A) = \mathcal{P}(A_1) + \mathcal{P}(A_2).
    \)

\end{itemize}

\subsection{Progressions and Lattices}\label{subsec:lattice}
A generalized arithmetic progression (or progression for short) in an additive group $G$ is defined to be any set of the form
\[
    \left\{a_0 + z_1a_1 + z_2a_2 + \cdots + z_na_n: 0 \leq z_j \leq m_j, z_j \in \mathbb{Z}\right\}, 
\]
where $a_0, a_1, \ldots, a_n \in G$ and $m_j \in \mathbb{Z}$.  We call $a_1, \ldots, a_n$ the basis of the progression. (See~\cite{TV06} for more discussion of progressions.) Throughout this paper, we consider only the case where $G = \mathbb{Z}^d$. In other words, we consider only the progressions of the form
\[
    \{\vec{a}_0 + z_1\vec{a}_1 + z_2\vec{a}_2 + \cdots + z_n\vec{a}_n: 0 \leq z_j \leq m_j, z_j \in \mathbb{Z}\}, 
\]
where $\vec{a}_0, \vec{a}_1, \ldots, \vec{a}_n \in \mathbb{Z}^d$ and $m_j \in \mathbb{Z}$.

A lattice can be regarded as an infinite version of a progression. Let $A = \{\vec{a}_1, \vec{a}_2, \ldots, \vec{a}_n\}$ be a set of $n$ vectors in $\mathbb{Z}^d$. The lattice generated by $A$ is defined as
\[
      \mathcal{L}(A) := \{z_1\vec{a}_1 + z_2\vec{a}_2 + \cdots + z_n\vec{a}_n : z_j \in \mathbb{Z}\}.
\]
If the vectors in $A$ are linearly independent, then we say that $A$ is a basis of $\mathcal{L}(A)$. If we further have $n = d$, then we say that $\mathcal{L}(A)$ is full-dimensional. Sometimes it is more convenient to view $A$ as a matrix whose column vectors are $\vec{a}_1, \ldots, \vec{a}_n$. For $i \in \{1, \ldots, d\}$ and for $j \in \{1, \ldots, n\}$, we use $a_{ij}$ to denote the entry of $A$ in the $i$-th row and $j$-th column. We shall switch between the set notation and the matrix notation without mentioning it explicitly.

A lattice can be represented by its basis. A lattice may have more than one basis, but all are equivalent in the sense that one can always convert one basis into another by multiplying by a unimodular matrix. Throughout the paper, we always represent a lattice by its basis in Hermite normal form. Therefore, whenever we obtain a lattice,  we actually obtain its basis in Hermite normal form. 

\begin{definition}[Hermite Normal Form]\label{def:hnf}
    Let $B \in \mathbb{Z}^{d \times n}$ be a matrix of rank $r$. We say that $B$ is in Hermite normal form if it satisfies the following properties.
    \begin{enumerate}[label={\normalfont (\roman*)}]
        \item $b_{ij} = 0$ if $j > r$. (Only the first $r$ columns are non-zero.)

        \item There exist $1 \leq i_1 < i_2 < \cdots < i_r \leq d$ such that $b_{ij} > 0$ if $i = i_j$ and $b_{ij} = 0$ if $i < i_j$.  (The heights of non-zero columns decrease strictly)

        \item For $j \leq r$ and $j' < j$, $0 \leq b_{i_jj'} < b_{i_jj}$. (The top non-zero entry of each non-zero column is the greatest in that row, and the entries in that row are non-negative)
    \end{enumerate}
\end{definition}

A classic result from Hermite~\cite{Her1851} states that for any integer matrix $A$, there is a unique matrix $B$ in Hermite normal form such that $\mathcal{L}(B) = \mathcal{L}(A)$. It is easy to see that the non-zero columns of $B$ are linearly independent, and form a basis of $\mathcal{L}(B)$. There are various algorithms to compute the Hermite normal form of a given integer matrix, and the current fastest is due to Storjohann and Labahn~\cite{SL96}, which runs in $\widetilde{O}(d^{\omega}n\log N)$, where $N$ is the maximum absolute entry value of $A$ and $\omega$ denotes the exponent for matrix multiplication. Since the dependence on $d$ does not really matter for our purpose, we use the trivial bound with $\omega = 3$.

\begin{lemma}\label{lem:compute-hnf}
    For any matrix $A \subseteq [N]^{d\times n}$, there is a unique matrix $B$ in Hermite normal form such that $\mathcal{L}(B) = \mathcal{L}(A)$. Moreover, $B$ can be computed in   
    \(
        \widetilde{O}(d^3n\log N)
    \)
    time.
\end{lemma}

Let $\Gamma \subseteq \mathbb{Z}^d$ be a lattice with basis $B$. If $\Gamma$ is full-dimensional, we define $\det(\Gamma)$ to be the absolute value of the determinant of its basis. That is, $\det(\Gamma): = |\det(B)|$. Note that $|\det(B)| \geq 1$ since $B$ is a full-rank integer matrix. (We remark that although lattice bases may not be unique, their determinants all have the same absolute value, so the definition of $\det(\Gamma)$ does not depend on the choice of basis.) If $\Gamma$ is not full-dimensional, we define $\det(\Gamma)$ to be $0$. Intuitively, $\det(\Gamma)$ represents the $d$-dimensional volume of a cell of the lattice, which is a parallelepiped formed by its basis. 

Let $\Gamma$ and $\Gamma'$ be two lattices. If $\Gamma$ is a (proper) subset of $\Gamma'$, then we say that $\Gamma$ is a (proper) sublattice of $\Gamma'$, and $\Gamma'$ is a (proper) suplattice of $\Gamma$.

\begin{fact}\label{fact:suplatt}
    Let $\Gamma \subseteq \mathbb{Z}^d$ be a full-dimensional lattice. For any proper suplattice $\Gamma' \supsetneq \Gamma$, we have $\det(\Gamma')$ is a proper divisor of $\det(\Gamma)$.
\end{fact}

\subsection{Remainders Modulo a Lattice}
Let $\Gamma \subseteq \mathbb{Z}^d$ be a full-dimensional lattice. We say two vectors $\vec{u}, \vec{v} \in \mathbb{Z}^d$ are congruent (modulo $\Gamma$) if $\vec{u} - \vec{v} \in \Gamma$, denoted as $\vec{u} \equiv \vec{v} \pmod \Gamma$. When $\Gamma$ is clear from the context, we simply say that $\vec{u}$ and $\vec{v}$ are congruent.

Let $b_{11}, b_{22}, \ldots, b_{dd}$ be the diagonal entries of the basis of $\Gamma$ in Hermite normal form. (Note that $b_{ii} > 0$ for all $i$ since $\Gamma$ is full-dimensional.) We define $\mathbb{Z}^d_\Gamma = [b_{11}] \times [b_{22}] \times \cdots \times [b_{dd}]$. Note that $|\mathbb{Z}^d_{\Gamma}| = \prod_{i=1}^d b_{ii} = \det(\Gamma)$. By the definition of Hermite normal form (Definition~\ref{def:hnf}), it is easy to see that the vectors in $\mathbb{Z}^d_{\Gamma}$ are not congruent to each other and that for every $\vec{v} \in \mathbb{Z}^d$, there is a unique vector in $\mathbb{Z}^d_{\Gamma}$ that is congruent to $\vec{v}$ modulo $\Gamma$. In the language of group theory, $\mathbb{Z}^d_{\Gamma}$ is an additive group isomorphic to the quotient group $\mathbb{Z}^d/\Gamma$. (See Chapter 3 of~\cite{TV06} for more details.) When $d = 1$, $\mathbb{Z}^d_{\Gamma}$ is exactly the cyclic group $\mathbb{Z}_{b_{11}}$ that contains all the remainders modulo $b_{11}$. We generalize the definition of remainders to multi-dimension.

\begin{definition}[Remainders]\label{def:rem}
    Let $\Gamma \subseteq \mathbb{Z}^d$ be a full-dimensional lattice. For every $\vec{v} \in \mathbb{Z}^d$, we define the remainder of $\vec{v}$ modulo $\Gamma$, denoted as $\vec{v} \bmod \Gamma$, to be the unique vector in $\mathbb{Z}^d_{\Gamma}$ that is congruent to $\vec{v}$.
\end{definition}

Let $A\subseteq \mathbb{Z}^d$ be a set of vectors. We define 
\(
    A \bmod \Gamma = \{\vec{a} \bmod \Gamma : \vec{a} \in A\}
\)
to be the set of all remainders (modulo $\Gamma$) of the vectors in $A$. 

Below is a useful fact about the number of remainders.
\begin{fact}
    $|\mathbb{Z}^d \bmod \Gamma| = \det(\Gamma)$. More generally, $|\Gamma' \bmod \Gamma| = \frac{\det(\Gamma)}{\det(\Gamma')}$ for any suplattice $\Gamma' \supseteq \Gamma$.
\end{fact}

When two sets $A$ and $B$ contain the same set of remainders modulo $\Gamma$ (that is, $A \bmod \Gamma = B \bmod \Gamma$), we simply write $A \equiv B \pmod \Gamma$. Recall that $\mathcal{S}(A)$ is the set of all subset sums of $A$. We write $\mathcal{S}(A) \bmod \Gamma$ as $\mathcal{S}_{\Gamma}(A)$. Basically, $\mathcal{S}_{\Gamma}(A)$ is the set of subset sums (modulo $\Gamma$) of $A$. As in the case of $\mathcal{S}(A)$, for any partition $A_1 \cup A_2$ of $A$, we have
\[
    \mathcal{S}_{\Gamma}(A) = (\mathcal{S}_{\Gamma}(A_1) + \mathcal{S}_{\Gamma}(A_2)) \bmod \Gamma.
\]

For a remainder $\vec{r} \in \mathbb{Z}^d_{\Gamma}$, we define its multiplicity in $A$ to be the number of elements in $A$ that are congruent (modulo $\Gamma$) to $\vec{r}$.  We use $\mu(A, \Gamma)$ to denote the largest multiplicity of any remainder (modulo $\Gamma$) in $A$.

The remainders can be computed efficiently. Recall that when we are given a lattice, we are actually given its basis in Hermite normal form.
\begin{restatable}{lemma}{lemcomputesinglerem}\label{lem:compute-rem}
    Let $\Gamma \subseteq \mathbb{Z}^d$ be a full-dimensional lattice.
    \begin{enumerate}[label={\normalfont (\roman*)}]
        \item Given a vector $\vec{v} \in \mathbb{Z}^d$, we can compute $\vec{v} \bmod \Gamma$ in $O(d^2)$ time.

        \item Given a suplattice $\Gamma' \supseteq \Gamma$, we can compute $\Gamma' \bmod \Gamma$ in $O(d^2\ell)$ time, where $\ell := \frac{\det(\Gamma)}{\det(\Gamma')}$.

        \item Given a vector $\vec{h} \in \mathbb{Z}^d$, we can compute $\Gamma' := \Gamma + \mathcal{L}(\vec{h})$ in $O(d^2 \log \ell)$ time, where $\ell := \det(\Gamma)$.
    \end{enumerate}
\end{restatable}
The proof of the above lemma is deferred to Appendix~\ref{apx:comp-latt}.

\subsection{Symmetry Sets and Kneser's Theorem}
The symmetry set is a tool of great importance when we discuss the remainders (modulo some lattice $\Gamma$) contained in a set $S$. Basically, the symmetry set is the set of vectors in $\mathbb{Z}^d_{\Gamma}$ such that the translate of $S$ along these vectors does not change the remainders (modulo $\Gamma$) contained in $S$. We remark that there is a much more concise description of symmetry sets if one considers $+$ as the group operation of the additive group $\mathbb{Z}^d_{\Gamma}$, but this may cause confusion in the future where we have to switch between different groups. To avoid such confusion, we instead use the mod operation to describe symmetry sets.

\begin{definition}\label{def:symset}
    Let $\Gamma \subseteq \mathbb{Z}^d$ be a full-dimensional lattice. Let $S \subseteq \mathbb{Z}^d$. The symmetry set of $S$ (with respect to $\Gamma$) is defined to be
    \[
        \mathrm{Sym}_{\Gamma}(S) := \{\vec{h} \in \mathbb{Z}^d_\Gamma: \vec{h} + S \equiv S \pmod \Gamma\}.
    \]
\end{definition}
An easy but important observation is that if $\vec{h}  + S \equiv S \pmod \Gamma$, then $\vec{h} + S + S' \equiv S + S' \pmod \Gamma$ for any $S'$.  This also implies that if $\vec{h}  + \mathcal{S}(A) \equiv \mathcal{S}(A) \pmod \Gamma$, then $\vec{h}  + \mathcal{S}(A') \equiv \mathcal{S}(A') \pmod \Gamma$ for any superset $A' \supseteq A$.

Suppose that $\mathrm{Sym}_{\Gamma}(S)$ contains a non-zero vector $\vec{h}$. Let $\Gamma' = \Gamma + \mathcal{L}(\vec{h})$. We have $\Gamma' + S \equiv S \pmod \Gamma$, which implies 
\(
    \Gamma' + (S \bmod \Gamma') \equiv S \pmod \Gamma.
\) 
Therefore, instead of discussing $S \bmod \Gamma$, we can discuss $S \bmod \Gamma'$.  Note that $\Gamma'$ is a proper suplattice of $\Gamma$ (since $\vec{h} \in \mathbb{Z}^d_{\Gamma}$ and $\vec{h}$ is non-zero), so by Fact~\ref{fact:suplatt}, $\det(\Gamma') \leq \det(\Gamma)/2$. So we can reduce our discussion to a smaller group $\mathbb{Z}^d_{\Gamma'}$.

The following celebrated theorem by Kneser gives a sufficient condition for $\mathrm{Sym}_{\Gamma}(S)$ to contain a non-zero vector.
\begin{theorem}[{Kneser's Theorem~\cite[Theorem 5.5]{TV06}}]\label{thm:kneser}
    Let $\Gamma \subseteq \mathbb{Z}^d$ be a full-dimensional lattice. Let $S_1, S_2$ be two sets of vectors in $\mathbb{Z}^d$. Let $S := S_1 + S_2$. Then 
    \[
        |\mathrm{Sym}_{\Gamma}(S_1 + S_2)| \geq |S_1 \bmod \Gamma| + |S_2 \bmod \Gamma| - |(S_1 + S_2) \bmod \Gamma|.
    \]
\end{theorem}

\begin{corollary}\label{coro:kneser-1}
    Let $\Gamma \subseteq \mathbb{Z}^d$ be a full-dimensional lattice. Let $S_1, \ldots, S_m$ be $m$ sets of vectors in $\mathbb{Z}^d$. Let $S := S_1 + \cdots + S_m$. If
    \[
        |S_1 \bmod \Gamma| + \cdots + |S_m \bmod \Gamma| \geq |S \bmod \Gamma| + m,
    \]
    then $\mathrm{Sym}_{\Gamma}(S)$ contains a non-zero vector $\vec{h}$.
\end{corollary}
\begin{proof}
    If for some $j \in \{1, \ldots, m-1\}$, we have
    \begin{equation}\label{eq:coro-keneser-2}
        |(S_1 + \cdots + S_j) \bmod \Gamma| + |S_{j+1} \bmod \Gamma| \geq |(S_1 + \cdots + S_{j+1}) \bmod \Gamma| + 2,
    \end{equation}
    then by Theorem~\ref{thm:kneser}, $\mathrm{Sym}_{\Gamma}(S_1 + \cdots + S_{j+1})$ contains a non-zero vector $\vec{h}$, and therefore, so does $\mathrm{Sym}_{\Gamma}(S_1 + \cdots + S_{m})$ since $\vec{h}$ also belongs to it.

    Suppose that~\eqref{eq:coro-keneser-2} does not hold for any $j$. Then, by summing the reverse inequality of~\eqref{eq:coro-keneser-2} over all $j \in \{1, \ldots, m-1\}$, we have
    \[
        |(S_1 + \cdots + S_{m}) \bmod \Gamma| \geq \sum_{i=1}^m|S_i \bmod \Gamma| - (m-1).
    \]
    But this contradicts the condition stated in the corollary.
\end{proof}

We will frequently invoke the above corollary with $S_i$ of the form $\mathcal{S}(A_i)$. For convenience, we restate it as the following.
\begin{corollary}\label{coro:kneser-2}
    Let $\Gamma \subseteq \mathbb{Z}^d$ be a full-dimensional lattice. Let $A_1, \ldots, A_m$ be $m$ disjoint sets of vectors in $\mathbb{Z}^d$. Let $A := A_1 \cup \cdots \cup A_m$. If
    \[
        |\mathcal{S}_{\Gamma}(A_1)| + \cdots + |\mathcal{S}_{\Gamma}(A_m)| \geq |\mathcal{S}_{\Gamma}(A)| + m,
    \]
    then there is a vector $\vec{h} \notin \Gamma$ such that $\vec{h} + \mathcal{S}(A) \equiv \mathcal{S}(A) \pmod \Gamma$.
\end{corollary}

\subsection{Convex Analysis Tools}
We also need tools from convex analysis to show that a zonotope is contained in another zonotope.

\begin{definition}[Support Functions]
    Let $P \subseteq \mathbb{R}^d$ be a polytope. The support function of $P$ is defined to be
    \[
        f_P(\vec{v}) := \max\{\vec{p}\cdot \vec{v} : \vec{p} \in P\}.
    \] 
\end{definition}

\begin{lemma}[Implied by the proof of {\cite[Theorem 13.1]{Roc72}}]\label{lem:roc72}
    Let $P \subseteq \mathbb{R}^d$ be a polytope. Let $\vec{u} \in \mathbb{R}^d$ be a vector. Then $\vec{u} \in P$ if and only if $\vec{u}\cdot  \vec{v} \leq f_{P}(\vec{v})$ for all facet normals $\vec{v}$ of $P$.
\end{lemma}

The following corollary follows directly from Lemma~\ref{lem:roc72}.
\begin{corollary}\label{coro:roc72}
    Let $P, Q \subseteq \mathbb{R}^d$ be two polytopes. Then $Q \subseteq P$ if and only if $f_Q(\vec{v}) \leq f_{P}(\vec{v})$ for all facet normals $\vec{v}$ of $P$.
\end{corollary}

The following lemma is a direct consequence of Lemma~\ref{lem:roc72} and Corollary~\ref{coro:roc72}.
\begin{lemma}\label{lem:vol-change}
    Let $P \subseteq \mathbb{R}^d$ be a polytope symmetric about the origin. Let $\varepsilon, \varepsilon' < 1$ be two positive reals. Then
    \begin{enumerate}[label = {\normalfont (\roman*)}]
        \item for any two vectors $\vec{p} \in \varepsilon P$ and $\vec{p}' \in \varepsilon' P$, we have
        \(
            \{\vec{p} + \vec{p}', \vec{p} - \vec{p}'\} \subseteq (\varepsilon + \varepsilon')P.
        \)

        \item for any polytope $Q \subseteq \varepsilon P$ and any vector $\vec{p}' \in \varepsilon' P$, we have
        \(
            Q + \vec{p}' \subseteq (\varepsilon + \varepsilon')P.
        \)

        \item for any polytope $Q$ with $Q + \varepsilon P \supseteq P$ , we have
        \(
            (1 - \varepsilon)P \subseteq Q.
        \)
    \end{enumerate}
\end{lemma}
\begin{proof}
    We first prove (i). Since $P$ is symmetric about the origin, we have that for any $\vec{p} \in P$ and any $\vec{v} \in \mathbb{R}^d$,
    \(
        |\vec{p} \cdot \vec{v}| \leq f_P(\vec{v}).
    \)
    Now consider $\vec{p} + \vec{p}'$. For any $\vec{v} \in \mathbb{R}^d$, we have
    \[
                        (\vec{p} + \vec{p}')\cdot \vec{v} \leq |\vec{p} \cdot \vec{v}| + |\vec{p}' \cdot \vec{v}| \leq (\varepsilon + \varepsilon')f_P(\vec{v}).
    \]
    By Lemma~\ref{lem:roc72}, we have $\vec{p} + \vec{p}' \in (\varepsilon + \varepsilon')P$. Similarly, we can show $\vec{p} - \vec{p}' \in (\varepsilon + \varepsilon')P$.

    (ii) follows directly by (i).

    We prove (iii).  Fix an arbitrary $\vec{v} \in \mathbb{R}^d$. Since $Q + \varepsilon P \supseteq P$, by Lemma~\ref{lem:roc72}, we have 
    \[
       f_Q(\vec{v}) + \varepsilon f_P(\vec{v}) \geq f_P(\vec{v}),
    \]
    or in equivalent form,
    \[
         (1 - \varepsilon)f_P(\vec{v}) \leq f_Q(\vec{v}).
    \]
    Again by Lemma~\ref{lem:roc72}, we have $(1 - \varepsilon)P \subseteq Q$.
\end{proof}

\section{Long Progressions in Dense Subset Sums}\label{sec:progression}

To simplify the notation, we make the following definition.
\begin{definition}[Density, Spreadness, Nondegeneracy for Sets]\label{def:dense-and-spread}
    Let $A \subseteq [N_1]\times \cdots \times [N_d]$ be a set of vectors. Let $\Phi := \prod_{j=1}^d N_j$. Let $\tau := \sqrt{\frac{\Phi}{\det(Z)}}$, where $Z := \mathcal{L}(A)$. We say that $A$ is
    \begin{itemize}
        \item $\rho$-dense if $|A| \geq \rho \tau$;

        \item $\gamma$-spread if $|A \setminus \Lambda| \geq \gamma\tau$ for any lattice $\Lambda \subsetneq Z$;

        \item $\delta$-nondegenerate if $|A \setminus \mathcal{V}| \geq \delta\tau$ for any lower-dimensional subspace $\mathcal{V} \subseteq \mathbb{R}^d$.
    \end{itemize} 
\end{definition}
Note that spreadness is stronger than nondegeneracy in the sense that $A$ must be $\gamma$-nondegenerate if $A$ is $\gamma$-spread, but not vice versa.

In this section, we consider only the case where $A$ is a subset of $[N]^d$ for simplicity. (Later, we will discuss how to extend the result to a more general case in Section~\ref{sec:extension}.) We shall prove that if $A$ is $\widetilde{O}(1)$-dense, $\widetilde{O}(1)$-spread, and $\widetilde{O}(1)$-nondegenerate, then $\mathcal{S}(A)$ contains a long progression whose basis is $\det(Z)\vec{e}_1, \ldots, \det(Z)\vec{e}_d$, where $Z: = \mathcal{L}(A)$. To ease the understanding, one may also assume that $Z = \mathbb{Z}^d$. 

The proof consists of three steps. In Subsection~\ref{subsec:ap-ind-basis}, we will show that if $A$ is $\widetilde{O}(1)$-nondegenerate, then $\mathcal{S}(A)$ contains a long progression whose basis consists of $d$ linearly independent vectors $\vec{b}_1, \ldots, \vec{b}_d$.  In Subsection~\ref{subsec:all-latt-pts}, we shall fill the ``gaps'' of the progression by generating all remainders (modulo $\Gamma$) of $Z$, where $\Gamma := \mathcal{L}(\vec{b}_1, \ldots, \vec{b}_d)$. As a result, we will show that $\mathcal{S}(A)$ contains all the points of $Z$ inside a large hypercube. In Subsection~\ref{subsec:ap-with-std-basis}, we will prove the main result of this section.

\subsection{Long Progressions with Independent Basis}\label{subsec:ap-ind-basis}
We need the following technical lemma, which is due to a result of Erd{\H{o}}s and S{\'a}rk{\"o}zy~\cite{ES92}. We defer its proof to Section~\ref{sec:equal-sum}.

\begin{restatable}{lemma}{lemhdespetals}\label{lem:hd-es-petals}
    Let $A$ be a multiset of vectors from $[N]^d$, where $N$ is sufficiently large. Then $A$ has at least $|A|/(d\log N)^{3d}$ disjoint subsets $A_1, \ldots, A_m$ such that $1 \leq |A_1| = \cdots = |A_m| \leq (d\log N)^{2d}$ and $\sigma(A_1) = \cdots = \sigma(A_m)$.
\end{restatable}

Basically, Lemma~\ref{lem:hd-es-petals} states that there is a small vector $\vec{b}$ such that $A$ has many disjoint subsets whose sum is $\vec{b}$. This implies that $\mathcal{S}(A)$ contains many multiples of $\vec{b}$. We shall iteratively use Lemma~\ref{lem:hd-es-petals} to construct $\vec{b}_1, \ldots, \vec{b}_d$ so that $A$ has many disjoint subsets whose sum is $\vec{b}_j$ for each $\vec{b}_j$. In the $j$-th iteration, we shall apply Lemma~\ref{lem:hd-es-petals} to the vectors in $A$ that do not fall into the subspace generated by $\vec{b}_1, \ldots, \vec{b}_{j-1}$. Since $A$ is $\widetilde{O}(1)$-nondegenerate, there must be many such vectors, and therefore, there must be many disjoint subsets whose sum is $\vec{b}_j$.

\begin{lemma}\label{lem:popular-sum}
    Let $A \subseteq [N]^d$ be a set of $n$ vectors that is $\delta$-nondegenerate. Let $\tau := \sqrt{N^d/\det(Z)}$, where $Z := \mathcal{L}(A)$. Assume that $N$ is sufficiently large and that $\delta\tau \geq 2d$.  Then there exist $d$ linearly independent vectors $\vec{b}_1, \ldots, \vec{b}_d$ such that the following properties hold.
    \begin{enumerate}[label={\normalfont (\roman*)}]
        \item $A$ contains disjoint subsets $\{A_{j,k} \subseteq A: 1 \leq j \leq d, 1\leq k \leq m\}$, where $m \geq \frac{\delta\tau}{(d\log N)^{4d}}$ and $1\leq |A_{j,k}| \leq (d\log N)^{2d}$ and 
        \(
            \sigma(A_{j,k}) = \vec{b}_j
        \)
        for all $j,k$.

        \item $|A \cap \{x_1\vec{b}_1 + \cdots + x_d\vec{b}_d:  |x_j| \leq 2^{d-1}\}| \geq n - \frac{\delta\tau}{2}$.

        \item $\frac{n}{2^{3d^2}} \leq \frac{\det(\Gamma)}{\det(Z)} \leq (d\log N)^{3d^2}\tau^2$, where $\Gamma := \mathcal{L}(\vec{b}_1, \ldots, \vec{b}_d)$.
    \end{enumerate}
\end{lemma}
\begin{proof}
     We shall iteratively generate $\vec{b}_j$ for $j = 1, \ldots, d$ using Lemma~\ref{lem:hd-es-petals}. Initially, $j := 1$ and $\vec{o}_1 := \vec{e}_1$ and $R_1 := A$. Basically, $\vec{o}_j$ is a unit-length vector that is orthogonal to $\vec{b}_1, \ldots, \vec{b}_{j-1}$, and $R_j$ is a subset of the elements in $A$ that can be used to generate $\vec{b}_j$. In iteration $j$, we select $\lfloor \frac{\delta\tau}{2d}\rfloor$ vectors from $R_j$ that have the longest projection onto $\vec{o}_j$, and denote them by $R^*_j$. (Recall that $\delta\tau \geq 2d$ and that $N$ is sufficiently large.) At least half of the vectors in $R^*_j$ have the same direction when projected onto $\vec{o}_j$, and applying Lemma~\ref{lem:hd-es-petals} to them yields at least 
     \[
         m \geq \frac{\delta\tau}{4d} \cdot \frac{1}{(d\log N)^{3d}} \geq \frac{\delta\tau}{(d\log N)^{4d}}
     \]
     non-empty disjoint subsets $A_{j,1}, \ldots, A_{j, m}$ of cardinality at most $(d\log N)^{2d}$ and $\sigma(A_{j,1}) = \ldots = \sigma(A_{j, m})$. Let $\vec{b}_j := \sigma(A_{j,1})$. If $j = d$, then we stop. Otherwise, let $\vec{o}_{j+1}$ be a unit-length vector that is orthogonal to  $\vec{b}_1, \ldots, \vec{b}_{j}$, let $R_{j+1} := R_j \setminus R^*_j$, and then proceed with the next iteration.

     We first show that $\vec{b}_1, \ldots, \vec{b}_d$ are linearly independent. Consider an arbitrary iteration $j$. 
     Since we remove at most $\frac{\delta\tau}{2d}$ elements of $A$ in each iteration, we have $|R_j| \geq |A| - \frac{\delta\tau}{2}$. Recall that $A$ is $\delta$-nondegenerate. $A$ must have at least $\delta\tau$ elements that are not orthogonal to $\vec{o}_j$, and therefore, $R_j$ has at least $\delta\tau/2$ elements that are not orthogonal to $\vec{o}_j$.  By the construction of $R^*_j$, all vectors in $R^*_j$ must have non-zero projection onto $\vec{o}_j$. This implies that the projections of vectors in $A_{j,1}$ onto $\vec{o}_j$ are either all positive or all negative. Therefore, $\vec{b}_j \cdot \vec{o}_j \neq 0$. On the other hand, we have that $\vec{o}_j \cdot \vec{b}_{j'} = 0$ for $j' \in \{1, \ldots, j-1\}$. It follows that $\vec{b}_j$ cannot be a linear combination of $\vec{b}_1, \ldots, \vec{b}_{j-1}$. This holds for all $j \in \{1, \ldots, d\}$, so $\vec{b}_1, \ldots, \vec{b}_d$ are linearly independent.

     Property (i) follows directly. Next we prove property (ii). Let $R_{d+1} := R_d\setminus R^*_d$. Since $|R^*_j| \leq \frac{\delta\tau}{2d}$ for all $j$, we have $|R_{d+1}| \geq n - \frac{\delta\tau}{2}$. To prove property (ii), it suffices to show that 
     \begin{equation}\label{eq:hd-ap-3}
        R_{d+1} \subseteq \{x_1\vec{b}_1 + \cdots + x_d\vec{b}_d: |x_j| \leq 2^{d-1}\}.
    \end{equation}
    For $j \in \{1, \ldots, d\}$, let $n_j = |A_{j,1}|$. That is, $\vec{b}_j$ is the sum of $n_j$ vectors in $A$.
    \begin{claim}\label{clm:hd-ap-2}
        For $j \in \{1, \ldots, d\}$, for any $\vec{a} \in R_{j+1}$, we have that 
        \(
            |\vec{b}_j \cdot \vec{o}_j| \geq n_j \cdot |\vec{a} \cdot \vec{o}_j|. 
        \)
    \end{claim}
    \begin{proof}
        Recall that $R_{j+1} = R_j\setminus R^*_j$. By the construction of $R^*_j$, for any $\vec{a}^* \in R^*_j$ and for any $\vec{a} \in R_{j+1}$, we have
        \(
             |\vec{a}^* \cdot \vec{o}_j| \geq |\vec{a} \cdot \vec{o}_j|. 
        \) 
        Since $\vec{b}_j$ is the sum of $n_j$ vectors from $R^*_j$, we have that for any $\vec{a} \in R_{j+1}$,
        \[
            |\vec{b}_j \cdot \vec{o}_j| \geq n_j\cdot |\vec{a} \cdot \vec{o}_j| \qedhere
        \]
    \end{proof}
    Now we are ready to prove~\eqref{eq:hd-ap-3}. Let $\vec{a}$ be an arbitrary vector in $R_{d+1}$. Since $\vec{b}_1, \ldots, \vec{b}_d$ are linearly independent, $\vec{a}$ can be represented as a linear combination of them. That is, for some $x_1, \ldots, x_d$,
    \[
        \vec{a} = x_1\vec{b}_1 + \cdots + x_d \vec{b}_d.
    \]
    We shall show that $ |x_j| \leq \frac{2^{(d-j)}}{n_j}$ for all $j \in \{1, \ldots, d\}$, which immediately implies~\eqref{eq:hd-ap-3}. When $j = d$, the inequality holds since
    \[
        \frac{1}{n_d} \cdot |\vec{b}_d \cdot \vec{o}_d| \geq |\vec{a} \cdot \vec{o}_d| = |x_d| |\vec{b}_d \cdot \vec{o}_d|. 
    \]
    The first inequality is due to Claim~\ref{clm:hd-ap-2} and the second is due to our choice of $\vec{o}_j$.
    Now suppose that the inequality holds for $x_{j+1}, \ldots , x_d$. We show that it also holds for $x_j$.  Consider $\vec{a} \cdot \vec{o}_{j}$. By the choice of $\vec{o}_j$, we have
    \begin{equation}\label{eq:hd-ap-4}
        \vec{a} \cdot \vec{o}_{j} = x_j \vec{b}_j \cdot \vec{o}_j + x_{j+1} \vec{b}_{j+1} \cdot \vec{o}_j + \cdots + x_d\vec{b}_d \cdot \vec{o}_j
    \end{equation}
    For $j' > j$, we have $\vec{b}_{j'}$ is the sum of $n_{j'}$ vectors in $R^*_{j'} \subsetneq R_{j'} \subsetneq R_{j+1}$. By Claim~\ref{clm:hd-ap-2},
    \[
        |\vec{b}_{j'} \cdot \vec{o}_j| \leq \frac{n_{j'}}{n_j} \cdot |\vec{b}_j \cdot \vec{o}_j|.
    \]
    By the inductive hypothesis,
    \[
        |x_{j'} \vec{b}_{j'} \cdot \vec{o}_j| \leq \frac{2^{(d-{j'})}}{n_{j'}} \cdot \frac{n_{j'}}{n_j} \cdot |\vec{b}_j \cdot \vec{o}_j| = \frac{2^{(d-{j'})}}{n_j}\cdot |\vec{b}_j \cdot \vec{o}_j|.
    \]
    In view of~\eqref{eq:hd-ap-4}, we have
    \begin{equation}\label{eq:hd-ap-5}
        |\vec{a} \cdot \vec{o}_{j}| \geq \left(|x_j| - \frac{2^{d - j} - 1}{n_j}\right) \cdot  |\vec{b}_j \cdot \vec{o}_j| 
    \end{equation}
    Since $\vec{a} \in R_{d+1} \subseteq R_{j+1}$, Claim~\ref{clm:hd-ap-2} implies that $|\vec{a} \cdot \vec{o}_{j}| \leq \frac{1}{n_j} \cdot |\vec{b}_j \cdot \vec{o}_j|$. In view of~\eqref{eq:hd-ap-5}, we have
    \[
         |x_j| \leq \frac{2^{d-j}}{n_j}. 
    \]

    It remains to prove property (iii). By property (i), we have $\vec{b}_j \in [(d\log N)^{2d}N]^d$ for all $j$, which implies
    \[
        \frac{\det(\Gamma)}{\det(Z)}  \leq \frac{(d(d\log N)^{2d}N)^d}{\det(Z)} \leq (d\log N)^{3d^2}\tau^2.
    \]
    Let 
    \[
        A' = A \cap \{x_1\vec{b}_1 + \cdots + x_d\vec{b}_d: |x_j| \leq 2^{d-1}\}.
    \]
    Property (ii) implies that
    \[
        \mu: = \mu(A', \Gamma) \leq (2^d + 1)^d.
    \]
    Then we have
    \[
        \frac{\det(\Gamma)}{\det(Z)} \geq |A' \bmod \Gamma| \geq \frac{|A'|}{\mu} \geq \frac{n - \delta\tau/2}{(2^d + 1)^d} = \frac{n}{2(2^d + 1)^d} \geq \frac{n}{2^{3d^2}}. \qedhere
    \]
\end{proof}

\subsection{Large Inradius of Zonotopes}
The above lemma also implies that if $A$ is $\widetilde{O}(1)$-dense and $\widetilde{O}(1)$-nondegenerate, then $\mathcal{P}(A)$ has a large inradius. This corollary will be useful in Sections~\ref{sec:john-type} and~\ref{sec:alg}.

\begin{lemma}\label{lem:inradius}
    Let $\vec{b}_1, \ldots, \vec{b}_d \in [N]^d$ be $d$ linearly independent vectors, and let $\Gamma$ be the lattice generated by them. Then 
    \[
        \left\{x_1\vec{e}_1 + \cdots + x_d\vec{e}_d : |x_j| \leq \frac{\det(\Gamma)}{d^{d}N^{d-1}}\right\} \subseteq \left\{x_1\vec{b}_1 + \cdots + x_d\vec{b}_d: |x_j| \leq 1\right\} 
    \]
\end{lemma}
\begin{proof}
    Let $P := \{x_1\vec{b}_1 + \cdots + x_d\vec{b}_d: |x_j| \leq 1\}$. We have 
    \[
        \mathrm{vol}(P) = 2^d\det(\Gamma).
    \]
    Note that $P$ is a parallelepiped, and its facets are defined by the following formulas. For each $k\in[d]$,  
    \begin{align*}
        F^+_k  &:= \{\text{$x_1\vec{b}_1 + \cdots + x_d\vec{b}_d$ : $x_k = 1$ and $|x_{j}| \leq 1$ for $j \neq k$}\},\\
        F^-_k  &:= \{\text{$x_1\vec{b}_1 + \cdots + x_d\vec{b}_d$ : $x_k = -1$ and $|x_{j}|\leq 1$ for $j \neq k$}\}.
    \end{align*} 
    Note that $F^+_k$ and $F^-_k$ are parallel, and have $(d-1)$-dimensional volume at most $(2\sqrt{d}N)^{d-1}$ (since $\|\vec{b}_i\|_2 \leq \sqrt{d}N$ for all $i$). Let $h_k$ be the distance between them. We have
    \[
        h_k \ge \frac{\mathrm{vol}(P)}{(2\sqrt{d}N)^{d-1}}
        = \frac{2\det(\Gamma)}{d^{(d-1)/2}N^{d-1}} \geq 2\sqrt{d} \cdot \frac{\det(\Gamma)}{d^{d}N^{d-1}}.
    \]
    Therefore, $P$ contains a ball centered at the origin of radius at least $\sqrt{d} \cdot \frac{\det(\Gamma)}{d^{d}N^{d-1}}$, which implies
    \[
        \left\{x_1\vec{e}_1 + \cdots + x_d\vec{e}_d : |x_j| \leq \frac{\det(\Gamma)}{d^{d}N^{d-1}}\right\}\subseteq P. \qedhere
    \]
\end{proof}

\begin{corollary}\label{coro:large-inradius}
    Let $A \subseteq [N]^d$ be a set of $n$ vectors that is $\rho$-dense and $\delta$-nondegenerate. Let $\tau := \sqrt{N^d/\det(Z)}$, where $Z := \mathcal{L}(A)$. Assume that $N$ is sufficiently large and that $\delta\tau \geq 2d$. Then we have that
    \[
        \mathcal{P}(A) \supseteq \left\{x_1\vec{e}_1 + \cdots + x_d \vec{e}_d: |x_j| \leq \frac{\rho\delta N}{(d\log N)^{7d^2}}\right\}.
    \]
\end{corollary}
\begin{proof}
    Let $\vec{b}_1, \ldots, \vec{b}_d$ be given as in Lemma~\ref{lem:popular-sum}. By Lemma~\ref{lem:popular-sum}(i), we have that
    \[
        \mathcal{S}(A) \supseteq \{z_1\vec{b}_1 + \cdots + z_d \vec{b}_d: 0 \leq z_j \leq m, z_j \in \mathbb{Z}\} \quad\text{and}\quad m \geq \frac{\delta\tau}{(d\log N)^{4d}}.
    \]
    Since $\mathcal{S}(A) \subseteq 2\mathcal{P}(A)$ and $\mathcal{P}(A)$ is symmetric about the origin, we have
    \begin{equation}\label{eq:large-inradius}
        \mathcal{P}(A) \supseteq \{x_1\vec{b}_1 + \cdots + x_d \vec{b}_d: |x_j| \leq \frac{\delta\tau}{2(d\log N)^{4d}}\}.
    \end{equation}
    By Lemma~\ref{lem:popular-sum}(i), $\vec{b}_j \in [(d\log N)^{2d}N]^d$ for all $j$ . By Lemma~\ref{lem:inradius}, we have
    \[
        \mathcal{P}(A) \supseteq \{x_1\vec{e}_1 + \cdots + x_d \vec{e}_d: |x_j| \leq L\},
    \]
    where 
    \begin{align*}
        L &\geq \frac{\delta\tau}{2(d\log N)^{4d}} \cdot \frac{\det(\Gamma)}{d^d (d\log N)^{2d^2} N^{d-1}}\\
            &\geq \frac{\delta\tau}{2(d\log N)^{4d}} \cdot \frac{1}{d^d (d\log N)^{2d^2} N^{d-1}} \cdot \frac{\rho\tau\det(Z)}{2^{3d^2}}\\
            &\geq \frac{\rho\delta N}{(d\log N)^{7d^2}}\cdot \frac{\tau^2 \det(Z)}{N^d}\\
            & = \frac{\rho\delta N}{(d\log N)^{7d^2}}. 
    \end{align*}
    The second inequality is due to Lemma~\ref{lem:popular-sum}(iii).
\end{proof}

\subsection{All Lattice Points in a Large Hypercube}\label{subsec:all-latt-pts}
Let $\vec{b}_1, \ldots, \vec{b}_d$ be the vectors given by Lemma~\ref{lem:popular-sum}. Let $\Gamma := \mathcal{L}(\vec{b}_1, \ldots, \vec{b}_d)$. Clearly, $\Gamma \subseteq Z$, where $Z: = \mathcal{L}(A)$, since each $\vec{b}_j$ is a subset sum of $A$. We already have a progression with basis $\vec{b}_1, \ldots, \vec{b}_d$. We shall fill the progression with remainders (modulo $\Gamma$) of $Z$, and show that $\mathcal{S}(A)$ contains all the lattice points of $Z$ inside a large hypercube.

To generate all these remainders, we need the following technical lemma, which states that we need only $O(\frac{\ell}{n})$ vectors to generate all the remainders of $Z$, where $\ell: = \det(\Gamma)/\det(Z)$. The lemma is due to Kneser's Theorem and a densification technique by Szemer{\'e}di and Vu~\cite{SV06a}, and we defer its proof to Section~\ref{sec:remainder}.

Recall that $\mu(A, \Gamma)$ is the largest multiplicity of any remainder (modulo $\Gamma$) in $A$.

\begin{restatable}{lemma}{lemgenremdense}\label{lem:gen-rem-dense}
    There exists a constant $c$ such that the following holds. Let $A$ be a set of $n$ vectors from $\mathbb{Z}^d$, where $Z : = \mathcal{L}(A)$ is full-dimensional. Let $\Gamma \subseteq Z$ be a full-dimensional lattice, and let $\ell := \det(\Gamma)/\det(Z)$. Let $A'$ be a subset of $A$, and let $\mu := \mu(A', \Gamma)$. Assume that 
    \begin{equation}\label{eq:gen-rem-dense-0}
        32\mu\ell(\log (\mu\ell))^{3d} \leq |A'|^2.
    \end{equation}
    If $|A \setminus \Lambda| \geq \frac{c\mu\ell}{|A'|}$ for any lattice $\Lambda$ with $\Gamma \subseteq \Lambda \subsetneq Z$, then $A$ contains a subset $R$ of at most $\frac{c\mu\ell}{|A'|}(\log (\mu\ell))^{4d}$ vectors such that $\mathcal{S}_{\Gamma}(R) = Z \bmod \Gamma$.
\end{restatable}

Compared with Lemma~\ref{lem:popular-sum}, Lemma~\ref{lem:all-lattice-pts} additionally requires $A$ to be $\widetilde{O}(1)$-dense and $\widetilde{O}(1)$-spread. This is to ensure that Lemma~\ref{lem:gen-rem-dense} is applicable to $A$. Basically, when $A$ is $\widetilde{O}(1)$-dense, Inequality~\eqref{eq:gen-rem-dense-0} will hold.  When $A$ is $\widetilde{O}(1)$-spread, $A$ will have many vectors not in $\Lambda$ for any $\Lambda \subsetneq Z$.

\begin{lemma}\label{lem:all-lattice-pts}
    Let $A \subseteq [N]^d$ be a set of vectors that is $\rho$-dense, $\gamma$-spread, and $\delta$-nondegenerate. Assume that $N$ is sufficiently large and that
    \begin{equation}\label{eq:all-lattice-pts-0}
        \rho \geq (d\log N)^{4d^2}\quad \text{and} \quad \rho\gamma \geq (d\log N)^{4d^2} \quad \text{and} \quad \rho\delta \geq (d\log N)^{13d^2}.
    \end{equation} 
    Then
    \[
        \mathcal{S}(A) \supseteq Z \cap (\vec{v}_0 +  \{z_1\vec{e}_1 + \cdots + z_d\vec{e}_d : 0 \leq z_j \leq \frac{\rho\delta N}{(d\log N)^{7d^2}}, z_j\in \mathbb{Z}\}),
    \]
    where $Z := \mathcal{L}(A)$ and $\vec{v}_0 \in Z$.
\end{lemma}
\begin{proof}
    Let $\tau: = \sqrt{N^d/\det(Z)}$. The hypercube $[N]^d$ can contain at most $2^d\tau^2$ vectors from $Z$. Therefore,
    \(
        \rho\tau \leq 2^d\tau^2,
    \)
    or equivalently, 
    \begin{equation}\label{eq:all-lattice-pts-rou-tau}
        \rho \leq 2^d\tau.
    \end{equation}
    Since $\rho\delta \geq (d\log N)^{12d^2}$ and $N$ is sufficiently large, we have
    \begin{equation}\label{eq:all-lattice-pts-bound-gamma-tau}
        \delta\tau \geq \frac{\rho\delta}{2^d} \geq \frac{(d\log N)^{12d^2}}{2^d} \geq (d\log N)^{11d^2}.
    \end{equation}
    So Lemma~\ref{lem:popular-sum} is applicable to $A$. Let $\vec{b}_1, \ldots, \vec{b}_d$ be the $d$ linearly independent vectors given by Lemma~\ref{lem:popular-sum}. We have that $\vec{b}_j \in [(d\log N)^{2d}N]^d$ for all $j$.
    Let $\Gamma := \mathcal{L}(\vec{b}_1, \ldots, \vec{b}_d)$. Clearly, $\Gamma \subseteq Z$ since each $\vec{b}_j$ is a subset sum of $A$. Let 
    \begin{equation}\label{eq:all-lattice-pts-4}
        A' = A \cap \{x_1\vec{b}_1 + \cdots + x_d\vec{b}_d: |x_j| \leq 2^{d-1}\}.
    \end{equation}
    Lemma~\ref{lem:popular-sum}(ii) guarantees that 
    \begin{equation}\label{eq:all-lattice-pts-bound-mu}
        |A'| \geq \frac{\rho\tau}{2}\quad\text{and}\quad \mu: = \mu(A', \Gamma) \leq (2^d + 1)^d.
    \end{equation}
    Let $\ell := \det(\Gamma)/\det(Z)$. Lemma~\ref{lem:popular-sum}(iii) guarantees that
    \begin{equation}\label{eq:all-lattice-pts-bound-ell}
        \frac{\rho\tau}{2^{3d^2}} \leq \ell \leq (d\log N)^{3d^2}\tau^2. 
    \end{equation}

    \begin{claim}\label{clm:all-latt-pts-1}
        $A$ contains a subset $R$ of at most $\frac{\tau}{8\rho}(d\log N)^{9d^2}$ vectors such that $\mathcal{S}_{\Gamma}(R) = Z \bmod \Gamma$ and
        \begin{equation}\label{eq:all-lattice-pts-3}
            \mathcal{P}(R) \subseteq \{x_1\vec{b}_1 + \cdots + x_d\vec{b}_d: |x_j| \leq \frac{\tau}{8\rho}(d\log N)^{9d^2}\}.
        \end{equation}
    \end{claim}
    \begin{proof}
        We shall apply Lemma~\ref{lem:gen-rem-dense} to $A$, $A'$ and $\Gamma$. Before that, we need to show that Lemma~\ref{lem:gen-rem-dense} is applicable. Let $c$ be the constant in Lemma~\ref{lem:gen-rem-dense}. Recall that $N$ is sufficiently large. In view of~\eqref{eq:all-lattice-pts-bound-mu} and~\eqref{eq:all-lattice-pts-bound-ell}, one can verify that
        \begin{equation}\label{eq:all-lattice-pts-2}
            c\mu\ell (\log (\mu\ell))^{4d} \leq \ell(d\log N)^{6d^2}.
        \end{equation}
        Then 
        \[
            |A'|^2 \geq  \frac{\rho^2\tau^2}{4} \geq \frac{(d\log N)^{8d^2}\tau^2}{4} \geq c\mu\ell (\log (\mu\ell))^{3d}. 
        \]
        The second inequality is due to~\eqref{eq:all-lattice-pts-0}. Since $A$ is $\gamma$-spread, for any lattice $\Lambda \subsetneq Z$, we have
        \[
            |A \setminus \Lambda| \geq \gamma\tau \geq \frac{(d\log N)^{4d^2}\tau}{\rho} 
             \geq \frac{(d\log N)^{4d^2}\tau^2}{n} \geq \frac{c\mu\ell}{n}.
        \]
        The second inequality is due to~\eqref{eq:all-lattice-pts-0}, and the last is due to~\eqref{eq:all-lattice-pts-bound-ell}.  

        Therefore, Lemma~\ref{lem:gen-rem-dense} is applicable, and $A$ contains a subset $R$ with $\mathcal{S}_{\Gamma}(R) = Z \bmod \Gamma$ and
        \[
            |R| \leq \frac{c\mu\ell (\log \mu\ell)^{4d}}{|A'|} \leq \frac{2\ell}{\rho\tau}(d\log N)^{6d^2} \leq \frac{\tau}{8\rho}(d\log N)^{9d^2}.
        \]
        The second inequality follows from $|A'| \geq \frac{\rho\tau}{2}$ and~\eqref{eq:all-lattice-pts-2}, and the last follows from~\eqref{eq:all-lattice-pts-bound-ell}. Since $\vec{a} \in [N]^d$ for every $\vec{a} \in R$, we have
        \[
            \mathcal{P}(R) \subseteq \{x_1\vec{e}_1 + \cdots + x_d\vec{e}_d: |x_j| \leq \frac{\ell N}{\rho\tau}(d\log N)^{6d^2}\}.
        \]
        Recall that $\vec{b}_j \in [(d\log N)^{2d}N]^d$ for all $j$. By Lemma~\ref{lem:inradius}, we have that
        \[
            \mathcal{P}(R) \subseteq \{x_1\vec{b}_1 + \cdots + x_d\vec{b}_d: |x_j| \leq L\},
        \]
        where 
        \[
            L \leq \frac{\ell N}{\rho\tau}(d\log N)^{6d^2} \cdot \frac{d^d (d\log N)^{2d^2}N^{d-1}}{\det(\Gamma)} = \frac{d^d(d\log N)^{8d^2}}{\rho} \cdot \frac{\ell N^d}{\tau \det(\Gamma)} \leq \frac{ (d\log N)^{9d^2}}{8\rho} \cdot \tau.
        \]
        The last inequality is due to the definition of $\ell$ and $\tau$.
    \end{proof}

    Let $m$ be defined as in Lemma~\ref{lem:popular-sum}(i). We have
    \begin{equation}\label{eq:all-lattice-pts-bound-m}
        m \geq \frac{\delta\tau}{(d\log N)^{4d}} \geq \frac{\tau}{\rho}(d\log N)^{9d^2} \geq 4.
    \end{equation}
    The second inequality is due to~\eqref{eq:all-lattice-pts-0} and the last is due to~\eqref{eq:all-lattice-pts-rou-tau}.
    Lemma~\ref{lem:popular-sum}(i) also gives $md$ disjoint subsets $\{A_{j,k} : 1 \leq j \leq d, 1\leq k \leq m\}$ of $A$. Let $G := \bigcup_{j,k} A_{j,k}$ be the union of these subsets.

    \begin{claim}\label{clm:all-latt-pts-2}
        Let $\vec{b}_0 := \frac{m}{4}(\vec{b}_1 + \cdots + \vec{b}_d)$. We have
        \[
            \mathcal{S}(G \cup R) \supseteq Z \cap (\vec{b}_0 + \{x_1\vec{b}_1 + \cdots + x_d\vec{b}_d: 0\leq x_j \leq \frac{m}{4}\}).
        \]
    \end{claim}
    \begin{proof}
        Take an arbitrary $\vec{v}$ in $Z \cap (\vec{b}_0 + \{x_1\vec{b}_1 + \cdots + x_d\vec{b}_d: 0\leq x_j \leq \frac{m}{4}\})$. We shall show that $\vec{v} \in \mathcal{S}(G \cup R)$. Recall that $\mathcal{S}_{\Gamma}(R) = Z \bmod \Gamma$, where $\Gamma := \mathcal{L}(\vec{b}_1, \ldots, \vec{b}_d)$. So there must be a vector $\vec{u} \in \mathcal{S}(R)$ such that $\vec{v} - \vec{u} \in \Gamma$. To show that $\vec{v} \in \mathcal{S}(G \cup R)$, it suffices to show that $\vec{v} - \vec{u} \in \mathcal{S}(G \setminus R)$.

        Consider $G \setminus R$. Removing $R$ from $G$ affects at most $|R|$ disjoint subsets stated in Lemma~\ref{lem:popular-sum}(i); for each $\vec{b}_j$, there are still $m - |R|$ disjoint subsets of $G\setminus R$ whose sum is $\vec{b}_j$. In view of~\eqref{eq:all-lattice-pts-bound-m}, we have $m > 4|R|$. Therefore, we have that
        \[
            \mathcal{S}(G \setminus R) \supseteq \{z_1\vec{b}_1 + \cdots + z_d\vec{b}_d: 0\leq z_j \leq \frac{3m}{4}, z_j \in \mathbb{Z}\}.
        \]
        Since we already have $\vec{v} - \vec{u} \in \Gamma$, to show that $\vec{v} - \vec{u} \in \mathcal{S}(A \setminus R)$, it suffices to show that
        \begin{equation}\label{eq:all-lattice-pts-5}
            \vec{v} - \vec{u} \in \{x_1\vec{b}_1 + \cdots + x_d\vec{b}_d: 0\leq x_j \leq \frac{3m}{4}\}.
        \end{equation}

        Since $\vec{u} \in \mathcal{S}(R)$ and $\mathcal{S}(R) \subseteq 2\mathcal{P}(R)$, in view of~\eqref{eq:all-lattice-pts-3} and~\eqref{eq:all-lattice-pts-bound-m}, we have that
        \[
            \vec{u} \in \{x_1\vec{b}_1 + \cdots + x_d\vec{b}_d:  |x_j| \leq \frac{m}{4}\}.
        \]
        Note that 
        \[
            \vec{v} \in \{x_1\vec{b}_1 + \cdots + x_d\vec{b}_d:  \frac{m}{4} \leq x_j \leq \frac{m}{2}\}.
        \]
        Therefore,
        \[
            \vec{v} - \vec{u} \in \{x_1\vec{b}_1 + \cdots + x_d\vec{b}_d: 0 \leq x_j \leq \frac{3m}{4}\}. \qedhere
        \]
    \end{proof}

    In view of Claim~\ref{clm:all-latt-pts-2}, to prove the theorem, it suffices to show that 
    \[
        P := \{x_1\vec{b}_1 + \cdots + x_d\vec{b}_d: 0\leq x_j \leq \frac{m}{4}\}
    \]
    contains a hypercube of side length at least $\frac{\rho\delta N}{(d\log N)^{7d^2}} + N$. This can be shown using exactly the same argument as that in Corollary~\ref{coro:large-inradius}. We omit the details.
\end{proof}

\subsection{Long Progressions with Axis-Parallel Basis}\label{subsec:ap-with-std-basis}
We prove the main theorem of this Section. We first show that all the lattice points of $Z$ in a large hypercube must contain a long progression with basis $\det(Z) \vec{e}_1, \ldots, \det(Z) \vec{e}_d$. 

\begin{lemma}\label{lem:standard-basis-in-large-area}
    Let $Z$ be a full-dimensional lattice. Let 
    \[
        P := \{x_1\vec{e}_1 + \cdots + x_d\vec{e}_d : 0 \leq x_j \leq N_j\det(Z)\},
    \]
    where $N_1, \ldots, N_d$ are positive integers.
    Then 
    \[
        P \cap Z \supseteq \{z_1\det(Z)\vec{e}_1 + \cdots + z_d\det(Z)\vec{e}_d : 0 \leq z_j \leq N_j, z_j \in \mathbb{Z}\}.
    \]
\end{lemma}
\begin{proof}
    Let $B = (\vec{b}_1, \ldots, \vec{b}_d)$ be a basis of $Z$. Without loss of generality, we assume that $\det(B) > 0$, since otherwise, we can pick an arbitrary column $\vec{b}_j$ and set $\vec{b}_j:= -\vec{b}_j$. Let $B^*$ be the adjugate matrix of $B$; then $B^*$ is also an integer matrix, and $BB^*=\det(B) I$, where $I$ is the identity matrix. Let $\vec{q}_j$ be the $j$-th column of $B^*$; then
    \[
        B\vec{q}_j =\det(B)\vec{e}_j.
    \]
    Note that $B\vec{q}_j = q_{1j}\vec{b}_1 + \cdots + q_{dj}\vec{b}_d$. Therefore, $\det(B)\vec{e}_j \in Z$, which implies $z_j\det(B)\vec{e}_j \in P \cap Z$ for any integer $z_j$ with $0\leq z_j \leq N_j$. Since this holds for all $j$, we have 
    \[
        P \cap Z \supseteq \{z_1\det(Z)\vec{e}_1 + \cdots + z_d\det(Z)\vec{e}_d : 0 \leq z_j \leq N_j\}. \qedhere
    \]
\end{proof}

The main theorem of this section follows directly from Lemmas~\ref{lem:all-lattice-pts} and~\ref{lem:standard-basis-in-large-area}.
\begin{theorem}\label{thm:ap-standard-basis}
    Let $A \subseteq [N]^d$ be a set of vectors that is $\rho$-dense, $\gamma$-spread, and $\delta$-nondegenerate. Assume that $N$ is sufficiently large and that
    \[
        \rho \geq (d\log N)^{4d^2}\quad \text{and} \quad \rho\gamma \geq (d\log N)^{4d^2} \quad \text{and} \quad \rho\delta \geq (d\log N)^{13d^2}\det(Z).
    \] 
    Then
    \[
        \mathcal{S}(A) \supseteq \vec{v}_0 +  \{z_1\cdot \det(Z)\vec{e}_1 + \cdots + z_d\cdot \det(Z)\vec{e}_d : 0 \leq z_j \leq N, z_j\in \mathbb{Z}\},
    \]
    where $Z := \mathcal{L}(A)$ and $\vec{v}_0 \in Z$.
\end{theorem}

\section{A John-Type Theorem for Dense Subset Sum}\label{sec:john-type}
Again, we consider only the case where $A$ is a subset of $[N]^d$, and to ease the understanding, one may also assume that $Z = \mathbb{Z}^d$. 

Recall Definition~\ref{def:dense-and-spread}. We shall show that when $A$ is $\widetilde{O}(1)$-dense, $\widetilde{O}(1)$-spread, and $\widetilde{O}(1)$-nondegenerate, then $\mathcal{S}(A)$ contains almost all the lattice points of $\mathcal{L}(A)$ inside the zonotope. That is, 
\[
    \mathcal{S}(A) \supseteq Z \cap \left((1-o(1))\mathcal{P}(A) + \frac{\sigma(A)}{2}\right).
\]

We first show that for any lattice point in the above zonotope, $\mathcal{S}(A)$ must contain a point that is close to it. This is a direct consequence of integer programming theory.

\begin{lemma}\label{lem:proximity}
    Let $A \subseteq [N]^d$ be a set of vectors. For each $\vec{t} \in \mathcal{P}(A) + \frac{\sigma(A)}{2}$, there exists $\vec{t}' \in \mathcal{S}(A)$ such that
    \[
        \vec{t} - \vec{t}' \in \{x_1\vec{e}_1 + \cdots + x_d\vec{e}_d: 0 \leq x_j \leq dN\}.
    \]
\end{lemma}
\begin{proof}
    Let $\vec{a}_1, \ldots, \vec{a}_n$ be the vectors in $A$. Take an arbitrary  $\vec{t} \in  \mathcal{P}(A) + \frac{\sigma(A)}{2}$. Consider the polytope
    \[
        K := \{(y_1, \ldots, y_n) : y_1\vec{a}_1 + \cdots + y_n\vec{a}_n = \vec{t}, 0 \leq y_j \leq 1 \}.
    \]
    $K$ has $n$ variables and $d + 2n$ constraints. Also $K$ is non-empty since $\vec{t} \in \mathcal{P}(A) + \frac{\sigma(A)}{2}$. Take an arbitrary vertex $(y_1, \ldots, y_n)$ of $K$. Classical integer programming theory implies that all but at most $d$ variables of $(y_1, \ldots, y_n)$ must be integral~\cite{Sch86}.  Without loss of generality, assume that only variables $y_1, \ldots, y_{d}$ may be non-integral. Let 
    \[
        \vec{t}' := y_{d+1} \vec{a}_{d+1} + \cdots + y_{n}\vec{a}_n.
    \]
    Since $y_{d+1}, \ldots, y_{n} \in \{0,1\}$, we have $\vec{t}' \in \mathcal{S}(A)$. Note that
    \(
        \vec{t} - \vec{t'} = y_1\vec{a}_1 + \cdots + y_{d}\vec{a}_{d}.  
    \)
    It is easy to verify that 
    \(
        \vec{t} - \vec{t}' \in \{x_1\vec{e}_1 + \cdots + x_d\vec{e}_d: 0 \leq x_i \leq dN\}.
    \)
\end{proof}

The above lemma implies that the ``holes'' between the elements of $\mathcal{S}(A)$ cannot be too large -- they can be covered by a hypercube of side length $dN$. Therefore, if we have a ``patch'' that contains all the lattice points inside a hypercube of side length $dN$, then we can cover the holes using the patch and obtain all the lattice points in the zonotope. 

The following lemma states that if we can use only a small subset $A^*$ of $A$ to generate the patch, then we only lose a small fraction of the  zonotope in which all the lattice points belong to $\mathcal{S}(A)$. By small, we mean $\mathcal{P}(A^*) \subseteq \varepsilon \mathcal{P}(A)$.

\begin{lemma}\label{lem:frac2int}
    Let $A \subseteq [N]^d$ be a set of vectors where $Z:= \mathcal{L}(A)$ is full-dimensional. Suppose that $A$ contains a subset $A^*$ such that $\mathcal{P}(A^*) \subseteq \varepsilon \mathcal{P}(A)$ for some $\varepsilon \leq 1/4$ and that for some vector $\vec{v}_0 \in Z$,
    \begin{equation}\label{eq:frac2int-1}
        \mathcal{S}(A^*) \supseteq Z \cap (\vec{v}_0 + \{z_1\vec{e}_1 + \cdots + z_d\vec{e}_d : z_j \in \mathbb{Z}, 0 \leq z_j \leq dN\}),
    \end{equation}
    then 
    \[
        \mathcal{S}(A)  \supseteq Z \cap \left((1-4\varepsilon)\mathcal{P}(A) + \frac{\sigma(A)}{2}\right).
    \]
\end{lemma}
\begin{proof}
    We first show that
    \begin{equation}\label{eq:frac2int-2}
        \mathcal{S}(A) \supseteq Z \cap \left(\mathcal{P}(A\setminus A^*) + \vec{v}_0 - \frac{\sigma(A^*)}{2} + \frac{\sigma(A)}{2}\right).
    \end{equation}
    Take an arbitrary $\vec{t}$ from the set on the right-hand side of~\eqref{eq:frac2int-2}. We have 
    \[
        \vec{t} - \vec{v}_0 \in \mathcal{P}(A\setminus A^*) - \frac{\sigma(A^*)}{2} + \frac{\sigma(A)}{2} 
        = \mathcal{P}(A\setminus A^*) + \frac{\sigma(A \setminus A^*)}{2}.
    \]
    By Lemma~\ref{lem:proximity}, there exists $\vec{t}' \in \mathcal{S}(A\setminus A^*)$ such that
    \[
        \vec{t} - \vec{v}_0 - \vec{t'}\in  \{x_1\vec{e}_1 + \cdots + x_d\vec{e}_d: 0 \leq x_j \leq dN\},
    \]
    or in equivalent form, 
    \[
        \vec{t}  - \vec{t'} \in \vec{v}_0 + \{x_1\vec{e}_1 + \cdots + x_d\vec{e}_d: 0 \leq x_j \leq dN\}.
    \]
    Also note that $\vec{t} - \vec{t'} \in Z$. In view of~\eqref{eq:frac2int-1}, we have $\vec{t} - \vec{t}' \in \mathcal{S}(A^*)$. Therefore, $\vec{t} \in \mathcal{S}(A)$.

    In view of~\eqref{eq:frac2int-2}, to prove the lemma, it suffices to show that 
    \begin{equation}\label{eq:frac2int-3}
        (1-4\varepsilon)\mathcal{P}(A) \subseteq \mathcal{P}(A\setminus A^*) + \vec{v}_0 - \frac{\sigma(A^*)}{2}.
    \end{equation}
    Take an arbitrary $\vec{v} \in (1-4\varepsilon)\mathcal{P}(A)$. Note that $\vec{v}_0 \in \mathcal{S}(A^*) \subseteq 2\mathcal{P}(A^*) \subseteq 2\varepsilon \mathcal{P}(A)$ and that $\frac{\sigma(A^*)}{2} \in \mathcal{P}(A^*) \subseteq \varepsilon \mathcal{P}(A)$. By Lemma~\ref{lem:vol-change}(i),
    \begin{equation}\label{eq:frac2int-4}
        \vec{v} - \vec{v}_0 + \frac{\sigma(A^*)}{2} \in (1 - \varepsilon)\mathcal{P}(A).
    \end{equation}
    Recall that $\mathcal{P}(A^*) \subseteq \varepsilon\mathcal{P}(A)$. Also note that $\mathcal{P}(A \setminus A^*) + \mathcal{P}(A^*) = \mathcal{P}(A)$. By Lemma~\ref{lem:vol-change}(iii), 
    \[
        (1 - \varepsilon)\mathcal{P}(A) \subseteq \mathcal{P}(A \setminus A^*).
    \]
    In view of~\eqref{eq:frac2int-4}, 
    \[
        \vec{v} - \vec{v}_0 + \frac{\sigma(A^*)}{2} \in \mathcal{P}(A \setminus A^*),
    \]
    or, in equivalent form,
    \[
        \vec{v} \in \mathcal{P}(A \setminus A^*) + \vec{v}_0 - \frac{\sigma(A^*)}{2}.
    \]
    Therefore,~\eqref{eq:frac2int-3} holds.
\end{proof}

Now we are very close to the main result of this section. Lemma~\ref{lem:all-lattice-pts} already implies that we can generate a patch using vectors in $A$. It remains to show that only a small subset of $A$ is needed.  

The major part of the proof of Lemma~\ref{lem:small-region-small-set} is the same as that of Lemma~\ref{lem:all-lattice-pts} -- we apply Lemma~\ref{lem:popular-sum} to $A$ and obtain a set $G$ such that $\mathcal{S}(G)$ contains a long progression with basis $\vec{b}_1, \ldots, \vec{b}_d$, and then obtain a subset $R$ of $A$ that can generate all remainders $Z \bmod \Gamma$, where $\Gamma := \mathcal{L}(\vec{b}_1, \ldots, \vec{b}_d)$. Then it follows that $\mathcal{S}(G \cup R)$ contains all the lattice points in $Z$ inside a large hypercube. 

It is actually easy to bound $\mathcal{P}(R)$ in terms of $\mathcal{P}(A)$. The main challenge is $\mathcal{P}(G)$. To bound $\mathcal{P}(G)$, we use a probabilistic argument. Instead of including all the disjoint subsets $\{A_{jk}\}_{j,k}$ returned by Lemma~\ref{lem:popular-sum}, we will sample them with probability $p \approx (d\log N)^{12d^2}/(\rho\tau)$ and include only the sampled subsets in $G$. Then in expectation, $\mathcal{P}(G) \subseteq p\cdot \mathcal{P}(A)$. We will also show that $\mathcal{S}(G \cup R)$ still contains all the lattice points in $Z$ inside a large hypercube.

\begin{lemma}\label{lem:small-region-small-set}
    Let $A \subseteq [N]^d$ be a set of vectors that is $\rho$-dense, $\gamma$-spread, and $\delta$-nondegenerate. Assume that $N$ is sufficiently large and that
    \[
        \rho \geq (d\log N)^{4d^2}\quad \text{and} \quad \rho\gamma \geq (d\log N)^{4d^2} \quad \text{and} \quad \rho\delta \geq 5(d\log N)^{13d^2}.
    \]
    Then $A$ contains a subset $A^*$ such that 
    \[
        \mathcal{S}(A^*) \supseteq Z \cap (\vec{v}_0 +  \{z_1\vec{e}_1 + \cdots + z_d\vec{e}_d : 0 \leq z_j \leq dN, z_j\in \mathbb{Z}\}),
    \]
    where $Z := \mathcal{L}(A)$ and $\vec{v}_0 \in Z$, and that
    \[
        \mathcal{P}(A^*) \subseteq \varepsilon \mathcal{P}(A),
    \]
    where $\varepsilon := \frac{5(d\log N)^{13d^2}}{\rho\delta}$.
\end{lemma}
\begin{proof}
    Let $\tau : = \sqrt{N^d/\det(Z)}$. Using exactly the same argument as in the proof of Lemma~\ref{lem:all-lattice-pts}, one can show that Lemma~\ref{lem:popular-sum} is applicable to $A$. Let $\vec{b}_1, \ldots, \vec{b}_d$ and $\{A_{j,k}: 1\leq j \leq d, 1 \leq k \leq m\}$ be the linearly independent vectors and disjoint subsets of $A$ returned by Lemma~\ref{lem:popular-sum}. We have $|A_{j,k}| \leq (d\log N)^{2d}$ and $m \geq \frac{\delta\tau}{(d\log N)^{4d}}$. Let $\ell := \det(\Gamma)/\det(Z)$. Lemma~\ref{lem:popular-sum}(iii) guarantees that
    \begin{equation}\label{eq:small-set-small-region-3}
        \frac{\rho\tau}{2^{3d^2}} \leq \ell \leq (d\log N)^{3d^2}\tau^2. 
    \end{equation}

    Let $p = \frac{2(d\log N)^{13d^2}}{\rho\delta}$. We randomly sample the subsets $\{A_{j,k}: 1\leq j \leq d, 1 \leq k \leq m\}$ with probability $p$, and let $G$ be the union of the sampled subsets. 

    \begin{claim}\label{clm:small-set-small-region-1}
        With error probability at most $N^{-d}$, $G$ contains disjoint subsets $\{G_{j,k}: 1\leq j \leq d, 1 \leq k \leq m'\}$, where $m' := \lceil \frac{\tau}{\rho}(d\log N)^{9d^2} \rceil$ and $\sigma(G_{j,k}) = \vec{b}_j$ for all $j, k$.
    \end{claim}
    \begin{proof}
        For each $j$, let $X_j$ be the number of sampled subsets from $\{A_{j,k}: 1 \leq k \leq m\}$. To prove the claim, it suffices to bound $\mathbf{Pr}[X_j \leq \frac{\tau}{\rho}(d\log N)^{9d^2}]$.

        It is easy to see that $\mathbf{E}[X_j] = pm \geq \frac{2\tau}{\rho}(d\log N)^{9d^2}$. Note that $X_j$ is a sum of independent Bernoulli trials. By standard Chernoff bound, we have
        \[
            \mathbf{Pr}[X_j \leq \frac{\tau}{\rho}(d\log N)^{9d^2}] \leq \exp(- \frac{\tau}{4\rho}(d\log N)^{9d^2}) \leq \exp(- \frac{(d\log N)^{9d^2}}{2^{d+2}}) \leq N^{-d^2}.
        \]
        The second inequality is due to~\eqref{eq:all-lattice-pts-rou-tau}. The claim follows by the union bound.
    \end{proof}

    Then using exactly the same argument (Claims~\ref{clm:all-latt-pts-1} and~\ref{clm:all-latt-pts-2}) as in Lemma~\ref{lem:all-lattice-pts}, one can show that $A$ contains a subset $R$ of at most $\frac{\tau}{8\rho}(d\log N)^{9d^2}$ vectors such that $\mathcal{S}_{\Gamma}(R) = Z \bmod \Gamma$ and
    \begin{equation}\label{eq:small-set-small-region-1}
        \mathcal{P}(R) \subseteq \{x_1\vec{b}_1 + \cdots + x_d\vec{b}_d: |x_j| \leq \frac{\tau}{8\rho}(d\log N)^{9d^2}\},
    \end{equation}
    and 
    \begin{equation}\label{eq:small-set-small-region-2}
        \mathcal{S}(G \cup R) \supseteq Z \cap (\vec{b}_0 + \{x_1\vec{b}_1 + \cdots + x_d\vec{b}_d: 0\leq x_j \leq \frac{m'}{4}\}),
    \end{equation}
    where $\vec{b}_0 := \frac{m'}{4}(\vec{b}_1 + \cdots + \vec{b}_d)$. 

    We shall show that $A^*:= G \cup R$ satisfies all the stated properties. 

    In view of~\eqref{eq:small-set-small-region-2}, we have
    \[
        \mathcal{S}(G \cup R) \supseteq Z \cap (\vec{v}_0 +  \{z_1\vec{e}_1 + \cdots + z_d\vec{e}_d : 0 \leq z_j \leq dN, z_j\in \mathbb{Z}\}),
    \]
    where $\vec{v}_0 \in Z$, since using exactly the same argument as in Corollary~\ref{coro:large-inradius}, the polytope $\{x_1\vec{b}_1 + \cdots + x_d\vec{b}_d: 0\leq x_j \leq \frac{m'}{4}\}$ contains a hypercube of side length at least $dN + N$. We omit the details.


    It remains to bound $\mathcal{P}(G \cup R)$ in terms of $\mathcal{P}(A)$.  It suffices to bound $\mathcal{P}(R)$ and $\mathcal{P}(G)$ separately.

    We first show that $\mathcal{P}(A)$ is large. Note that $\mathcal{S}(A) \subseteq 2\mathcal{P}(A)$ and $\mathcal{P}(A)$ is symmetric about the origin. Lemma~\ref{lem:popular-sum}(i) implies that
    \begin{equation}\label{eq:small-set-small-region-4}
        \mathcal{P}(A) \supseteq \left\{x_1\vec{b}_1 + \cdots + x_d \vec{b}_d: |x_j| \leq \frac{\delta\tau}{2(d\log N)^{4d}}\right\}.
    \end{equation}
    By Corollary~\ref{coro:large-inradius}, we have
    \begin{equation}\label{eq:small-set-small-region-5}
        \mathcal{P}(A) \supseteq \left\{x_1\vec{e}_1 + \cdots + x_d \vec{e}_d: |x_j| \leq \frac{\rho\delta N}{(d\log N)^{7d^2}}\right\}.
    \end{equation}

    Now we can bound $\mathcal{P}(R)$.  In view of~\eqref{eq:small-set-small-region-1} and~\eqref{eq:small-set-small-region-4}, we have that 
    \[
        \mathcal{P}(R) \subseteq \varepsilon_R\mathcal{P}(A),
    \]
    where 
    \[
        \varepsilon_R \leq \frac{(d\log N)^{13d^2}}{4\rho\delta}.
    \]

    Next we show that $\mathcal{P}(G) \subseteq 2p \mathcal{P}(A)$, which will immediately imply that
    \[
        \mathcal{P}(G \cup R) \subseteq \mathcal{P}(G) + \mathcal{P}(R) \subseteq (\varepsilon_R + 2p) \mathcal{P}(A) \subseteq \frac{5(d\log N)^{13d^2}}{\rho\delta}.
    \] By Corollary~\ref{coro:roc72}, $\mathcal{P}(G) \subseteq 2p \mathcal{P}(A)$ if and only if $f_{P'}(\vec{v}) \leq 2pf_{P}(\vec{v})$, where $P' := \mathcal{P}(G)$ and $P := \mathcal{P}(A)$, for all facet normals $\vec{v}$ of $\mathcal{P}(A)$. Below we show that this happens with error probability at most $N^{-d}$. This will finish the proof, since the total error probability of the construction of $G$ and $R$ is at most $2N^{-d} < 1$, which implies the existence of the target $G$ and $R$.

    Fix an arbitrary vector $\vec{v} \in \mathbb{R}^d$ with $\|\vec{v}\|_2 = 1$. By definition, it is easy to observe that  $f_{P'}(\vec{v}) = \frac{1}{2}\sum_{\vec{a} \in G}|\vec{a} \cdot \vec{v}|$. For simplicity of notation, we relabel the subsets $\{A_{j,k}: 1\leq j \leq d, 1 \leq k \leq m\}$ as $\{A_1, A_2, \ldots, A_{md}\}$. For each $A_{i}$, let $f_i(\vec{v}) = \frac{1}{2}\sum_{\vec{a} \in A_i}|\vec{a} \cdot \vec{v}|$. Then $f_{P'}(\vec{v})$ can be viewed as the sum of values that are sampled with probability $p$ from $\{f_1(\vec{v}), \ldots, f_{md}(\vec{v})\}$, and therefore, we have
    \begin{align*}
            \mathbf{E}[f_{P'}(\vec{v})] &= \sum_{i}p f_{i}(\vec{v}) = \frac{p}{2}\sum_{i}\sum_{\vec{a} \in A_{i}}|\vec{a} \cdot \vec{v}| \leq \frac{p}{2}\sum_{\vec{a} \in A}|\vec{a} \cdot \vec{v}| = p\cdot f_P(\vec{v})\\
            \mathbf{Var}[f_{P'}(\vec{v})] &= p(1-p)\sum_{i}f_{i}(\vec{v})^2 \leq p \cdot \max_{i}f_{i}(\vec{v}) \cdot \sum_{i}f_{i}(\vec{v})  \leq p \cdot \max_{i}f_{i}(\vec{v}) \cdot f_P(\vec{v}).
    \end{align*}
    Let $\eta = p f_P(\vec{v})$. By Bernstein's inequality, we have
        \begin{align*}
            \mathbf{Pr}[f_{P'}(\vec{v}) \geq 2p \cdot f_P(\vec{v})] 
        \leq & \mathrm{exp}\left(-\frac{\eta^2}{2\mathbf{Var}[f_{P'}(\vec{v})] + \frac{2}{3}\eta\max_{i} f_{i}(\vec{v})}\right) \nonumber\\
        \leq & \mathrm{exp}\left(- \min\left\{\frac{\eta^2}{4\mathbf{Var}[f_{P'}(\vec{v})]}, \frac{\eta}{2\max_{i} f_{i}(\vec{v})}\right\}\right)\\
        \leq &     \mathrm{exp}\left(- \min\left\{ \frac{\eta}{4\max_{i} f_{i}(\vec{v})},  \frac{\eta}{2\max_{i} f_{i}(\vec{v})}\right\}\right) \\
        \leq& \mathrm{exp}\left(-\frac{pf_P(\vec{v})}{4\max_{i} f_{i}(\vec{v})}\right).\\
        \end{align*}
    In view of~\eqref{eq:small-set-small-region-5}, we have 
    \[
        f_P(\vec{v}) \geq \frac{\rho\delta N}{(d\log N)^{7d^2}}.
    \]
    Since $|A_i| \leq (d\log N)^{2d}$ and $\vec{a} \in [N]^d$ for all $\vec{a} \in A_i$, we have
    \[
        f_i(\vec{v}) \leq \frac{1}{2} \cdot (d\log N)^{2d} \cdot \sqrt{d}{N} \leq \frac{\sqrt{d}N}{2}(d\log N)^{2d}.
    \]
    Therefore, we have
    \begin{align*}
        &\mathbf{Pr}[f_{P'}(\vec{v}) \geq 2p \cdot f_P(\vec{v})]\\
        \leq &\exp(- \frac{2(d\log N)^{12d^2}}{\rho\delta} \cdot \frac{\rho\delta N}{(d\log N)^{7d^2}} \cdot \frac{2}{\sqrt{d}(d\log N)^{2d} N}) \leq N^{-2d^2}\\
    \end{align*}
    Note that $\mathcal{P}(A)$ can have at most $|A|^d \leq N^{d^2}$ facets. Therefore, by the union bound, we have
    \[
        \mathbf{Pr}[\text{$f_{P'}(\vec{v}) \geq 2p\cdot f_{P}(\vec{v})$ for some facet normal of $\mathcal{P}(A)$} ] \leq N^{-2d^2} \cdot N^{d^2} \leq N^{-d}. \qedhere
    \]
\end{proof}

The main theorem of this section follows from Lemmas~\ref{lem:frac2int} and~\ref{lem:small-region-small-set}.

\begin{theorem}\label{thm:john-type}
    Let $A \subseteq [N]^d$ be a set of vectors that is $\rho$-dense, $\gamma$-spread, and $\delta$-nondegenerate. Assume that $N$ is sufficiently large and that
    \[
        \rho \geq (d\log N)^{4d^2}\quad \text{and} \quad \rho\gamma \geq (d\log N)^{4d^2} \quad \text{and} \quad \rho\delta \geq (d\log N)^{14d^2}.
    \]
    Then 
    \[
         Z \cap \left((1 - \varepsilon) \mathcal{P}(A) + \frac{\sigma(A)}{2}\right) \subseteq \mathcal{S}(A) \subseteq Z \cap  \left(\mathcal{P}(A) + \frac{\sigma(A)}{2}\right),
    \]
    where $Z := \mathcal{L}(A)$ and $\varepsilon := \frac{(d\log N)^{14d^2}}{\rho\delta}$.
\end{theorem}

\section{An Algorithm for Dense Subset Sum}\label{sec:alg}
Recall Definition~\ref{def:dense-and-spread}. We shall show that given a set $A \subseteq [N]^d$ that is $\widetilde{O}(1)$-dense and $\widetilde{O}(1)$-nondegenerate, for any $\vec{t} \in (1 - o(1))\mathcal{P}(A) + \frac{\sigma(A)}{2}$, we can decide whether $\vec{t} \in \mathcal{S}(A)$ in $\widetilde{O}(n \cdot d^d \cdot \mathrm{poly}(d))$ time.

Basically, we will extract a subset $A^*$ of $A$ that is $\widetilde{O}(1)$-dense, $\widetilde{O}(1)$-spread, and $\widetilde{O}(1)$-nondegenerate. Using this subset and Theorem~\ref{thm:john-type}, we will be able to obtain a lattice $\Gamma$ such that $\vec{t} \in \mathcal{S}(A)$ if and only if $\vec{t} \in \mathcal{S}_{\Gamma}(A)$. Then it suffices to solve the modular subset sum problem.

Our algorithm relies on the fact that modular sumsets and symmetry sets can be computed efficiently. The proof of the following lemma is deferred to Section~\ref{sec:modular-sumset}.

\begin{restatable}{lemma}{lemfastermodularss}\label{lem:faster-modular-ss}
    Let $A \subseteq \mathbb{Z}^d$ be a set of vectors in some full-dimensional lattice $Z$. Let $\Gamma \subseteq Z$ be a full-dimensional lattice. We can compute $\mathcal{S}_{\Gamma}(A)$ and $\mathrm{Sym}_{\Gamma}(\mathcal{S}(A))$ in $\widetilde{O}((n + 2^d\ell) \cdot \mathrm{poly}(d))$ time, where $\ell := \det(\Gamma)/\det(Z)$.
\end{restatable}

\subsection{Computing Almost Common Lattices}
To extract an $\widetilde{O}(1)$-spread subset of $A$, we need to compute almost common lattices of $A$. By an almost common lattice, we mean a proper sublattice $\Gamma \subsetneq \mathcal{L}(A)$ such that most of the vectors in $A$ lie in $\Gamma$. In the 1-dimensional case, assuming the greatest common divisor of $A$ is $1$, an almost common lattice corresponds to a divisor that divides most of the integers in $A$, and can be computed in $\widetilde{O}(n + \sqrt{N})$ time via an algorithm by Bringmann and Wellnitz~\cite{BW21}. Below we summarize their algorithm.

\begin{theorem}[{\cite[Theorem 3.8]{BW21}}]\label{thm:factorization}
    The prime factorization of any $n$ given integers in $[N]$ can be computed in $\widetilde{O}(n + \sqrt{N})$ time. 
\end{theorem}

Although in a slightly different form, the following lemma is essentially the same as Theorem 4.12 in~\cite{BW21}. We provide a proof for completeness.
\begin{lemma}\label{lem:almostc-common-divisor}
    Given a multiset $A$ of $n$ integers in $[N]$, and an integer $k$, we can decide whether $A$ contains $k$ elements that have a common divisor greater than $1$, and compute such a divisor if it exists, in $\widetilde{O}(n + \sqrt{N})$ time.
\end{lemma}
\begin{proof}
    Note that if some integers have a common divisor greater than $1$, then they must have a prime common divisor.  We first use Theorem~\ref{thm:factorization} to compute the prime factorization of all integers in $A$ in $\widetilde{O}(n + \sqrt{N})$ time. Then for each prime $p$, we can determine the number of integers in $A$ that are divisible by $p$. If some $p$ divides at least $k$ integers in $A$, then we return $p$. Otherwise, no $k$ elements of $A$ have a common divisor greater than $1$.
\end{proof}

We shall generalize the above result to multi-dimension as follows. For $j \in \{1, \ldots, d\}$, let $A_j$ be the projection of $A$ onto the subspace spanned by $\vec{e}_1, \ldots, \vec{e}_j$. We will iteratively check whether $A_j$ has an almost common lattice. If some $A_j$ has an almost common lattice, then so does $A$. Suppose that $A_j$ does not have an almost common lattice. Then using Observation~\ref{lem:common-divisor-for-det}, we show that deciding whether $A_{j+1}$ has an almost common lattice can be reduced to deciding whether a certain multiset of determinants has an almost common divisor, which can be solved via Lemma~\ref{lem:almostc-common-divisor}.

\begin{observation}\label{lem:common-divisor-for-det}
    Let $A \subseteq \mathbb{Z}^d$ be a finite set of vectors. If $A$ lies in some full-dimensional lattice $\Gamma$, then $\det(\Gamma)$ is a common divisor of the integers in the following multiset.
    \[
        \left\{\det(\vec{a}_1, \ldots, \vec{a}_d) : \vec{a}_i \in A\right\}
    \]
\end{observation}
\begin{proof}
    Since $A \subseteq \Gamma$, the lattice generated by $\vec{a}_1, \ldots, \vec{a}_d$ must be a sublattice of $\Gamma$. Therefore, $\det(\Gamma)$ divides $\det(\vec{a}_1, \ldots, \vec{a}_d)$.
\end{proof}

\begin{lemma}\label{lem:almost-common-lattice}
    Let $A \subseteq [N]^d$ be a finite set of $n$ vectors, where $N$ is sufficiently large and $Z : =\mathcal{L}(A)$ is full-dimensional. Let $\tau := \sqrt{N^d/\det(Z)}$. For any integer $k \geq 1$, in $\widetilde{O}((n + d^d\tau)\cdot \mathrm{poly}(d))$ time, with error probability at most $N^{-\Omega(d)}$, we can 
    \begin{itemize}
        \item either find a proper sublattice $\Gamma \subsetneq Z$ such that
        \[
            |A \setminus \Gamma| \leq k\cdot (d\log N)^{2d};
        \]
        
        \item or assert that $|A \setminus \Gamma| \geq k$ for any proper sublattice $\Gamma \subsetneq Z$.
    \end{itemize}
\end{lemma}
\begin{proof}
    When $k = 1$, the lemma holds trivially since $|A \setminus \Gamma| \geq 1$ for any proper sublattice $\Gamma \subsetneq \mathcal{L}(A)$. Assume that $k \geq 2$.  

    For $j \in \{1, \ldots, d\}$, let $A_j$  be the projection of $A$ onto the $j$-dimensional space spanned by $\vec{e}_1, \ldots, \vec{e}_j$. (Note that $A_j$ can be a multiset.) Let $Z_j := \mathcal{L}(A_j)$. Note that $Z_j$ must be full-dimensional. Let $k_j := k\cdot (d\log N)^{2(d-j)}$. We say that $A_j$ is bad if $|A_j \setminus \Gamma_j| \leq k_j$ for some proper sublattice $\Gamma_j \subsetneq Z_j$, and good otherwise. We shall iteratively determine whether $A_j$ is good or bad for $j \in \{1, \ldots, d\}$. If some $\Gamma_j$ certifies that $A_j$ is bad, then by lifting $\Gamma_j$ to $d$-dimension, we can obtain a proper sublattice $\Gamma \subsetneq Z$ such that
    \[
        |A \setminus \Gamma| \leq k_j \leq k\cdot (d\log N)^{2d};
    \] 
    otherwise, $A_d$ is good. (Note that $A_d = A$ and $k_d = k$.) This implies 
    \(
        |A \setminus \Gamma| \geq k
    \)
    for any $\Gamma \subsetneq Z$.

    Consider the base case where $j = 1$. Note that $A_1$ is a multiset of integers from $[N]$, and all its elements are multiples of $\det(Z_1)$. Let $A'_1$ be obtained by dividing each element of $A_1$ by $\det(Z_1)$. To determine whether $A_1$ is bad or not, it suffices to check whether $n - k_1$ integers in $A'_1$ have a common divisor greater than $1$, and find such a divisor if it exists. This can be done by Lemma~\ref{lem:almostc-common-divisor} in $\widetilde{O}(n + \sqrt{N})$ time. 

    If $A_{j-1}$ is bad, then we are done. Suppose that $A_{j-1}$ is good. We need to check whether $A_j$ is bad or not. We shall show that if $A_j$ is bad, then in $\widetilde{O}((n + \sqrt{N}) \cdot \mathrm{poly}(j))$ time and with error probability at most $\frac{1}{8}$, we can obtain a lattice $\Gamma^*_j \subsetneq Z_j$ certifying that $A_j$ is bad. By repeating the procedure for $O(d\log N)$ times, we can reduce the error probability to $N^{-\Omega(d)}$. The lemma follows directly by simply iterating over $j \in \{1, \ldots, d\}$.

    Suppose that $A_j$ is bad. That is, $|A_j \setminus \Gamma_j| \leq k_j$ for some $\Gamma_j \subsetneq Z_j$. We sample the elements of $A_j$ with probability $\frac{1}{k_j}$, and let $S$ be the set of sampled elements. Let $S'$ be the projection of $S$ onto the $(j-1)$-dimensional space spanned by $\vec{e}_1, \ldots, \vec{e}_{j-1}$. 

    \begin{claim}\label{clm:almost-common-lattice-1}
       With probability at least $\frac{1}{8}$, we have $S \subseteq \Gamma_j$ and  $\mathcal{L}(S') = Z_{j-1}$.
    \end{claim}
    \begin{proof}
        It suffices to show that
        \begin{align}
            \mathbf{Pr}[\text{$S \subseteq \Gamma_j$}] \geq \frac{1}{4},\label{eq:cl-1}\\
            \mathbf{Pr}[\mathcal{L}(S') = Z_{j-1}] \geq \frac{7}{8}\label{eq:cl-2}
        \end{align}
        Then the lemma holds by the union bound.

        Inequality~\eqref{eq:cl-1} is easy to prove. Note that $S \subseteq \Gamma_j$ if none of the elements in $A_j \setminus \Gamma_j$ are sampled. (Recall that $k_j \geq k \geq 2$.) Therefore,
        \[
            \mathbf{Pr}[\text{$S \subseteq \Gamma_j$}] \geq (1- \frac{1}{k_j})^{k_j} \geq \frac{1}{4}.
        \]

        Next we show~\eqref{eq:cl-2}. Since $S' \subseteq Z_{j-1}$, we have $\mathcal{L}(S') \neq Z_{j-1}$ if and only if $S' \subseteq \Gamma$ for some lattice $\Gamma \subsetneq Z_{j-1}$. It suffices to bound the probability of the latter. Note that $S'$ can be regarded as a multiset whose elements are sampled from $A_{j-1}$ with probability $\frac{1}{k_j}$. Fix an arbitrary proper sublattice $\Gamma \subsetneq Z_{j-1}$. Since $A_{j-1}$ is good, we have that 
        \[
            |A_{j-1} \setminus \Gamma| \geq k_{j-1}.
        \]
        Note that $S' \subseteq \Gamma$ if and only if none of the elements in $A_{j-1} \setminus \Gamma$ are sampled. Therefore, we have
        \[
            \mathbf{Pr}[S' \subseteq \Gamma] \leq (1 - \frac{1}{k_j})^{k_{j-1}} \leq \exp(-k_{j-1}/k_j) = \exp(-d^2\log^2 N) \leq N^{-d^2\log N}.
        \]
        We claim that if $S'$ lies in some proper sublattice of $Z_{j-1}$, it must also lie in a proper sublattice of $Z_{j-1}$ generated by vectors in $[N]^{j-1}$. To see the claim, consider $\mathcal{L}(S')$. It is a proper sublattice of $Z_{j-1}$ generated by vectors in $[N]^{j-1}$, and $S'$ lies in it. Since there are at most $N^{(j-1)^2}$ lattices that can be generated by vectors in $[N]^{j-1}$, by the union bound, we have
        \[
            \mathbf{Pr}[\text{$S' \subseteq \Gamma$ for some $\Gamma \subsetneq Z_{j-1}$}] \leq N^{-d^2\log N} \cdot N^{(j-1)^2} \leq N^{-d^2\log N+d^2} \leq \frac{1}{8},
        \]
        which implies~\eqref{eq:cl-2}.
    \end{proof}

    \begin{claim}\label{clm:almost-common-lattice-2}
        Suppose that $S \subseteq \Gamma_j$ and  $\mathcal{L}(S') = Z_{j-1}$. In $\widetilde{O}((n + \sqrt{N})\cdot \mathrm{poly}(j))$ time, we can compute a lattice $\Gamma^*_j \subsetneq Z_j$ such that $|A_j \setminus \Gamma^*_j| \leq k_j$.
    \end{claim}
    \begin{proof}
        Consider the lattice $\mathcal{L}(S)$. Let $\vec{b}_1, \ldots, \vec{b}_{j-1}, \vec{b}_j$ be the basis of $\mathcal{L}(S)$ in Hermite normal form. (Note that $\vec{b}_j$ can be $\vec{0}$ if $\mathcal{L}(S)$ is of rank $j-1$, but $\vec{b}_1, \ldots, \vec{b}_{j-1}$ must be non-zero since $\mathcal{L}(S') = Z_{j-1}$.) This basis can be obtained in $\widetilde{O}(j^3n\log N)$ time via Lemma~\ref{lem:compute-hnf}. Now consider the following multiset of integers.
        \[
            D := \{|\det(\vec{b}_1, \ldots, \vec{b}_{j-1}, \vec{a})| : \vec{a} \in A_j\}
        \]
        Note that $\max(D) \leq \det(Z_j) \cdot \sqrt{j} N$ since $\|\vec{a}\|_{2} \leq \sqrt{j}N$. Moreover, all the integers in $D$ are multiples of $\det(Z_j)$ due to $S \subseteq A_j \subseteq Z_j$.  Since $|A_j \setminus \Gamma_j| \leq k_j$ and $S \subseteq \Gamma_j$, by Observation~\ref{lem:common-divisor-for-det}, $D$ contains at least $n - k_j$ integers that are divisible by $q\cdot \det(Z_j)$ for some $q > 1$.  Such a $q$ can be found in $\widetilde{O}((n + \sqrt{N}) \cdot \mathrm{poly}(j))$ time by first dividing all the elements in $D$ by $\det(Z_j)$ and then applying Lemma~\ref{lem:almostc-common-divisor}. Let 
        \[
            A^*_j := \{\text{$\vec{a} \in A_j$ : $q \cdot \det(Z_j)$ divides $\det(\vec{b}_1, \ldots, \vec{b}_{j-1}, \vec{a})$}\}.
        \]
        Note that $|A^*_j| \geq n - k_j$.  Let $\Gamma^*_j := \mathcal{L}(A^*_j)$. We show that $\Gamma^*_j$ is our target lattice. 

       It is easy to see that $|A_j \setminus \Gamma^*_j| = |A_j| - |A^*_j| \leq k_j$. It remains to show that $\Gamma^*_j \subsetneq Z_j$. Since $A^*_j \subseteq A_j \subseteq Z_j$, we have $\Gamma^*_j \subseteq Z_j$. If $\Gamma^*_j$ is not full-dimensional, then it must be that $\Gamma^*_j \neq Z_j$ since $Z_j$ is full-dimensional; we are done. Suppose that $\Gamma^*_j$ is full-dimensional. To show $\Gamma^*_j \neq Z_j$, it suffices to show that $\det(\Gamma^*_j) > \det(Z_j)$. Recall that $\mathcal{L}(S') = Z_{j-1}$. It implies that the projections of $\vec{b}_1, \ldots, \vec{b}_{j-1}$ on the $(j-1)$-dimensional space form the basis of $Z_{j-1}$. Since $A^*_j \subseteq Z_j$, for any $\vec{a} \in A^*_j$, it can be represented as
        \[
            \vec{a} = z_1\vec{b}_1 + \cdots + z_{j-1}\vec{b}_{j-1} + z_j \vec{e}_j,
        \]
        and therefore, we have $\det(\vec{b}_1, \ldots, \vec{b}_{j-1}, \vec{a}) =z_j \cdot \det(Z_{j-1})$. By definition of $A^*_j$, we have $q\cdot \det(Z_{j})$ divides $z_j \cdot \det(Z_{j-1}) $. It implies that $A^*_j$ is in the following lattice with basis
        \[
            \vec{b}_1, \ldots ,\vec{b}_{j-1}, \frac{q\cdot \det(Z_{j})}{\det(Z_{j-1})}\cdot \vec{e}_j.
        \]
        The determinant of this lattice is at least $q\cdot \det(Z_{j}) > \det(Z_{j})$. Recall that $\Gamma^*_j = \mathcal{L}(A^*_j)$. It follows that $\det(\Gamma^*_j) \geq q\cdot \det(Z_{j}) > \det(Z_{j})$.
    \end{proof}

    One can verify that the total time cost of the construction is $\widetilde{O}((n + \sqrt{N}) \cdot \mathrm{poly}(d))$. We briefly explain how to reduce it to $\widetilde{O}((n + d^d\tau) \cdot \mathrm{poly}(d))$. Basically, using a linear map $B^{-1}$, where $B$ represents a Mahler's basis of $Z$, one can map $A$ to a set $A'$ of integer vectors in $[-N' ,N']^d$, where 
    \[
        N' = d^{2d}\cdot \frac{N^d}{\det(Z)} = (d^d\tau)^2.
    \]
    Then we apply the above procedure to $A'$. The Mahler's basis $B$ can be computed in $O(d^3n\log N)$ time, and applying the above procedure to $A'$ takes time
    \[
        \widetilde{O}((n + \sqrt{N'})\cdot \mathrm{poly}(d)) = \widetilde{O}((n + d^d\tau)\cdot \mathrm{poly}(d)). \qedhere
    \]
\end{proof}

\subsection{Extracting Spread Subsets}
Once we can compute almost common lattices, given an $\widetilde{O}(1)$-dense and $\widetilde{O}(1)$-nondegenerate set $A \subseteq [N]^d$, we are able to extract a subset $A^*$ that is $\widetilde{O}(\sqrt{\ell})$-dense, $\widetilde{O}(\sqrt{\ell})$-spread, and $\widetilde{O}(\sqrt{\ell})$-nondegenerate, where $\ell: = \frac{\det(Z^*)}{\det(Z)}$ and $Z^* := \mathcal{L}(A^*)$ and $Z: = \mathcal{L}(A)$.

The approach is simple. If $A$ is not spread, then there must be a proper sublattice $\Gamma$ of $\mathcal{L}(A)$ such that only a few vectors in $A$ are not in $\Gamma$. Then we remove those vectors not in $\Gamma$. We repeat the procedure until the remaining vectors form a spread subset.

\begin{lemma}\label{lem:compute-spread-subset}
    Let $A \subseteq [N]^d$ be a set of $n$ vectors that is $\rho$-dense and $\delta$-nondegenerate, and let $\gamma$ be a positive real. Let $Z: = \mathcal{L}(A)$ and $\tau := \sqrt{N^d/\det(Z)}$. Assume that $\rho \geq \delta > 2(d\log N)^{4d}\cdot \gamma$. In $\widetilde{O}((n + d^d\tau)\cdot \mathrm{poly}(d))$ time, with error probability at most $N^{-\Omega(d)}$, we can compute a subset $A^* \subseteq A$ that is $\frac{\rho}{2}\sqrt{\ell}$-dense, $\gamma\sqrt{\ell}$-spread, and $\frac{\delta}{2}\sqrt{\ell}$-nondegenerate, where $\ell: = \frac{\det(Z^*)}{\det(Z)}$ and $Z^* := \mathcal{L}(A^*)$.
\end{lemma}
\begin{proof}
    Let $k := \gamma \cdot \sqrt{N^d/\det(Z)}$.  We shall construct a sequence $A_0 \supsetneq A_1 \supsetneq \cdots \supsetneq A_r$ until $Z_r := \mathcal{L}(A_r)$ is not full-dimensional or $|A_r \setminus \Gamma| \geq k$ for any $\Gamma \subsetneq Z_r$. Initially, $i := 0$ and $A_0 := A$. Given $A_i$, we can check whether $Z_i := \mathcal{L}(A_i)$ is full-dimensional or not in $\widetilde{O}(d^3n\log N)$ time via Lemma~\ref{lem:compute-hnf}. If not, then we stop and set $r := i$. Suppose that $Z_i$ is full-dimensional. Then we invoke Lemma~\ref{lem:almost-common-lattice} on $A_i$ and $k$, which takes $\widetilde{O}((n + d^d\cdot \tau)\cdot \mathrm{poly}(d))$ time. If it asserts that $|A_i \setminus \Gamma| \geq k$ for any $\Gamma \subsetneq Z_i$, then we stop and set $r := i$. Otherwise, it returns a $\Gamma_i \subsetneq Z$ such that 
    \[
        |A_i \setminus \Gamma_i| \leq k\cdot (d\log N)^{2d}.
    \]
    We set $A_{i+1} := A_i \cap \Gamma_i$, and then proceed with $i := i+1$.

    We first show that $r \leq 1 + d\log (dN)$. Suppose not. Note that $Z_{r-1}$ must be full-dimensional. Observing that $Z_{i} \subsetneq Z_{i-1}$ for all $i$, we have 
    \[
        \det(Z_{r-1}) \geq 2^{r-1} \cdot \det(Z_0) > (dN)^{d}.
    \] 
    But this is impossible since $Z_{r-1} = \mathcal{L}(A_{r-1})$ and the vectors in $A_{r-1}$ have $\ell_2$-norm at most $\sqrt{d} N$.

    Note that $|A_{i}| - |A_{i+1}| \leq k\cdot (d\log N)^{2d}$. We have
    \[
        |A_r| \geq |A| - rk\cdot (d\log N)^{2d} \geq |A| - k\cdot (d\log N)^{4d}.
    \]
    $Z_r$ must be full-dimensional since $A$ is $\delta$-nondegenerate and $\delta > 2(d\log N)^{4d}\cdot \gamma$. By the stopping condition, it must be that $|A_r \setminus \Gamma| \geq k$ for any $\Gamma \subsetneq Z_r$. Note that 
    \[
        k = \gamma \cdot \sqrt{N^d/\det(Z)} = \gamma \sqrt{\frac{\det(Z_r)}{\det(Z)}} \cdot \sqrt{\frac{N^d}{\det(Z_r)}}.
    \]
    So $A_r$ is $\gamma\sqrt{\frac{\det(Z_r)}{\det(Z)}}$-spread. Also $A_r$ is $\frac{\rho}{2}\sqrt{\frac{\det(Z_r)}{\det(Z)}}$-dense and $\frac{\delta}{2}\sqrt{\frac{\det(Z_r)}{\det(Z)}}$-nondegenerate, since at most $\frac{\delta}{2} \cdot \sqrt{N^d/ \det(Z)}$ elements are removed from $A$ to obtain $A_r$.

    It is easy to verify that the total running time is $\widetilde{O}((n + d^d\tau)\cdot \mathrm{poly}(d))$ time and the total error probability is at most $N^{-\Omega(d)}$.
\end{proof}

\subsection{Reducing to Modular Subset Sum}
Using the $\frac{\rho}{2}\sqrt{\ell}$-dense, $\gamma\sqrt{\ell}$-spread, and $\frac{\delta}{2}\sqrt{\ell}$-nondegenerate subset $A^*$ obtained via Lemma~\ref{lem:compute-spread-subset}, we show that we can obtain a lattice $\Gamma$ so that deciding whether $\vec{t} \in \mathcal{S}(A)$ is equivalent to deciding whether $\vec{t} \in \mathcal{S}_{\Gamma}(A)$. That is, the subset sum problem can be reduced to the modular subset sum problem, which can be solved efficiently via Lemma~\ref{lem:faster-modular-ss}.

The lattice $\Gamma$ is obtained via the following lemma, which heavily relies on Kneser's theorem.

\begin{lemma}\label{lem:reduce-1}
    Let $A \subseteq [N]^d$ be a set of $n$ vectors, where $Z: = \mathcal{L}(A)$ is full-dimensional. For any full-dimensional lattice $Z^* \subseteq Z$, in $\widetilde{O}((n + 2^d\ell) \cdot \mathrm{poly}(d))$ time, where $\ell := \det(Z^*)/\det(Z)$, we can obtain a lattice $\Gamma$ with $Z^* \subseteq \Gamma \subseteq Z$ such that the following holds. $A$ can be partitioned into $R \cup G$ so that 

    \begin{enumerate}[label={\normalfont (\roman*)}]
        \item $|R| \leq 2\ell$ and $G \subseteq \Gamma$

        \item for any $\vec{t}$ with $\vec{t} \bmod \Gamma \in \mathcal{S}_{\Gamma}(A)$, we have $\vec{t} \bmod Z^* \in \mathcal{S}_{Z^*}(R)$.
    \end{enumerate}
\end{lemma}
\begin{proof}
    We shall derive a sequence of lattices $\Gamma_0 \subsetneq \Gamma_1 \subsetneq \cdots \subsetneq \Gamma_r \subseteq Z$ as follows. Initially, $i = 0$ and $\Gamma_0 := Z^*$. Given $\Gamma_i$, let $\ell_i := \det(\Gamma_i)/\det(Z)$. If $|A \setminus \Gamma_i| < \ell_i$, then we stop, let $A_i := A \setminus \Gamma_i$, and let $r := i$. Otherwise, $A$ has at least $\ell_i$ elements not in $\Gamma_i$. Let them be $\{\vec{a}_1, \ldots, \vec{a}_{\ell_i}\} =: A_i$. Note that
    \[
        |\{0, \vec{a}_1\}| + \cdots + |\{0, \vec{a}_{\ell_i}\}| \geq 2\ell_i.
    \]
    Also note that $|\mathcal{S}_{\Gamma_i}(A_i)| \leq |Z \bmod \Gamma_i| \leq \det(\Gamma_i)/\det(Z) \leq \ell_i$. So by Corollary~\ref{coro:kneser-2}, there is a vector $\vec{h}_i \notin \Gamma_i$ such that 
    \begin{equation}\label{eq:fill-lattice}
        \vec{h}_i + \mathcal{S}(A_i) \equiv \mathcal{S}(A_i) \pmod {\Gamma_i}.
    \end{equation}
    We set $\Gamma_{i+1} := \Gamma_i + \mathcal{L}(\vec{h}_i)$, and repeat the procedure with $i := i+1$. 

    We first show that $\Gamma_0 \subsetneq \Gamma_1 \subsetneq \cdots \subsetneq \Gamma_r \subseteq Z$. Since $\vec{h}_{i} \notin \Gamma_i$, by the construction of $\Gamma_{i+1}$, it is easy to see that $\Gamma_i \subsetneq \Gamma_{i+1}$. It remains to prove that $\Gamma_i \subseteq Z$ for all $i$. We prove this by induction on $i$. The base case is trivial since $\Gamma_0 = Z^* \subseteq Z$. Assume that $\Gamma_i \subseteq Z$. We show that $\Gamma_{i+1} \subseteq Z$. Since $\vec{0} \in \mathcal{S}(A_i)$, in view of~\eqref{eq:fill-lattice}, $\vec{h}_i \in \mathcal{S}(A_i) + \Gamma_i$. Note that $A_i \subseteq A \subseteq Z$, so $\vec{h}_i\in Z$, which implies $\Gamma_{i+1} \subseteq Z$.

    Since $\det(\Gamma_i) \leq \det(\Gamma_{i-1})/2$ for all $i \geq 1$, we have that 
    \[
        r \leq \log \frac{\det(Z^*)}{\det(Z)} = \log \ell.
    \]

    Let $R := A_0 \cup \cdots \cup A_r$ and let $\Gamma := \Gamma_r$. We show that $R$ and $\Gamma$ satisfy the stated properties in the lemma. By construction, we have $|A_i| \leq \det(\Gamma_i)/\det(Z)$. So
    \[
        |R| \leq (\det(\Gamma_0) + \cdots + \det(\Gamma_r))/\det(Z) \leq 2\det(Z^*)/\det(Z) = 2\ell.
    \]
    Also since $A_r = A \setminus \Gamma_r$, we have $A \setminus R \subseteq \Gamma_r$. Property (i) is satisfied.

    It remains to show property (ii). We claim that
    \begin{equation}\label{eq:reduce-1-key}
        \Gamma + \mathcal{S}(R) = \mathcal{S}(R) + Z^*.
    \end{equation}
    Suppose that the claim holds. Pick an arbitrary $\vec{t}$ with $\vec{t} \bmod \Gamma \in \mathcal{S}_{\Gamma}(A)$. Since $G \subseteq \Gamma$, we have that $\vec{t} \bmod \Gamma \in \mathcal{S}_{\Gamma}(R)$. In other words, $\vec{t} \in \Gamma + \mathcal{S}(R)$. In view of~\eqref{eq:reduce-1-key}, we have $\vec{t} \in \mathcal{S}(R) + Z^*$, or equivalently, $\vec{t} \bmod Z^* \in \mathcal{S}_{Z^*}(R)$. So property (ii) holds.

    Now we prove~\eqref{eq:reduce-1-key}. We shall prove, by induction on $i$, that $\Gamma_i + \mathcal{S}(R) = \mathcal{S}(R) + Z^*$ for all $i = 0, \ldots, r$. The base case holds trivially, since $\Gamma_0 = Z^*$. Suppose the equation holds for some $i$. We show that it also holds for $i+1$. Recall that $\Gamma_{i+1} = \Gamma_i + \mathcal{L}(\vec{h}_i)$, where $\vec{h}_i$ satisfies
    \[
        \vec{h}_i + \mathcal{S}(A_i) \equiv \mathcal{S}(A_i) \pmod {\Gamma_i},
    \]
    or in equivalent form,
    \[
        \vec{h}_i + \Gamma_i + \mathcal{S}(A_i) =  \mathcal{S}(A_i) + \Gamma_i.
    \]
    Therefore, 
    \[
        \Gamma_{i+1} + \mathcal{S}(A_i) = \mathcal{S}(A_i) + \Gamma_i.
    \]
    Adding $\mathcal{S}(R \setminus A_i)$ to both sides gives
    \[
        \Gamma_{i+1} + \mathcal{S}(R) = \mathcal{S}(R) + \Gamma_i.
    \]
    By the inductive hypothesis, we have
    \[
        \Gamma_{i+1} + \mathcal{S}(R) = \mathcal{S}(R) + Z^*. 
    \]

    We bound the time cost of the construction. There are at most $\log \ell$ iterations. In the $i$-th iteration, we need to do the following.
    \begin{itemize}
        \item Select elements of $A$ that do not belong to $\Gamma_i$, which takes $O(d^2n)$ time via Lemma~\ref{lem:compute-rem}(i).

        \item Compute a non-zero $\vec{h}_i \in \mathrm{Sym}_{\Gamma_i}(\mathcal{S}(A_i))$, which takes $\widetilde{O}((n + 2^d\ell) \cdot \mathrm{poly}(d))$ via Lemma~\ref{lem:faster-modular-ss}.

        \item Compute $\Gamma_{i+1} := \Gamma_i + \mathcal{L}(\vec{h}_i)$ via Lemma~\ref{lem:compute-rem}(iii), which takes $O(d^2\log \ell')$ time. Note that $\ell' = \det(\Gamma_i) \leq \ell \det(Z) \leq \ell N^d$.
    \end{itemize}
    Therefore, the total time cost is $\widetilde{O}((n + 2^d\ell) \cdot \mathrm{poly}(d))$.
\end{proof}

Next we show that using $\Gamma$, we can reduce the subset sum problem to the modular subset sum problem.

\begin{lemma}\label{lem:reduce2modular}
    Let $A \subseteq [N]^d$ be a set of vectors, and let $A^*$ be a $\rho\sqrt{\ell}$-dense, $\gamma\sqrt{\ell}$-spread, and $\delta\sqrt{\ell}$-nondegenerate subset of $A$, where $\ell := \det(Z^*)/\det(Z)$, $Z := \mathcal{L}(A)$ and $Z^* := \mathcal{L}(A^*)$. Assume that $N$ is sufficiently large and 
    \[
        \rho \geq (d\log N)^{4d^2} \quad \text{and} \quad \rho\gamma \geq (d\log N)^{4d^2} \quad \text{and} \quad \rho\delta \geq (d\log N)^{15d^2}. 
    \]
    Then, in $\widetilde{O}((n + 2^d \ell) \cdot \mathrm{poly}(d))$ time, we can obtain a lattice $\Gamma$ with $Z^* \subseteq \Gamma \subseteq Z$ such that the following holds. For every
    \[
        \vec{t} \in (1- \varepsilon)\mathcal{P}(A) + \frac{\sigma(A)}{2},
    \]
    we have
    \[
        \vec{t} \in \mathcal{S}(A) \text{ if and only if } \vec{t} \bmod \Gamma \in \mathcal{S}_{\Gamma}(A),
    \]
    where $\varepsilon := \frac{(d\log N)^{15d^2} }{\rho\delta}$.
\end{lemma}
\begin{proof}
   We first apply Lemma~\ref{lem:reduce-1} to $A$ and $Z^*$. It gives a lattice $\Gamma$ with $Z^* \subseteq \Gamma \subseteq Z$ and a partition $R \cup G$ of $A$ such that 
    \begin{enumerate}[label={\normalfont (\roman*)}]
        \item $|R| \leq 2\ell$ and $G \subseteq \Gamma$

        \item for any $\vec{t}$ with $\vec{t} \bmod \Gamma \in \mathcal{S}_{\Gamma}(A)$, we have $\vec{t} \bmod Z^* \in \mathcal{S}_{Z^*}(R)$.
    \end{enumerate}

    Note that $\Gamma$ can be obtained in $\widetilde{O}((n + 2^d \ell) \cdot \mathrm{poly}(d))$ time via Lemma~\ref{lem:reduce-1}. We shall show that $\Gamma$ satisfies the stated property.

    Let $R^* := R \setminus A^*$ and $G^* := G \setminus A^*$. Then $A^* \cup R^* \cup G^*$ forms a partition of $A$. Since $A^* \subseteq Z^*$, it is easy to see that $R^*$ and $G^*$ also satisfy properties (i) and (ii) of $R$ and $G$.

    Let $\varepsilon' := \frac{(d\log N)^{14d^2}}{\rho\delta}$.
    \begin{claim}\label{clm:struct-1}
        $2\mathcal{P}(R^*)\subseteq \varepsilon'\mathcal{P}(A^*)$ and $\{x_1\vec{e_1} + \cdots + x_d\vec{e}_d: 0\leq x_j \leq dN\}\subseteq \varepsilon'\mathcal{P}(A^*)$.
    \end{claim}
    \begin{proof}
        Since $A^*$ is $\rho\sqrt{\ell}$-dense and $\delta\sqrt{\ell}$-nondegenerate, by Corollary~\ref{coro:large-inradius}, we have
        \begin{equation}\label{eq:struct-1-1}
            \mathcal{P}(A^*) \supseteq \left\{x_1\vec{e}_1 + \cdots + x_d\vec{e}_d : |x_j| \leq \frac{\rho\delta\ell N}{(d\log N)^{7d^2}}\right\}.
        \end{equation}
        Since $|R^*| \leq 2\det(Z^*)/\det(Z) = 2\ell$, we have
        \[
            \mathcal{P}(R^*) \subseteq \left\{x_1\vec{e}_1 + \cdots + x_d\vec{e}_d : |x_j| \leq 2\ell N\right\}.
        \]
        Then the two inequalities in the claim follow directly.
    \end{proof}

    \begin{claim}\label{clm:struct-2}
        For every 
        \[
            \vec{t} \in (1 - 2\varepsilon')\mathcal{P}(A^*) + \frac{\sigma(A^*)}{2},
        \]
        $\vec{t} \in \mathcal{S}(A^* \cup R^*)$ if and only if $\vec{t}\bmod \Gamma\in \mathcal{S}_{\Gamma}(A^* \cup R^*)$.
    \end{claim}
    \begin{proof}
        The only-if part is trivial. We prove the if part. Fix an arbitrary $\vec{t} \in (1 - 2\varepsilon')\mathcal{P}(A^*) + \frac{\sigma(A^*)}{2}$. Note that $\vec{t}\bmod \Gamma\in \mathcal{S}_{\Gamma}(A^* \cup R^*)$ implies that $\vec{t} \bmod \Gamma \in \mathcal{S}_{\Gamma}(A)$. Then by property (ii), there exists $\vec{r} \in \mathcal{S}(R^*)$ such that $\vec{t} \equiv \vec{r} \pmod {Z^*}$.  Let $\vec{t'} = \vec{t} - \vec{r}$. To show $\vec{t} \in \mathcal{S}(A^* \cup R^*)$, it suffices to show that $\vec{t'} \in \mathcal{S}(A^*)$.

        Since $\vec{t} \equiv \vec{r} \pmod {Z^*}$, we have $\vec{t}' \in Z^*$. In view of Claim~\ref{clm:struct-1}, we have $\vec{r} \in \mathcal{S}(R^*) \subseteq 2\mathcal{P}(R^*) \subseteq \varepsilon' \mathcal{P}(A^*)$. By Lemma~\ref{lem:vol-change}(i), we have
        \[
            \vec{t}' \in \left((1 - \varepsilon')\mathcal{P}(A^*) + \frac{\sigma(A^*)}{2}\right).
        \]
        Recall that $A^*$ is $\rho\sqrt{\ell}$-dense, $\gamma\sqrt{\ell}$-spread, and $\delta\sqrt{\ell}$-nondegenerate. By Theorem~\ref{thm:john-type}, we have $\vec{t}' \in \mathcal{S}(A^*)$. 
    \end{proof}

    \begin{claim}\label{clm:struct-3}
        For any 
        \[
            \vec{t} \in (1 - 3\varepsilon')\mathcal{P}(A^* \cup G^*) + \frac{\sigma(A^* \cup G^*)}{2},
        \]
        $\vec{t} \in \mathcal{S}(A)$ if and only if $\vec{t} \bmod \Gamma \in \mathcal{S}_{\Gamma}(A)$.
    \end{claim}
    \begin{proof}
        The only-if part is trivial. We prove the if part. 

        Take an arbitrary $\vec{t} \in (1 - 3\varepsilon')\mathcal{P}(A^* \cup G^*) + \frac{\sigma(A^* \cup G)}{2}$ with $\vec{t} \bmod \Gamma \in \mathcal{S}_{\Gamma}(A)$. The vector $\vec{t}$ can be written as $\vec{t} = \vec{t}_{A^*} + \vec{t}_{G^*}$ where $\vec{t}_{A^*} \in (1 - 3\varepsilon')\mathcal{P}(A^*) + \frac{\sigma(A^*)}{2}$ and $\vec{t}_{G^*} \in (1 - 3\varepsilon')\mathcal{P}(G^*) + \frac{\sigma(G^*)}{2}$.  By Lemma~\ref{lem:proximity}, there exists $\vec{t}'_{G^*} \in \mathcal{S}(G^*)$ such that
        \begin{equation}\label{eq:struct-3-1}
            \vec{t}_{G^*} - \vec{t}'_{G^*} \in \{x_1\vec{e}_1 + \cdots + x_d\vec{e}_d : 0 \leq x_i \leq dN\}.
        \end{equation}

        Now consider $\vec{t} - \vec{t}'_{G^*} = \vec{t}_{A^*} + \vec{t}_{G^*} - \vec{t}'_{G^*}$. To prove the claim, it suffices to show $\vec{t} - \vec{t}'_{G^*} \in \mathcal{S}(A^* \cup R^*)$.  In view of Claim~\ref{clm:struct-1} and~\eqref{eq:struct-3-1}, we have $\vec{t}_{G^*} - \vec{t}'_{G^*} \in \varepsilon' \mathcal{P}(A^*)$. By Lemma~\ref{lem:vol-change}(i),
        \[
            \vec{t}_{A^*} + \vec{t}_{G^*} - \vec{t}'_{G^*} \in (1 - 2\varepsilon')\mathcal{P}(A^*) + \frac{\sigma(A^*)}{2},
        \]
        or equivalently,
        \begin{equation}\label{eq:struct-3-2}
            \vec{t} - \vec{t}'_{G^*}  \in (1 - 2\varepsilon')\mathcal{P}(A^*) + \frac{\sigma(A^*)}{2}.
        \end{equation}
        Recall that $G^* \subseteq \Gamma$. So 
        \(
            \vec{t} - \vec{t}'_{G^*} \equiv \vec{t} \pmod {\Gamma}.
        \)
        Since $\vec{t} \bmod \Gamma \in \mathcal{S}_{\Gamma}(A)$, we have 
        \[
            (\vec{t} - \vec{t}'_{G^*}) \bmod \Gamma \in \mathcal{S}_{\Gamma}(A) = \mathcal{S}_{\Gamma}(A^* \cup R^*).
        \]
        The last equality follows from $G^* \subseteq \Gamma$. Then by Claim~\ref{clm:struct-2} and~\eqref{eq:struct-3-2}, 
        \[
            \vec{t} - \vec{t}'_{G^*} \in \mathcal{S}(A^* \cup R^*). \qedhere
        \]
    \end{proof}

    Note that $5\varepsilon' \leq \varepsilon$.  In view of Claim~\ref{clm:struct-3}, to prove the theorem, it suffices to show that 
    \[
        (1 - 5\varepsilon')\mathcal{P}(A) + \frac{\sigma(A)}{2} \subseteq (1 - 3\varepsilon')\mathcal{P}(A^* \cup G^*) + \frac{\sigma(A^* \cup G^*)}{2},
    \]
    or equivalently,
    \begin{equation}\label{eq:struct-4-1}
        (1 - 5\varepsilon')\mathcal{P}(A) + \frac{\sigma(R^*)}{2} \subseteq (1 - 3\varepsilon')\mathcal{P}(A^* \cup G^*).
    \end{equation}
   Recall Claim~\ref{clm:struct-1}. Since $\frac{\sigma(R^*)}{2} \in \mathcal{P}(R^*) \subseteq \varepsilon' \mathcal{P}(A^*) \subseteq \varepsilon' \mathcal{P}(A)$, by Lemma~\ref{lem:vol-change}(ii), we have
   \[
       (1 - 5\varepsilon')\mathcal{P}(A) + \frac{\sigma(R^*)}{2} \subseteq  (1 - 4\varepsilon')\mathcal{P}(A).
   \]
   Note that $(1 - 4\varepsilon')\mathcal{P}(A) = (1 - 4\varepsilon')P(A^* \cup G^*) + (1 - 4\varepsilon')\mathcal{P}(R^*)$ and that $\mathcal{P}(R^*) \subseteq \varepsilon' \mathcal{P}(A^*)\subseteq \varepsilon' \mathcal{P}(A^* \cup G^*)$. Again, by Lemma~\ref{lem:vol-change}(ii), we have
   \[
       (1 - 4\varepsilon')\mathcal{P}(A) \subseteq (1 - 3\varepsilon')P(A^* \cup G^*).
   \]
   So Inequality~\eqref{eq:struct-4-1} follows.
\end{proof}

\subsection{Putting Things Together}
Now we are ready to present a near-linear-time algorithm for dense subset sum in multi-dimension.

\begin{theorem}\label{thm:dense-alg}
   Let $A \subseteq [N]^d$ be a set of $n$ vectors that is $\rho$-dense and $\delta$-nondegenerate, where $N$ is sufficiently large, $\rho \geq \delta$, and
   \[
        \rho \geq (d\log N)^{5d^2}\quad \text{and} \quad \rho\delta \geq (d\log N)^{16d^2}. 
    \]
    Then for any 
    \[
        \vec{t} \in \left((1- \varepsilon)\mathcal{P}(A) + \frac{\sigma(A)}{2}\right),
    \]
    where $\varepsilon: = \frac{(d\log N)^{16d^2}}{\rho\delta}$, we can decide whether $\vec{t} \in \mathcal{S}(A)$ or not in $\widetilde{O}(n \cdot d^{d}\cdot \mathrm{poly}(d))$ time, with error probability at most $N^{-\Omega(d)}$. 
\end{theorem}
\begin{proof}
    Let $Z := \mathcal{L}(A)$. Let $\tau := \sqrt{N^d/\det(Z)}$. Since $A$ is $\rho$-dense, we have that $n \geq \rho\tau$. 

    Let $\gamma : = (d\log N)^{4d^2}/\rho$. By Lemma~\ref{lem:compute-spread-subset}, in $\widetilde{O}((n + d^d \tau)\cdot \mathrm{poly}(d))$ time, with error probability at most $N^{-\Omega(d)}$, we can obtain a subset $A^*$ that is $\frac{\rho}{2}\sqrt{\ell}$-dense and $\gamma \sqrt{\ell}$-spread, and $\frac{\delta}{2}\sqrt{\ell}$-nondegenerate, where $\ell :=\frac{\det(Z^*)}{\det(Z)}$ and $Z^* := \mathcal{L}(A^*)$. 

    Then we apply Lemma~\ref{lem:reduce2modular} to obtain a lattice $\Gamma$ such that for every
    \(
        \vec{t} \in (1- \varepsilon)\mathcal{P}(A) + \frac{\sigma(A)}{2},
    \)
    \[
        \vec{t} \in \mathcal{S}(A) \text{ if and only if } \vec{t} \bmod \Gamma \in \mathcal{S}_{\Gamma}(A),
    \]
    where 
    \[
        \varepsilon = \frac{4(d\log N)^{15d^2} \ell}{\rho\delta \ell} \leq \frac{(d\log N)^{16d^2}}{\rho\delta}.
    \]
    The time cost of Lemma~\ref{lem:reduce2modular} is
    \(
        O((n + 2^d \ell) \cdot \mathrm{poly}(d)).
    \)
    Note that $|A^*| \geq |A|/2 \geq \sqrt{N^d/\det(Z)}$. Therefore, 
    \[
        \det(Z^*) \leq \frac{2^d N^d}{|A^*|} \leq 2^d\sqrt{N^d\cdot \det(Z)},
    \]
    which implies $\ell \leq 2^d n$. So the time cost of this step is $\widetilde{O}(n\cdot 2^{2d}\cdot  \mathrm{poly}(d))$.

    Then to decide whether $\vec{t} \in \mathcal{S}(A)$, it suffices to decide whether $\vec{t} \bmod \Gamma \in \mathcal{S}_{\Gamma}(A)$. This can be done in $\widetilde{O}((n + 2^d\ell') \cdot \mathrm{poly}(d))$ time via Lemma~\ref{lem:faster-modular-ss}, where $\ell' := \det(\Gamma)/\det(Z)$. Note that $\det(\Gamma) \leq \det(Z^*)$ since $Z^* \subseteq \Gamma$. So $\ell' \leq \ell \leq 2^d n$. The running time of this step is bounded by $\widetilde{O}(n\cdot d^{d}\cdot  \mathrm{poly}(d))$.
\end{proof}

\section{Extending to More General Cases}\label{sec:extension}
Sections~\ref{sec:progression},~\ref{sec:john-type}, and~\ref{sec:alg} deal only with the case where $A$ is a set and $A \subseteq [N]^d$. We shall briefly explain how to extend Theorems~\ref{thm:ap-standard-basis},~\ref{thm:john-type}, and~\ref{thm:dense-alg} to a more general case where $A$ is a multiset of vectors from $[-N_1, N_1] \times [-N_2, N_2] \times \cdots \times [-N_d, N_d]$.

We first extend to a set $A \subseteq [N_1] \times \cdots \times [N_d]$. Let $\Phi = \prod_{j=1}^d N_j$. Let $D$ be the $d\times d$ diagonal matrix where $d_{jj} = \Phi/N_j$ for $j \in \{1, \ldots, d\}$. In other words, $D$ is a linear map that stretches the $j$-th coordinate by a factor of $\Phi/N_j$ for each $j$. Now consider $A' := DA$. It is easy to see that $A' \subseteq [\Phi]^d$. Recall Definition~\ref{def:dense-and-spread}. One can verify that the linear map $D$ does not change the value of $\tau$, and therefore, preserves the density, spreadness, and nondegeneracy of $A$. To extend Theorems~\ref{thm:john-type} and~\ref{thm:dense-alg}, we can simply apply them to $A'$ and then map the results back to the original space by left-multiplying $D^{-1}$. Extending Theorem~\ref{thm:ap-standard-basis} is a bit different: we will apply Lemma~\ref{lem:all-lattice-pts} to $A'$, map the result back to the original space by left-multiplying $D^{-1}$, and then apply Lemma~\ref{lem:standard-basis-in-large-area}.

We further consider a set $A \subseteq [-N_1, N_1] \times \cdots \times [-N_d, N_d]$. We claim that allowing $A$ to contain vectors with negative coordinates only affects the terms in the original theorems by a factor of $2^d$.

Finally, we consider the case of multisets. We claim that it suffices to extend Definition~\ref{def:dense-and-spread} to multisets. Let $A$ be a multiset of vectors. For a vector in $A$, we define its multiplicity to be the number of its occurrences in $A$, and we define $\mu(A)$ to be the largest multiplicity of any vector in $A$.

\begin{definition}[Density, Spreadness, Nondegeneracy for Multisets]\label{def:dense-spread-multiset}
    Let $A \subseteq [-N_1, N_1]\times \cdots \times [-N_d, N_d]$ be a multiset of vectors. Let $\Phi := \prod_{j=1}^d N_j$. Let $\mu:= \mu(A)$. Let $\tau := \sqrt{\frac{\mu\Phi}{\det(Z)}}$, where $Z := \mathcal{L}(A)$. We say that $A$ is
    \begin{itemize}
        \item $\rho$-dense if $|A| \geq \rho \tau$;

        \item $\gamma$-spread if $|A \setminus \Lambda| \geq \gamma\tau$ for any lattice $\Lambda \subsetneq Z$;

        \item $\delta$-nondegenerate if $|A \setminus \mathcal{V}| \geq \delta\tau$ for any lower-dimensional subspace $\mathcal{V} \subseteq \mathbb{R}^d$.
    \end{itemize} 
\end{definition}

We summarize this section by the extended versions of Theorems~\ref{thm:ap-standard-basis},~\ref{thm:john-type}, and~\ref{thm:dense-alg}.

\begin{theorem}[Formal version of Theorem~\ref{thm:ap-standard-basis-informal}]\label{thm:ap-standard-basis-extended}
    There exists a constant $C_0$ such that the following holds. Let $A$ be a multiset of vectors from  $[-N_1, N_1] \times \cdots \times [-N_d, N_d]$, where $Z := \mathcal{L}(A)$ is full-dimensional. Let $\mu:=\mu(A)$. Assume that $\Phi := \prod_{j=1}^d N_j$ is sufficiently large and that the following conditions hold.
    \begin{itemize}
        \item (Density Condition) $|A| \geq (d\log \Phi)^{C_0d^2} \cdot \sqrt{\mu \Phi}$.

        \item (Spreadness Condition) $|A \setminus \Lambda| \geq (d\log \Phi)^{C_0d^2} \cdot \frac{\mu\Phi}{|A| \cdot \det(Z)}$ for any sublattice $\Lambda \subsetneq Z$.

        \item (Nondegeneracy Condition) $|A \setminus \mathcal{V}| \geq (d\log \Phi)^{C_0d^2} \cdot \frac{\mu \Phi}{|A|}$ for any lower-dimensional subspace $\mathcal{V} \subseteq \mathbb{R}^d$.
    \end{itemize}
    Then
    \[
        \mathcal{S}(A) \supseteq \vec{v}_0 +  \{z_1s\vec{e}_1 + \cdots + z_ds\vec{e}_d : 0 \leq z_j \leq N_j, z_j\in \mathbb{Z}\},
    \]
    where $\vec{v}_0 \in Z$ and $s:= \det(Z)$.
\end{theorem}

\begin{theorem}[Formal version of Theorem~\ref{thm:john-type-informal}]\label{thm:john-type-extended}
    There exists a constant $C_0$ such that the following holds. Let $A$ be a multiset of vectors from  $[-N_1, N_1] \times \cdots \times [-N_d, N_d]$, where $Z := \mathcal{L}(A)$ is full-dimensional. Let $\mu:=\mu(A)$. Assume that $\Phi := \prod_{j=1}^d N_j$ is sufficiently large and that the  following conditions hold.
    \begin{itemize}
        \item (Density Condition) $|A| \geq (d\log \Phi)^{C_0d^2} \cdot \sqrt{\frac{\mu \Phi}{\det(Z)}}$.

        \item (Spreadness Condition) $|A \setminus \Lambda| \geq (d\log \Phi)^{C_0d^2} \cdot \frac{\mu\Phi}{|A| \cdot \det(Z)}$ for any sublattice $\Lambda \subsetneq Z$.

        \item (Nondegeneracy Condition) for some $n'\geq (d\log \Phi)^{C_0d^2} \cdot \frac{\mu \Phi}{|A|\cdot \det(Z)}$, it holds that $|A \setminus \mathcal{V}| \geq n' $ for any lower-dimensional subspace $\mathcal{V} \subseteq \mathbb{R}^d$.
    \end{itemize}
    Then 
    \[
         Z \cap \left((1 - \varepsilon) \mathcal{P}(A) + \frac{\sigma(A)}{2}\right) \subseteq \mathcal{S}(A) \subseteq Z \cap  \left(\mathcal{P}(A) + \frac{\sigma(A)}{2}\right),
    \]
    where $\varepsilon \leq \frac{(d\log \Phi)^{C_0d^2}}{n'|A|} \cdot \frac{\mu\Phi}{\det(Z)}$.
\end{theorem}

\begin{theorem}[Formal version of Theorem~\ref{thm:dense-alg-informal}]\label{thm:dense-alg-extended}
   There exists a constant $C_0$ such that the following holds. Let $A$ be a multiset of vectors from  $[-N_1, N_1] \times \cdots \times [-N_d, N_d]$, where $Z := \mathcal{L}(A)$ is full-dimensional. Let $\mu:=\mu(A)$. Assume that $\Phi := \prod_{j=1}^d N_j$ is sufficiently large and that the following conditions hold.
    \begin{itemize}
        \item (Density Condition) $|A| \geq (d\log \Phi)^{C_0d^2} \cdot \sqrt{\frac{\mu \Phi}{\det(Z)}}$.

        \item (Nondegeneracy Condition) for some $n'\geq (d\log \Phi)^{C_0d^2} \cdot \frac{\mu \Phi}{|A|\cdot \det(Z)}$, it holds that $|A \setminus \mathcal{V}| \geq n' $ for any lower-dimensional subspace $\mathcal{V} \subseteq \mathbb{R}^d$.
    \end{itemize}
    Then for any 
    \[
        \vec{t} \in \left((1- \varepsilon)\mathcal{P}(A) + \frac{\sigma(A)}{2}\right)
    \]
    where $\varepsilon \leq \frac{(d\log \Phi)^{C_0d^2}}{n'|A|} \cdot \frac{\mu\Phi}{\det(Z)}$, we can decide whether $\vec{t} \in \mathcal{S}(A)$ or not in $\widetilde{O}(n \cdot d^{d}\cdot \mathrm{poly}(d))$ time, with error probability at most $\Phi^{-\Omega(d)}$. 
\end{theorem}

\section{Generating Disjoint Subsets with Equal Sum}\label{sec:equal-sum}
We shall prove Lemma~\ref{lem:hd-es-petals}, which can be regarded as a generalization of the following theorem by Erd{\H{o}}s and S{\'a}rk{\"o}zy.

\begin{theorem}[{\cite[in the proof of Theorem~1]{ES92}}]\label{thm:es-petals}
    Let $A$ be a multiset of integers from $[N]$, where $N$ is sufficiently large. Then there exist at least $\frac{|A|}{18(\log N)^2}$ disjoint subsets of $A$ such that $1\leq |A_1| = \cdots = |A_m| \leq \log N$ and $\sigma(A_1) = \cdots = \sigma(A_m)$.
\end{theorem}

To generalize the above theorem to multi-dimension, the rough idea is to iteratively apply Theorem~\ref{thm:es-petals} to each dimension.

\lemhdespetals*
\begin{proof}
    We prove the lemma by induction on $d$. When $d = 1$, the lemma holds by Theorem~\ref{thm:es-petals}. Now assume that the lemma holds for $d$. We show that it also holds for $d+1$. 

    Let $A$ be a multiset of vectors from $[N]^{d+1}$. We ignore the $(d+1)$-th dimension, regard the vectors in $A$ as $d$-dimensional vectors, and apply the inductive hypothesis. This gives at least $|A|/(d\log N)^{3d}$ disjoint subsets $A_1, \ldots, A_{g}$ of $A$ such that $\sigma(A_1), \ldots, \sigma(A_{g})$ have the same $j$-th coordinates for each $j \in \{1, \ldots, d\}$ and that $1 \leq |A_1| = \cdots = |A_g| \leq h$, where $h := (d\log N)^{2d}$. 

    Let $\vec{b}_i = \sigma(A_i)$ for $i \in \{1, \ldots, g\}$. Let $B = \{\vec{b}_1, \ldots, \vec{b}_{g}\}$. Note that the vectors in $B$ differ only in the $(d+1)$-th coordinate, and their $(d+1)$-th coordinates are in $[0, hN]$. Consider only the $(d+1)$-th coordinate of $B$ and apply Theorem~\ref{thm:es-petals}. We can obtain at least $\frac{|B|}{18(\log hN)^2}$ subsets $B_1, \ldots, B_{g'}$ such that the $(d+1)$-th coordinates of $\sigma(B_1), \ldots, \sigma(B_{g'})$ are the same and that $1 \leq |B_1| = \cdots = |B_{g'}| \leq h'$, where $h' := \log(hN)$. 

    Since $|B_1| = \cdots = |B_{g'}|$ and the vectors in $B$ have the same $j$-th coordinate for each $j \in \{1, \ldots, d\}$, we have that $\sigma(B_1), \ldots, \sigma(B_{g'})$ also have the same $j$-th coordinate for each $j \in \{1, \ldots, d\}$. Therefore,
    \[
        \sigma(B_1) = \cdots = \sigma(B_{g'}).
    \]
    It is easy to see that $B_1, \ldots, B_{g'}$ correspond to $g'$ disjoint subsets of $A$ that have the same sum and the same cardinality. Moreover, the cardinality of these subsets of $A$ is at most $hh'$.  It remains to bound $hh'$ and $g'$. Since $N$ is sufficiently large, we have
    \[
        \log(hN) = (2d\log d + 2d\log\log N + \log N) \leq (d\log N)^{1.1}.
    \]
    Then 
    \begin{align*}
        hh' &=  h \log (hN) \leq  (d\log N)^{2d} \cdot (d\log N)^{1.1} \leq ((d+1)\log N)^{2d+2}.
    \end{align*}
    The last inequality holds because $N$ is sufficiently large. We also have that 
    \[
        g' \geq \frac{g}{18(\log hN)^2} \geq \frac{|A|}{(d\log N)^{3d}}\cdot \frac{1}{18(d\log N)^{2.2}}
             \geq \frac{|A|}{((d+1)\log N)^{3d+3}}.
    \]
    This completes the inductive step.
\end{proof}

\section{Generating Remainders}\label{sec:remainder}
We shall prove Lemma~\ref{lem:gen-rem-dense}. The following lemma explains how we will use a non-zero vector $\vec{h}$ in $\mathrm{Sym}_{\Gamma}(\mathcal{S}(A))$. Basically, instead of discussing $\mathcal{S}_{\Gamma}(A)$, we can discuss $\mathcal{S}_{\Gamma'}(A)$, where $\Gamma' := \Gamma + \mathcal{L}(h)$. Note that $\Gamma' \supsetneq \Gamma$, and therefore $\det(\Gamma') \leq \det(\Gamma)/2$. So such a $\vec{h}$ helps us reduce to a lattice with smaller determinant.

\begin{lemma}\label{lem:reduce-by-stablizer}
    Let $A \subseteq \mathbb{Z}^d$ be a set of vectors. Let $\Gamma$ be a full-dimensional lattice. Let $\vec{h}$ be a vector such that $\vec{h} + \mathcal{S}(A) \equiv \mathcal{S}(A) \pmod \Gamma$. Then 
    \[
        \mathcal{S}_{\Gamma}(A) = (\mathcal{S}_{\Gamma'}(A) + \Gamma') \bmod \Gamma,
    \]
    where $\Gamma' := \Gamma + \mathcal{L}(\vec{h})$.
\end{lemma}
\begin{proof}
    Since $\vec{h} + \mathcal{S}(A) \equiv \mathcal{S}(A) \pmod \Gamma$, we have
    \[
        \mathcal{L}(\vec{h}) + \mathcal{S}(A) \equiv \mathcal{S}(A) \pmod \Gamma,
    \]
    or in equivalent form,
    \[
        \Gamma + \mathcal{L}(\vec{h}) + \mathcal{S}(A) \equiv \mathcal{S}(A)  \pmod \Gamma.
    \]
    Note that $\Gamma + \mathcal{L}(\vec{h}) = \Gamma'$ and that $\Gamma' + \mathcal{S}_{\Gamma'}(A) = \Gamma' + \mathcal{S}(A)$. Therefore,
    \[
       \Gamma' + \mathcal{S}_{\Gamma'}(A) \equiv \mathcal{S}(A)  \pmod \Gamma. \qedhere
    \]
\end{proof}

The following lemma is a consequence of Kneser's theorem and Lemma~\ref{lem:reduce-by-stablizer}. Assuming $Z = \mathbb{Z}^d$ for simplicity, it states that to generate all the remainders modulo a lattice $\Gamma$, we need roughly $O(\det(\Gamma))$ vectors.

\begin{lemma}\label{lem:gen-rem-std}
    Let $A$ be a set of vectors from $\mathbb{Z}^d$, where $Z:=\mathcal{L}(A)$ is full-dimensional. Let $\Gamma \subseteq Z$ be a full-dimensional lattice, and let $\ell := \det(\Gamma)/\det(Z)$. If $|A \setminus \Lambda| \geq \ell$ for any lattice $\Gamma \subseteq \Lambda \subsetneq Z$, then $A$ contains a subset $R$ of at most $2\ell$ vectors such that $\mathcal{S}_{\Gamma}(R) = Z \bmod \Gamma$.
\end{lemma}
\begin{proof}
    Note that $\ell$ must be a positive integer since $\Gamma \subseteq Z$ and $\Gamma$ is full-dimensional. We prove the lemma by induction on $\ell$.  When $\ell = 1$, we have that $\Gamma = Z$, and the lemma holds trivially. Suppose that the lemma holds for all $\ell$ with $\ell \leq k$ for some $k \geq 1$. We show that it also holds when $\ell = k+1$.

    Since $\ell = k+1 > 1$, it must be that $\Gamma \subsetneq Z$. Therefore, $A$ contains at least $\ell$ vectors not belonging to $\Gamma$. Let $R_1 = \{\vec{a}_1, \ldots, \vec{a}_{\ell}\}$ be a set of $\ell$ such vectors. 
    It is easy to see that
    \[
        |\{0, \vec{a}_1\}| + \cdots + |\{0, \vec{a}_{\ell}\}| = 2\ell \geq |\mathcal{S}_{\Gamma}(R_1)| + \ell.
    \]
    The last inequality is due to that $|\mathcal{S}_{\Gamma}(R_1)| \leq |Z \bmod \Gamma| = \ell$. By Corollary~\ref{coro:kneser-2}, there is a vector $\vec{h} \notin \Gamma$ such that 
    \begin{equation}\label{eq:gen-rem-std-1}
        \vec{h} + \mathcal{S}(R_1) \equiv \mathcal{S}(R_1) \pmod {\Gamma}.
    \end{equation}

    Let $\Gamma' := \Gamma + \mathcal{L}(\vec{h})$ and $\ell' := \det(\Gamma')/\det(Z)$. We show that the inductive hypothesis is applicable to $\Gamma'$. Since $\vec{0} \in \mathcal{S}(R_1)$, Equation~\eqref{eq:gen-rem-std-1} implies that $\vec{h} \in \mathcal{S}(R_1) + \Gamma$. Since $R_1 \subseteq Z$ and $\Gamma \subseteq Z$, we have $\vec{h} \in Z$, and therefore, $\Gamma' \subseteq Z$. Since $\vec{h} \notin \Gamma$, we have $\Gamma \subsetneq \Gamma'$. So $\ell' \leq \ell/2 \leq k$. Also, for any $\Lambda$ with $\Gamma' \subseteq \Lambda \subsetneq Z$, we have $\Gamma \subseteq \Lambda$, and therefore,
    \[
        |A \setminus \Lambda| \geq \ell > \ell'.
    \]
    By the inductive hypothesis,
    $A$ contains a subset $R_2$ of at most $2\ell' \leq \ell$ vectors such that 
    \begin{equation}\label{eq:gen-rem-std-2}
        \mathcal{S}_{\Gamma'}(R_2) = Z \bmod \Gamma'.  
    \end{equation} 

    Let $R = R_1 \cup R_2$. Clearly, $|R| \leq |R_1| + |R_2| \leq 2\ell$. It remains to show that $\mathcal{S}_{\Gamma}(R) = Z \bmod \Gamma$. Equation~\eqref{eq:gen-rem-std-1} implies that 
    \(
        \vec{h} + \mathcal{S}(R) \equiv \mathcal{S}(R) \pmod {\Gamma},
    \)
    and Equation~\eqref{eq:gen-rem-std-2} implies that $\mathcal{S}_{\Gamma'}(R) = Z \bmod \Gamma'$. By Lemma~\ref{lem:reduce-by-stablizer}, we have $\mathcal{S}_{\Gamma}(R) = Z \bmod \Gamma$.
\end{proof}

The above lemma is the best one can do in general. But in the dense case, where $A$ contains many distinct remainders modulo $\Gamma$, one can do much better. We will show that in the dense case, we only need $O(\frac{\det(\Gamma)}{|A|})$ to generate all the remainders modulo $\Gamma$. To do that, we need the following densification technique by Szemer{\'e}di and Vu~\cite{SV06a}. Using this technique, we can create a set of cardinality $|A|$ using roughly $\log |A|$ vectors from $A$. Repeating the process, we can obtain roughly $|A|/\log |A|$ sets of cardinality $|A|$, and therefore, increase the total cardinality from $|A|$ to roughly $|A|^2$. As a consequence, Kneser's theorem (more precisely, Corollary~\ref{coro:kneser-2}) can be applied to these sets, even when $|A|$ is not large enough.

\begin{lemma}[{Implied by Lemma 7.9 in~\cite{SV06a}}]\label{lem:sz-set}
    Let $A \subseteq \mathbb{Z}^d$ be a finite set of vectors where $|A|$ is sufficiently large. Then $A$ contains a subset $A'$ of at most $20 \log |A|$ vectors such that $|\mathcal{S}(A')| > |A|$.
\end{lemma}

We generalize their technique to the finite group $\mathbb{Z}^d_{\Gamma}$.

\begin{lemma}\label{lem:sz-lattice}
    Let $\Gamma \subseteq \mathbb{Z}^d$ be a full-dimensional lattice. Let $A \subseteq \mathbb{Z}^d_\Gamma$ be a set of vectors where $|A|$ is sufficiently large. Then $A$ contains a subset $A'$ of at most $20 \log |A|$ vectors such that 
    \[
        |\mathcal{S}_{\Gamma}(A')| > \frac{|A|}{20^{d}(\log |A|)^d}.
    \]
\end{lemma}
\begin{proof}
    By Lemma~\ref{lem:sz-set}, $A$ contains a subset $A'$ of at most $20 \log |A|$ vectors such that $|\mathcal{S}(A')| > |A|$. To prove the lemma, it suffices to show that 
    \[
        |\mathcal{S}_{\Gamma}(A')| \geq \frac{|\mathcal{S}(A')|}{20^{d}(\log |A|)^d}.
    \]
    Let $k := \lfloor 20\log |A| \rfloor$. Since $A' \subseteq A \subseteq \mathbb{Z}^d_\Gamma$ and $|A'| \leq  k$, we have that 
    \(
        \mathcal{S}(A') \subseteq k\mathbb{Z}^d_{\Gamma}.
    \)
    Therefore, each congruence class (modulo $\Gamma$) of $\mathcal{S}(A')$ contains at most $k^d$ vectors from $\mathcal{S}(A')$. This implies that 
    \[
        |\mathcal{S}_{\Gamma}(A')| \geq \frac{|\mathcal{S}(A')|}{k^d} \geq \frac{|\mathcal{S}(A')|}{20^{d}(\log |A|)^d}. \qedhere
    \]
\end{proof}

\begin{lemma}\label{lem:densify}
    Let $\Gamma \subseteq \mathbb{Z}^d$ be a full-dimensional lattice. Let $A \subseteq \mathbb{Z}^d_\Gamma$ be a set of vectors where $|A|$ is sufficiently large. For any $\ell$ with $|A| \leq \ell \leq |A|^2/(\log |A|)^{3d}$, $A$ contains disjoint subsets $A_1, \ldots, A_m$ of at most $\frac{\ell}{2|A|}(\log |A|)^{3d}$ vectors in total such that 
    \begin{equation}\label{eq:lemdensify-0}
            \sum_{i=1}^m |\mathcal{S}_{\Gamma}(A_i)| \geq \ell + m.
    \end{equation}
\end{lemma}
\begin{proof}
    We iteratively extract subsets $A_1, \ldots, A_m$ from $A$ as follows. Initially, $m := 0$ and every vector in $A$ is unmarked. If $A$ contains fewer than $|A|/2$ unmarked vectors or~\eqref{eq:lemdensify-0} is satisfied, then we stop. Otherwise, we pick any $|A|/2$ unmarked vectors, apply Lemma~\ref{lem:sz-lattice} to them, and let $A_{m+1}$ be the resulting subset. Then we mark all the vectors in $A_{m+1}$ and repeat with $m := m+1$.

    It is easy to see that $A_1, \ldots, A_m$ are disjoint.  It remains to bound their total size and show that~\eqref{eq:lemdensify-0} holds. By Lemma~\ref{lem:sz-lattice}, we have 
    \begin{equation}\label{eq:lemdensify-1}
        |A_i| \leq 20 \log (|A|/2) \leq 20 \log |A|
    \end{equation}
    and
    \begin{equation}\label{eq:lemdensify-2}
        |\mathcal{S}(A_i)| \geq \frac{|A|}{2 \cdot 20^d (\log |A|)^d} \geq \frac{80|A|}{(\log |A|)^{2d}} + 1.
    \end{equation}
    The last inequality is due to that $|A|$ is sufficiently large. In view of~\eqref{eq:lemdensify-2}, when the construction of $A_1, \ldots, A_m$ terminates, we have 
    \begin{equation}\label{eq:lemdensify-3}
        m \leq \left\lceil \frac{\ell}{80|A|}(\log |A|)^{2d} \right\rceil \leq \frac{\ell}{40|A|}(\log |A|)^{2d}.
    \end{equation}
    In view of~\eqref{eq:lemdensify-1} and~\eqref{eq:lemdensify-3}, we have
    \begin{equation}\label{eq:lemdensify-4}
        \sum_{i=1}^m |A_i| \leq 20 \log |A| \cdot \frac{\ell}{40|A|}(\log |A|)^{2d} = \frac{\ell}{2|A|}(\log |A|)^{3d}.
    \end{equation}
    Recall that $\ell \leq |A|^2/(\log |A|)^{3d}$. Inequality~\eqref{eq:lemdensify-4} implies that $A$ always has at least $|A|/2$ unmarked elements during the construction of $A_1, \ldots, A_m$. Therefore,~\eqref{eq:lemdensify-0} must hold when the construction terminates.
\end{proof}

Now we are ready to prove Lemma~\ref{lem:gen-rem-dense}.

\lemgenremdense*
\begin{proof} 
    Let $c$ be a sufficiently large constant. Note that $\ell$ must be a positive integer since $\Gamma \subseteq Z$ and $\Gamma$ is full-dimensional. We prove, by induction on $\ell$, a stronger lemma in which the bound on $|R|$ is replaced by $\frac{c\mu\ell}{|A'|}(\log (\mu\ell))^{3d}\log \ell$. 

    When $\ell = 1$, we have that $\Gamma = Z$, and the lemma holds trivially. Suppose that the lemma holds for all $\ell$ with $\ell \leq k$ for some $k \geq 1$. We show that it also holds when $\ell = k+1$.

    If $c \mu > |A'|$, the lemma follows directly from Lemma~\ref{lem:gen-rem-std}.  Assume that $c\mu \leq |A'|$. 

    \begin{claim}
        $A'$ contains a subset $R_1$ of at most $\frac{16\mu\ell}{|A'|}(\log (\mu\ell))^{3d}$ elements such that for some $\vec{h} \notin \Gamma$,
        \begin{equation}\label{eq:gen-rem-dense-1}
            \vec{h} + \mathcal{S}(R_1) \equiv \mathcal{S}(R_1) \pmod \Gamma.
        \end{equation}
    \end{claim}
    \begin{proof}
        Let $n : = |A'|$.  We construct $\mu$ disjoint subsets $A_1, \ldots, A_{\mu}$ of $A'$ such that $|A_i| = |A_i \bmod \Gamma| = \lceil \frac{n}{4\mu} \rceil$ for each $i$ as follows. Initially, all elements in $A'$ are unmarked. For $i = 1, \ldots, \mu$, find $\lceil \frac{n}{4\mu}\rceil$ unmarked elements in $A'$ that are not congruent to each other, mark them and put them into $A_i$. We show that the construction always succeeds. Recall that $c\mu\leq n$. Each $A_i$ contains $\lceil \frac{n}{4\mu}\rceil \leq \frac{n}{2\mu}$ elements. Therefore, at the beginning of each round, we have at least $\frac{n}{2}$ unmarked elements. By the definition of $\mu$, there are at least $\frac{n}{2\mu} \geq \lceil \frac{n}{4\mu}\rceil$ unmarked elements not congruent to each other. Note that each $A_i$ has cardinality $\lceil \frac{n}{4\mu}\rceil \geq \frac{n}{4\mu} \geq \frac{c}{4}$. Thus, $|A_i|$ is sufficiently large, and Lemma~\ref{lem:densify} is applicable to $A_i$ (in fact, $A_i \bmod \Gamma$). 
        
        Define 
        \[
            g: = \frac{n^2}{16\mu^2(\log (\mu\ell))^{3d}}.
        \]
        We consider the first case where $\ell \leq g$. By definition of $\mu$, we have $\ell \geq n/\mu > |A_1|$.  One can also verify that in this case 
        \(
            \ell \leq \frac{|A_1|^2}{(\log |A_1|)^{3d}}.
        \)
        Lemma~\ref{lem:densify} is applicable to $A_1$ and $\ell$, and therefore, $A_1$ contains disjoint subsets $A_{1,1}, \ldots, A_{1,m_1}$ of at most $\frac{\ell}{2|A_1|}(\log |A_1|)^{3d}$ elements in total such that  
        \[
             \sum_{j=1}^{m_1} |\mathcal{S}_{\Gamma}(A_{1,j})| \geq \ell + m_1.
        \]
        Note that $|\mathcal{S}_{\Gamma}(R_1)| \leq \ell$. By Corollary~\ref{coro:kneser-2}, there exists a vector $\vec{h} \notin \Gamma$ such that $\vec{h} + \mathcal{S}(R_1) \equiv \mathcal{S}(R_1) \pmod \Gamma$. Moreover, 
        \[
            |R_1| \leq \sum_{j=1}^{m_1}|A_{1,j}| \leq \frac{\ell}{2|A_1|}(\log |A_1|)^{3d} \leq \frac{2\mu\ell}{n}(\log n)^{3d}.
        \]

        Now we consider the second case where $\ell > g$.
        By Lemma~\ref{lem:densify}, each $A_i$ contains disjoint subsets $A_{i,1}, \ldots, A_{i, m_i}$ such that 
        \[
            \sum_{j=1}^{m_i} |\mathcal{S}_{\Gamma}(A_{i,j})| \geq \frac{|A_i|^2}{(\log |A_i|)^{3d}} + m_i \geq g + m_i.
        \]
        In view of~\eqref{eq:gen-rem-dense-0}, we have 
        \[
            \left\lceil\frac{\ell}{g}\right\rceil \leq  \frac{32\ell\mu^2(\log (\mu\ell))^{3d}}{n^2} \leq \mu.
        \]
        Summing the above inequality over all $i$ with $1 \leq i \leq \lceil\frac{\ell}{g}\rceil$, we have
        \[
            \sum_{i=1}^{\lceil\frac{\ell}{g}\rceil}\sum_{j=1}^{m_i} |\mathcal{S}_{\Gamma}(A_{i,j})| \geq \ell + \sum_{i=1}^{\lceil\frac{\ell}{g}\rceil}m_i.
        \]
        Let $R_1 = \bigcup_{i=1}^{\lceil\frac{\ell}{g}\rceil}\bigcup_{j=1}^{m_i} A_{i,j}$. Note that $\mathcal{S}_{\Gamma}(R_1) \leq \ell$. By Corollary~\ref{coro:kneser-2}, there exists a vector $\vec{h} \notin \Gamma$ such that $\vec{h} + \mathcal{S}(R_1) \equiv \mathcal{S}(R_1) \pmod \Gamma$. Moreover, 
        \[
            |R_1| \leq \sum_{i=1}^{\lceil\frac{\ell}{g}\rceil}|A_{i}| \leq \frac{2\ell}{g} \cdot \frac{n}{2\mu} = \frac{16\mu\ell}{n} (\log (\mu\ell))^{3d}. \qedhere
        \]
    \end{proof}

    Let $\Gamma' = \Gamma + \mathcal{L}(\vec{h})$ and $\ell' := \det(\Gamma')/\det(Z)$ and $\mu' := \mu(A', \Gamma')$. We show that the inductive hypothesis is applicable to $\Gamma'$. Since $\vec{0} \in \mathcal{S}(R_1)$, Equation~\eqref{eq:gen-rem-dense-1} implies that $\vec{h} \in \mathcal{S}(R_1) + \Gamma$. Since $R_1 \subseteq Z$ and $\Gamma \subseteq Z$, we have $\vec{h} \in Z$, and therefore, $\Gamma' \subseteq Z$. Since $\vec{h} \notin \Gamma$, we have $\Gamma \subsetneq \Gamma'$. So $\ell' \leq \ell/2 \leq k$. We also have that
    \[
        \frac{\mu'}{\mu} \leq \frac{\det(\Gamma)}{\det(\Gamma')} = \frac{\ell}{\ell'},
    \]
    which implies $\ell'\mu'\leq \ell \mu \leq n^2/(\log n)^{3d}$. For any $\Lambda$ with $\Gamma' \subseteq \Lambda \subsetneq Z$, we have $\Gamma \subseteq \Lambda$, and therefore,
    \[
        |A \setminus \Lambda| \geq c\mu\ell/n \geq c\mu'\ell'/n.
    \]
    By the inductive hypothesis, $A$ contains a subset $R_2$ of at most $\frac{c\mu'\ell'}{n} (\log (\mu'\ell'))^{3d} \log \ell'$ vectors such that 
    \begin{equation}\label{eq:gen-rem-dense-2}
        \mathcal{S}_{\Gamma'}(R_2) = Z \bmod \Gamma'.
    \end{equation}
    Let $R = R_1 \cup R_2$. We have
    \begin{align*}
        |R| \leq |R_1| + |R_2| &= \frac{16\mu\ell}{n} (\log (\mu\ell))^{3d} + \frac{c\mu'\ell'}{n} (\log (\mu'\ell'))^{3d} \log \ell' \\
            & \leq \frac{c\mu\ell}{n} (\log (\mu\ell))^{3d} (1 + \log \ell') \leq \frac{c\mu\ell}{n} (\log (\mu\ell))^{3d}\log \ell.
    \end{align*}
    Equation~\eqref{eq:gen-rem-dense-1} implies that 
    \(
        \vec{h} + \mathcal{S}(R) \equiv \mathcal{S}(R) \pmod {\Gamma},
    \)
    and Equation~\eqref{eq:gen-rem-dense-2} implies that $\mathcal{S}_{\Gamma'}(R) = Z \bmod \Gamma'$. By Lemma~\ref{lem:reduce-by-stablizer}, we have $\mathcal{S}_{\Gamma}(R) = Z \bmod \Gamma$.
\end{proof}

\section{Computing Modular Sumsets in Multi-dimension}\label{sec:modular-sumset}
We shall prove Lemma~\ref{lem:faster-modular-ss}. We solve a more general problem of computing $(S_1 + \cdots + S_m) \bmod \Gamma$ for some full-dimensional lattice $\Gamma$.

The modular sumset of two sets can be computed efficiently using fast Fourier transform.

\begin{lemma}\label{lem:modulo-sumset}
    Let $\Gamma \subseteq \mathbb{Z}^d$ be a full-dimensional lattice with $\ell: = \det(\Gamma)$. Let $S_1, S_2 \subseteq \mathbb{Z}_{\Gamma}^d$ be two sets of vectors.  We can compute 
    \[
        S:= (S_1 + S_2) \bmod \Gamma
    \]  
    in $O(2^{d}|S| \cdot \mathrm{polylog}(2^d \ell))$ time.
\end{lemma}
\begin{proof}
    Recall that 
    \(
       S_1, S_2 \subseteq [b_{11}] \times \cdots \times [b_{dd}]
    \),
    where $b_{11}, \ldots, b_{dd}$ are the diagonal elements of the basis of $\Gamma$ in Hermite normal form. Note that
    \(
        \prod_{i=1}^d b_{ii} = \ell.
    \)
    It follows that 
    \[
        S_1+S_2\subseteq [2b_{11}] \times \cdots \times [2b_{dd}].
    \]
    Each residue class modulo $\Gamma$ contains at most $2^d$ vectors in $[2b_{11}] \times \cdots \times [2b_{dd}]$. Therefore, 
    \[
        |S_1+S_2|\leq 2^d |S|.
    \]

    Each vector in $[2b_{11}] \times \cdots \times [2b_{dd}]$ can be mapped to a unique integer in $\mathbb{Z}$ in the range $[\prod_{i=1}^{d} 2b_{ii}] = [2^d\ell]$.  
    Then we can compute $S_1 + S_2$ using algorithms for sparse convolution~\cite{BFN22} in time
    \[
        O(|S_1 + S_2| \cdot \mathrm{polylog}(2^d \ell)) = O(2^d |S| \cdot \mathrm{polylog}(2^d \ell)).
    \]
    Given $S_1 + S_2$, we can obtain
    \(
        S
    \)
    in $O(|S_1 + S_2| d^2)=O(2^d |S|\cdot d^2)$ time via Lemma~\ref{lem:compute-rem}(i). One can verify that the total time cost is bounded by
    \(
        O(2^d |S| \cdot \mathrm{polylog}(2^d\ell)).
    \)
\end{proof}

Recall Definition~\ref{def:symset}. Next we will present an efficient algorithm to compute the symmetry set. Below we give a closed form of the symmetry set.

\begin{lemma}\label{lem:symset-formula}
    Let $S \subseteq \mathbb{Z}^d$ be a set of vectors. Let $\Gamma$ be a full-dimensional lattice. Then we have
    \[
        \mathrm{Sym}_{\Gamma}(S) \equiv (S - S) \setminus (((S - S) + S )\setminus S - S) \pmod \Gamma.
    \]
\end{lemma}
\begin{proof}
     Let $H:=\mathrm{Sym}_{\Gamma}(S)$. For simplicity of notation, in this proof, all the operations $+$ and $-$ are modulo $\Gamma$. That is, they are group operations of the additive group $\mathbb{Z}^d_{\Gamma}$. We first show that $H \subseteq S - S$. For any $\vec{h} \in H$, by definition, we have
        \(
            \vec{h} + S = S. 
        \)
        Therefore, there exists $\vec{s}_1,\vec{s}_2 \in S$ such that 
        \(
            \vec{h} + \vec{s}_1 = \vec{s}_2.
        \) 
        This implies
        \(
            \vec{h} \in S - S.
        \)
    Therefore, we have $H \subseteq S - S$. 

    Let $X := (S - S) \setminus H$.
    \begin{claim}\label{clm:compute-stablizer-2}
        Let $\vec{v} \in S - S$ be a vector. $\vec{v} \in X$ if and only if
        \[
           \vec{v} \in ((X + S) \setminus S) - S.
        \]
    \end{claim}
    \begin{proof}
        We first prove the only-if part. Assume that $\vec{v} \in X$. Since $\vec{v}\notin H$, we have
        \[
            \vec{v} + S \neq S.
        \]
        This implies that $(\{\vec{v}\} + S) \setminus S \neq \emptyset$ and therefore
        \[
            \vec{v} \in ((\{\vec{v}\} + S) \setminus S) - S.
        \] 
        Note that $((\{\vec{v}\} + S) \setminus S) - S \subseteq ((X + S) \setminus S) - S$. This completes the proof of the only-if part.

        Next we prove the if part. For $\vec{v} \in H$, we have
        \(
            \{\vec{v}\} + S = S.
        \)
        Therefore,
        \[
            (\{\vec{v}\} + S) \cap ((X + S) \setminus S) = \emptyset.
        \]
        This implies that 
        \[
            \vec{v} \notin ((X + S) \setminus S) - S. \qedhere
        \]
    \end{proof}

    \begin{claim}\label{clm:compute-stablizer-3}
        $(X + S) \setminus S = ((S - S) + S) \setminus S$.
    \end{claim}
    \begin{proof}
        By the definition of $H$ and $X$, we have 
        \[
            (S - S) + S = (X \cup H) +  S = (X + S) \cup (H  +  S) = (X + S) \cup S.
        \]
        Removing the elements of $S$ from both sides of the above equation, we obtain the stated equation in the claim.
    \end{proof}

    Note that $H = (S - S)\setminus X$. Claims~\ref{clm:compute-stablizer-2} and~\ref{clm:compute-stablizer-3} imply that
    \[
        (S - S)\setminus X \equiv (S - S)\setminus ((X + S) \setminus S) - S) = (S - S) \setminus (((S - S) + S) \setminus S - S). \qedhere
    \]
\end{proof}

The following lemma follows directly from Lemmas~\ref{lem:modulo-sumset} and~\ref{lem:symset-formula}.
\begin{lemma}\label{lem:compute-stablizer}
    Let $\Gamma \subseteq \mathbb{Z}^d$ be a full-dimensional lattice and $\ell :=\det(\Gamma)$. Let $S \subseteq \mathbb{Z}^d_\Gamma$ be a set of vectors. We can compute $\mathrm{Sym}_{\Gamma}(S)$ in $O(2^d \ell\cdot \mathrm{polylog}(2^d \ell))$ time.
\end{lemma}

\begin{algorithm}
    \caption{$\mathtt{PARTIALSUM}(S_1, \ldots, S_m, \Gamma)$} 
    \label{alg:partialsum}
    \begin{algorithmic}[1]
            \Statex Assume that the operation $+$ is modulo $\Gamma$ and that $m$ is a power of $2$. Recall that $\ell := \det(\Gamma)$ \;
            \If{$m == 1$}
                \State \Return $S_1$\;
            \EndIf
        \State $t \leftarrow 0$ and $j \leftarrow 1$\;
        \While{$j \leq m/2$ and $t < \ell + m/2$}
            \State $S'_j \leftarrow (S_{2j-1} + S_{2j}) \bmod \Gamma$ via Lemma~\ref{lem:modulo-sumset}\;
            \State $t \leftarrow t + |S'_j|$\;
            \State $j \leftarrow j + 1$\;
        \EndWhile
        \State \Return $\mathtt{PARTIALSUM}(S'_1, \ldots, S'_{j-1}, \Gamma)$\;
    \end{algorithmic}
\end{algorithm}

Using Lemmas~\ref{lem:modulo-sumset} and~\ref{lem:compute-stablizer}, we can compute the $m$-fold modular sumset using an idea similar to that of Bringmann and Nakos~\cite{BN21}. Basically, we can either compute $(S_1 + \ldots + S_m) \bmod \Gamma$ or a non-zero vector $\vec{h}\in\mathrm{Sym}_{\Gamma}(S_1 + \ldots + S_m)$. In the latter case, we can use $\vec{h}$ to reduce the problem to a smaller space $Z_{\Gamma'}^d$, where $\Gamma' \supsetneq \Gamma$.

\begin{lemma}\label{lem:sumset-or-stablizer}
    Let $\Gamma \subseteq \mathbb{Z}^d$ be a full-dimensional lattice and $\ell:=\det(\Gamma)$. Let $S_1, \ldots, S_m \subseteq \mathbb{Z}^d_{\Gamma}$ be $m$ sets of vectors with total cardinality $n$.  In $O(n + 2^d \ell\cdot \mathrm{polylog}(2^d \ell))$ time, we can obtain either $(S_1 + \cdots + S_m) \bmod \Gamma$ or a non-zero vector $\vec{h}\in\mathrm{Sym}_{\Gamma}(S_1 + \cdots + S_m)$.
\end{lemma}
\begin{proof}
    For simplicity of notation, in this proof, all the operations $+$ and $-$ are modulo $\Gamma$. That is, they are group operations of the additive group $\mathbb{Z}^d_{\Gamma}$. We also assume that $m$ is a power of $2$. This can be done by adding dummy sets of the form $\{\vec{0}\}$.

    We shall invoke Algorithm~\ref{alg:partialsum}. Basically, Algorithm~\ref{alg:partialsum} computes $S_1 + \cdots + S_m$ recursively as follows.  It computes $S'_j = S_{2j-1} + S_{2j}$ via Lemma~\ref{lem:modulo-sumset} for $j \in \{1, \ldots, m/2\}$, and then recursively computes $S'_1 + \cdots + S'_{m/2}$. The trick is that when $|S'_1| + \cdots + |S'_j| \geq \ell + m/2$ for some $j$, it stops computing the remaining sets (because $\mathrm{Sym}_\Gamma(S'_1 + \cdots + S'_j)$ already contains a non-zero vector), and immediately starts the recursion on $S'_1, \ldots, S'_j$.

    \begin{claim}\label{clm:sumset-or-stablizer}
        Algorithm~\ref{alg:partialsum} returns $S_1 + \cdots + S_j$ for some $j \leq m$. Moreover, if $j < m$, then $\mathrm{Sym}_{\Gamma}(S_1 + \cdots + S_j)$ contains a non-zero vector.
    \end{claim}
    \begin{proof}
        If the criterion $t < \ell + m/2$ in line 4 is never violated, then Algorithm~\ref{alg:partialsum} will return $S_1 + \cdots + S_m$. If this criterion is violated in some recursion, then Algorithm~\ref{alg:partialsum} returns $S_1 + \cdots + S_j$ for some $j < m$. Since $S_1, \ldots, S_m \subseteq \mathbb Z^d_{\Gamma}$, we have $|S_1+\dots+S_j|\leq \ell$ for any $j\leq m$. It follows that $|S_1| + \cdots + |S_j| \geq \ell + m/2 \geq |S_1+\dots+S_j| + j$. By Corollary~\ref{coro:kneser-1}, there exists a non-zero vector $\vec{h} \in \mathrm{Sym}_{\Gamma}(S_1 + \cdots + S_j)$.
    \end{proof} 

    Let $S$ be the result of Algorithm~\ref{alg:partialsum}. Note that $S\subseteq S_1 + \cdots + S_m$. We can compute $\mathrm{Sym}_{\Gamma}(S)$ in $O(2^d \ell\cdot \log(2^d\ell))$ time via Lemma~\ref{lem:compute-stablizer}. Suppose $\mathrm{Sym}_{\Gamma}(S)$ contains a non-zero vector $\vec{h}$.
    Recall that if $\vec{h}\in \mathrm{Sym}_{\Gamma}(S)$, then $\vec{h}\in \mathrm{Sym}_{\Gamma}(S_1 + \cdots + S_m)$. In this case, we are done. Otherwise, if we obtain no such vector, then by Claim~\ref{clm:sumset-or-stablizer}, $S = S_1 + \cdots + S_m$.

    The recursion tree of Algorithm~\ref{alg:partialsum} has at most $\log m$ levels. Note that $|S'_j| \leq \mathbb |Z^d_{\Gamma}|= \ell$. In each level $i$, when the criterion $t < \ell + m/2^i$ is violated,
    the total cardinality (that is, $t$) is bounded by $\ell + m/2^i + \ell$. Therefore, the total time cost of applying Algorithm~\ref{alg:partialsum} and Lemma~\ref{lem:compute-stablizer} is 
    \begin{align}
        &O(n) + O(2^d\ell\log(2^d\ell)) + \sum_{i = 1}^{\log m} O\left( 2^d (2\ell + \frac{m}{2^i}) \cdot \mathrm{polylog}(2^d \ell)\right)\nonumber\\
        = &O\left( n+ 2^d(m + \ell\log m)\cdot \mathrm{polylog}(2^d \ell)\right).\label{eq:sumset-or-stablizer-1} 
    \end{align}

    Next, we show that this running time can be reduced to $O(n + 2^d \ell\cdot \mathrm{polylog}(2^d \ell))$ by dealing with the sets of cardinality $1$ separately. Suppose that $|S_i| \geq 2$ for all $i \leq m'$ and that $|S_i| = 1$ for all $i > m'$. 

    Suppose that $m' \leq \ell$. We can apply Algorithm~\ref{alg:partialsum} to $S_1, \ldots, S_{m'}$ and then Lemma~\ref{lem:compute-stablizer} to the result of Algorithm~\ref{alg:partialsum}. The total time cost of Algorithm~\ref{alg:partialsum} and Lemma~\ref{lem:compute-stablizer} is bounded by~(\ref{eq:sumset-or-stablizer-1}) with $m$ replaced by $\ell$, which is
    \begin{equation}\label{eq:sumset-or-stablizer-2}
        O\left(n + 2^d  \ell\cdot \mathrm{polylog}(2^d\ell)\right). 
    \end{equation}
    If we obtain a non-zero vector $\vec{h}\in \mathrm{Sym}_{\Gamma}(S_1+\dots+S_m)$, then the time cost is bounded by~(\ref{eq:sumset-or-stablizer-2}). Otherwise, we obtain $S = S_1 + \cdots + S_{m'}$. Since $|S_i| = 1$ for all $i > m'$, we can compute $S' = S_{m'+1} + \cdots  + S_{m}$ trivially in $O(n + d^2)$ time by Lemma~\ref{lem:compute-rem}. Then we compute $S + S'$, which takes $O(d^2 \ell)$ time since $|S|\leq \ell$ and $|S'| = 1$.

    Now consider the case where $m' \geq \ell$. We have that 
    \[
        \sum_{i=1}^{\ell} |S_i| \geq 2\ell\ge |S_1+\dots+S_{\ell}|+\ell.
    \]
    By Corollary~\ref{coro:kneser-1}, there exists a non-zero vector $\vec{h}\in\mathrm{Sym}_{\Gamma}(S_1 + \cdots + S_{\ell})$. Therefore, applying Algorithm~\ref{alg:partialsum} to $S_1, \ldots, S_{\ell}$ and then Lemma~\ref{lem:compute-stablizer} to the result of Algorithm~\ref{alg:partialsum} will return a non-zero vector $\vec{h}$, which is also in the symmetry set of $S_1 + \cdots + S_m$. The total time cost is bounded by 
    \(
        O\left(2^d \ell\cdot \mathrm{polylog}(2^d \ell)\right). 
    \)    
\end{proof}

\begin{algorithm}
    \caption{$\mathtt{MODULARSUMSET}(S_1, \ldots, S_m, \Gamma)$} 
    \label{alg:modular-sumset}
    \begin{algorithmic}[1]
        \If{$\det(\Gamma) == 1$}
            \State \Return $\{\vec{0}\}$\;
        \EndIf
        \State Invoke Lemma~\ref{lem:sumset-or-stablizer} on $S_1, \ldots, S_m, \Gamma$
        \If{Lemma~\ref{lem:sumset-or-stablizer} returns $(S_1 + \cdots + S_m) \bmod \Gamma$}
            \State \Return $(S_1 + \cdots + S_m) \bmod \Gamma$\;
        \Else
            \State Let $\vec{h}\in \mathrm{Sym}_{\Gamma}(S_1+\dots+S_m)$ be the vector given by Lemma~\ref{lem:sumset-or-stablizer}\;
            \State $\Gamma' \leftarrow \Gamma + \mathcal{L}(\vec{h})$\; \Comment{Compute the Hermite normal form via Lemma~\ref{lem:compute-hnf}}
            \For{$i = 1$ to $m$}
                \State $S'_i \leftarrow S_i \bmod \Gamma'$\;
            \EndFor
            \State $R \leftarrow \mathtt{MODULARSUMSET}(S'_1, \ldots, S'_m, \Gamma')$\;
            \State \Return $(R + (\Gamma' \bmod \Gamma)) \bmod \Gamma$\; \Comment{$\Gamma' \bmod \Gamma$ can be computed via Lemma~\ref{lem:compute-rem}(ii)}
        \EndIf
    \end{algorithmic}
\end{algorithm}

If Algorithm~\ref{alg:partialsum} returns $(S_1 + \cdots + S_m) \bmod \Gamma$, then we are done. Otherwise, it returns a non-zero vector $\vec{h}\in\mathrm{Sym}_{\Gamma}(S_1 + \cdots + S_m)$. Then by Lemma~\ref{lem:reduce-by-stablizer}, we can use $\vec{h}$ to reduce the computation to a lattice $\Gamma'$ with $\det(\Gamma') \leq \det(\Gamma)/2$.

\begin{theorem}\label{thm:modular-m-fold-sumset}
    Let $S_1, \ldots, S_m \subseteq \mathbb{Z}^d_{\Gamma}$ be $m$ sets of vectors with total cardinality $n$. Let $\Gamma \subseteq \mathbb{Z}^d$ be a full-dimensional lattice with $\ell := \det(\Gamma)$. We can compute $(S_1 + \cdots + S_m)\bmod \Gamma$ in time
    \[
        O(d^2n\log \ell + 2^d\ell\cdot \mathrm{polylog}(2^d \ell)).
    \] 
\end{theorem}
\begin{proof}
    We use Algorithm~\ref{alg:modular-sumset}. The correctness follows from Lemma~\ref{lem:reduce-by-stablizer}.

     Let $T(\ell)$ be the running time of the algorithm.

    Line~3 takes $O(n + 2^d \ell\cdot \mathrm{polylog}(2^d \ell))$ time. Line~8 takes $O(d^2 + d\log \ell)$ time by Lemma~\ref{lem:compute-rem}(iii). Lines~9-10 take $O(d^2n)$ time by Lemma~\ref{lem:compute-rem}(i).  
    Since $\vec{h}\in \mathrm{Sym}_{\Gamma}(S_1 + \cdots + S_m)$ is non-zero, we have $\Gamma' \supsetneq \Gamma$, and therefore $\det(\Gamma') \leq \det(\Gamma)/2$. By recursion, line~11 takes at most $T(\ell/2)$ time.  Computing $(\Gamma' \bmod \Gamma)$ takes $O(\ell d^3)$ time by Lemma~\ref{lem:compute-rem}(ii). So line~12 takes $O(2^d \ell\cdot \mathrm{polylog} (2^d \ell))$ time. In summary, we have $T(1) = O(1)$ and
    \[
        T(\ell) \leq T(\frac{\ell}{2}) + O(d^2n + 2^d \ell\cdot \mathrm{polylog}(2^d \ell)).
    \]
    Solving the recurrence, we have $T(\ell) = O(d^2n \log \ell + 2^d \ell \cdot \mathrm{polylog}(2^d \ell))$.
\end{proof}

Now we are ready to prove Lemma~\ref{lem:faster-modular-ss}.

\lemfastermodularss*
\begin{proof}
    Let $B = (\vec{b}_1, \ldots, \vec{b}_d)$ be the basis of $Z$ in Hermite normal form. Note that $B^{-1}B = I$, where $I$ is the identity matrix. Computing $B^{-1}$ takes only $\mathrm{poly}(d)$ time. $B^{-1}$ can be viewed as a linear transformation that maps $Z$ to $\mathbb{Z}^d$. Let $A' = B^{-1}A$. Let $\Gamma' = B^{-1}\Gamma$. It takes $O(d^2n)$ time to compute $A'$, and $\mathrm{poly}(d)$ time to compute (the basis of) $\Gamma'$. We have that $\det(\Gamma') = \det(\Gamma)/\det(Z) = \ell$. 

    We can compute $\mathcal{S}_{\Gamma'}(A')$ via Theorem~\ref{thm:modular-m-fold-sumset} in $O(d^2 n\log \ell + 2^d \ell\cdot \mathrm{polylog}(2^d \ell))$. Given $\mathcal{S}_{\Gamma'}(A')$, we can compute $\mathrm{Sym}_{\Gamma'}(\mathcal{S}(A'))$ via Lemma~\ref{lem:compute-stablizer} in $O(d^2 n\log \ell + 2^d \ell\cdot \mathrm{polylog}(2^d \ell))$.

    We can obtain $\mathcal{S}_{\Gamma}(A)$ and $\mathrm{Sym}_{\Gamma}(\mathcal{S}(A))$ from $\mathcal{S}_{\Gamma'}(A')$ and $\mathrm{Sym}_{\Gamma'}(\mathcal{S}(A'))$ by left-multiplying them by $B$, which takes $O(d^2\ell)$ time.

    One can verify that the total time cost is $\widetilde{O}((n + 2^d\ell) \cdot \mathrm{poly}(d))$.
\end{proof}

\appendix
\section{Sharpness of the Density Threshold}\label{apx:lb}

Let $A\subseteq [N_1]\times [N_2]\times\cdots\times [N_d]$. Let $\Phi=\prod_{i=1}^d N_i$. We show that if $|A|\leq (1-o(1)) \sqrt{\Phi}$, then $\mathcal{S}(A)$ may contain many ``holes'' even if $|A \setminus \Gamma| \geq |A|/2$ for any full-dimensional lattice $\Gamma \subsetneq \mathbb Z^d$. 

\begin{theorem}\label{thm:sharpness}
    There exists a set $A\subseteq [N_1]\times [N_2]\times\cdots\times [N_d]$ with
    \[
        |A| \leq (1-o(1))\sqrt{\Phi},
    \]
    where $\Phi = \prod_{j=1}^d N_j$, such that $|A \setminus \Gamma| \geq |A|/4$ for any full-dimensional lattice $\Gamma \subsetneq \mathbb Z^d$ and that there is no convex body $K$ with $K\cap \mathbb{Z}^d\subseteq \mathcal{S}(A)$ and $|K\cap \mathbb{Z}^d|=(1-o(1))|\mathcal{S}(A)|$.
\end{theorem}
\begin{proof} 
    Let $p_1, \ldots, p_d$ be distinct primes in 
    \(
        [\sqrt{N_i}, \sqrt{N_i} + N_i^{0.3}],
    \)
    which exist due to~\cite{BHP01}. Define
    \[
        A = \bigl\{(x_1 p_1 + 1, x_2 p_2 + 1, \dots, x_d p_d + 1) : 0 \leq x_i \leq \left\lfloor \frac{N_i}{p_i} \right\rfloor, \, x_i \in \mathbb{Z} \bigr\}.
    \] 
    One can verify that $|A \setminus \Gamma| \geq |A|/4$ for any lattice $\Gamma \subsetneq \mathbb Z^d$ and that
    \(
        |A| \leq (1-o(1))\sqrt{\Phi}.
    \)  

    Let $\Gamma \subseteq \mathbb{Z}^d$ be the lattice generated by $\{p_i \vec{e}_i : 1 \leq i \leq d\}$. Note that $\det(\Gamma) = \prod_{i=1}^{d} p_i$. We also have that 
    \(
        A \pmod{\Gamma} = \{\vec{1}\}.
    \)
    Hence,
    \(
        \mathcal{S}_\Gamma(A)
    \)
    contains at most 
    \[
        |A| \leq \frac{\prod_{i=1}^{d} N_i}{\prod_{i=1}^{d} p_i}\leq (1-o(1)) \prod_{i=1}^{d} p_i = (1-o(1))\det(\Gamma)
    \]
    distinct residue classes of $\mathbb{Z}^d_\Gamma$.  

    Suppose, for contradiction, that there exists a convex body $K$ such that 
    \[
        K \cap \mathbb{Z}^d \subseteq \mathcal{S}(A) \quad \text{and} \quad |K \cap \mathbb{Z}^d| = (1- o(1)) |\mathcal{S}(A)|.
    \] 
    One can verify that the convex hull of $\mathcal{S}(A)$ contains a large hyperrectangle of side lengths $N_1, \ldots, N_d$. Therefore, $K$ must also contain a large hyperrectangle. More precisely, there exists a vector $\vec{a} \in \mathbb Z^d$ such that $\vec{a}+[0,p_1-1]\times \cdots \times [0,p_d-1]\subseteq K$. This contradicts the fact that $\mathcal{S}_\Gamma(A)$ contains at most $(1-o(1))\det(\Gamma)$ distinct residue classes of $\mathbb{Z}^d_\Gamma$.
\end{proof}

\section{Testing Nondegeneracy}\label{app:decide-nondegen}
\begin{lemma}\label{lem:degenerate-or-not}
    Let $A \subseteq [N]^d$ be a set of $n$ vectors, where $Z: =\mathcal{L}(A)$ is full-dimensional. Let $\tau := \sqrt{N^d/\det(Z)}$. Given a real $\delta > 0$, in $\widetilde{O}((n + d^d\tau)\cdot \mathrm{poly}(d))$ time, with error probability at most $N^{-\Omega(d)}$, we can 
    \begin{itemize}
        \item either assert that $A$ is $\delta$-nondegenerate;

        \item or assert that $A$ is not $\delta\cdot (d\log N)^{4d}$-nondegenerate.
    \end{itemize}
\end{lemma}
\begin{proof}
    By Lemma~\ref{lem:compute-hnf}, in $\widetilde{O}(d^3n\log N)$ time, we can obtain $Z := \mathcal{L}(A)$. If $Z$ is not full-dimensional, then we can immediately conclude that $A$ is not $\delta\cdot (d\log N)^{4d}$-nondegenerate.  Assume that $Z$ is full-dimensional. 

    Let $k := \delta \cdot \sqrt{N^d/\det(Z)}$.  We shall construct a sequence $A_0 \supsetneq A_1 \supsetneq \cdots \supsetneq A_r$ until $Z_r := \mathcal{L}(A_r)$ is not full-dimensional or $|A_r \setminus \Gamma| \geq k$ for any $\Gamma \subsetneq Z_r$. Initially, $i := 0$ and $A_0 := A$. Given $A_i$, we can check whether $Z_i := \mathcal{L}(A_i)$ is full-dimensional or not in $\widetilde{O}(d^3n\log N)$ time via Lemma~\ref{lem:compute-hnf}. If not, then we stop and set $r := i$. Suppose that $Z_i$ is full-dimensional. Then we invoke Lemma~\ref{lem:almost-common-lattice} on $A_i$ and $k$, which takes $\widetilde{O}((n + \sqrt{N})\cdot \mathrm{poly}(d))$ time. If it asserts that $|A_i \setminus \Gamma| \geq k$ for any $\Gamma \subsetneq Z_i$, then we stop and set $r := i$. Otherwise, it returns a $\Gamma_i \subsetneq Z$ such that 
    \[
        |A_i \setminus \Gamma_i| \leq k\cdot (d\log N)^{2d}.
    \]
    We set $A_{i+1} := A_i \cap \Gamma_i$, and then proceed with $i := i+1$.

    We first show that $r \leq 1 +  d\log N$. Suppose not. Note that $Z_{r-1}$ must be full-dimensional. Observing that $Z_{i} \subsetneq Z_{i-1}$ for all $i$, we have 
    \[
        \det(Z_{r-1}) \geq 2^{r-1} \cdot \det(Z_0) > N^d.
    \] 
    But this is impossible since $Z_{r-1} = \mathcal{L}(A_{r-1})$ and $A_{r-1} \subseteq [N]^d$.

    Note that $|A_{i}| - |A_{i+1}| \leq k\cdot (d\log N)^{2d}$. We have
    \[
        |A_r| \geq |A| - rk\cdot (d\log N)^{2d} \geq |A| - k\cdot (d\log N)^{4d}.
    \]
    Therefore, if the construction ends with $Z_r$ not full-dimensional, then we can assert that $A$ is not $\delta\cdot (d\log N)^{4d}$-nondegenerate. Otherwise, it ends with $|A_r \setminus \Gamma| \geq k$ for any $\Gamma \subsetneq Z_r$. This implies that $A_r$ is $\delta$-nondegenerate, and hence so is $A$. 

    It is easy to verify that the total running time is $\widetilde{O}((n + \sqrt{N})\cdot \mathrm{poly}(d))$ time and the total error probability is at most $N^{-\Omega(d)}$.
\end{proof}

\section{Deciding Zonotope Membership}\label{apx:zonotope}

We shall prove the following lemma.
\begin{restatable}{lemma}{lemellipsoid}\label{lem:ellipsoid}
    Let $A \subseteq [N]^d$ be a set of vectors, and let $\vec{t} \in \mathbb{R}^d$. We can decide whether $\vec{t}\in \mathcal{P}(A)$ in $O(d^3n\log N)$ time.
\end{restatable}

Let $P :=\{\sum_{i=1}^n \vec{a}_ix_i: -1\leq x_i\leq 1\}\subseteq \mathbb{R}^d$, and let $\vec{t}\in \mathbb{R}^d$ be an arbitrary point. The goal is to determine whether $\vec{t}\in P$. Without loss of generality, we assume that $P$ is not contained in any lower-dimensional subspace of $\mathbb{R}^d$, since otherwise, we can reduce the problem to a lower-dimensional subspace as follows: if $\vec{t}$ does not lie on that subspace, the decision problem is trivially ``no''; otherwise, we simply project all $\vec{a}_i$'s and $\vec{t}$ to the subspace.

We shall reformulate the above problem into another linear program with $O(d)$ variables and many (possibly an exponential number of) constraints, and then show that the resulting linear program can be solved in linear time using the ellipsoid method.

By Lemma~\ref{lem:roc72}, $\vec{t}\in P$ if and only if 
\begin{eqnarray}\label{eq:zono-test11}
 \vec{t}\cdot\vec{y} \le \max_{\vec{p}\in P} \vec{p}\cdot\vec{y}, \quad \forall \vec{y}\in\mathbb{R}^d.
\end{eqnarray}
Since $P$ is centrally symmetric, it is easy to see that $\max_{\vec{p}\in P} \vec{v}\cdot\vec{p}\ge 0$ for all $\vec{v}\in\mathbb{R}^d$. Thus, by scaling,~\eqref{eq:zono-test11} is equivalent to the following:
\begin{eqnarray}\label{eq:zono-test2}
    \vec{t}\cdot\vec{y}\leq \max_{\vec{p}\in P} \vec{p}\cdot\vec{y}, \quad \forall \vec{y}\in\mathbb{R}^d \text{ such that } \max_{\vec{p}\in P} \vec{p}\cdot \vec{y} = 1
\end{eqnarray}
Given that $P$ is a zonotope, we further have the following.
\begin{lemma}\label{lem:zero-test-1}
    Given $\vec{y}\in\mathbb{R}^d$, $\vec{p}\cdot \vec{y}\le 1$ for all $\vec{p}\in P$ if and only if $\sum_i|\vec{a}_i\cdot\vec{y}|\le 1$.
\end{lemma}

\begin{proof}
    Suppose $\sum_i | \vec{a}_i\cdot\vec{y}|\le 1$. Then for arbitrary $\vec{p}\in P$, $\vec{p}\cdot\vec{y}=\vec{y}\cdot \sum_{i}\vec{a}_ix_i=\sum_i (\vec{y}\cdot\vec{a}_i)x_i\le \sum_i |\vec{y}\cdot \vec{a}_i|\le 1$.
    
    Suppose  $\vec{p}\cdot\vec{y}\le 1$ for all $\vec{p}\in P$. Suppose, to the contrary, that $\sum_i |\vec{y}\cdot \vec{a}_i|> 1$. Let $x_i^*=1$ if $\vec{y}\cdot \vec{a}_i\ge 0$, and $x_i^*=-1$ otherwise. Then it is clear that $\vec{p}^*=\sum_i\vec{a}_ix_i^*\in P$, and $\vec{y}\cdot\vec{p}^*=\sum_i |\vec{y}\cdot \vec{a}_i|>1$, which is a contradiction.  
\end{proof}

Thus, $\vec{t}\not\in P$ if and only if there exists some $\vec{y}\in \mathbb{R}^d$ such that $\sum_i|\vec{y}\cdot \vec{a}_i|\le 1$, and $\vec{y}\cdot\vec{t}>1$.

To determine whether $\vec{t}\in P$, it suffices to solve the following linear program.
\begin{align}
    \max \quad & \vec{t}\cdot\vec{y} \label{eq:lp}\\
     s.t.\quad   &\sum_i|\vec{a}_i\cdot \vec{y}|\leq 1 \nonumber\\
        &\vec{y}\in \mathbb{R}^d \nonumber
\end{align}
    
First, we observe that the constraint $\sum_i|\vec{a}_i\cdot \vec{y}|\leq 1$ can be expressed as $2^n$ constraints, each of the form $\sum_i \alpha_i (\vec{a}_i\cdot \vec{y})\le 1$ where $\alpha_i\in\{1,-1\}$. Thus,~\eqref{eq:lp} is indeed a linear program. 

Second, we observe that if the optimal objective value of~\eqref{eq:lp} is larger than 1, then its optimal solution $\vec{y}^*$ guarantees that $\sum_i|\vec{y}^*\cdot \vec{a}_i|\le 1$, and $\vec{y}^*\cdot\vec{t}>1$, implying that $\vec{t}\not\in P$. On the other hand, if the optimal objective value of ~\eqref{eq:lp} is smaller than or equal to 1, then there does not exist any $\vec{y}$ such that $\sum_i|\vec{y}\cdot \vec{a}_i|\le 1$, and $\vec{y}\cdot\vec{t}>1$. By Lemma~\ref{lem:zero-test-1}, this means $\vec{t}\in P$.

The linear program~\eqref{eq:lp} has only $d$ variables, but $2^n$ constraints, so we employ the ellipsoid method. It suffices to design an algorithm for the feasibility test problem, namely, whether $\mathcal{F}=\{\vec{y}: \vec{t}\cdot\vec{y}\ge \beta,  \sum_i|\vec{a}_i\cdot \vec{y}|\le 1\}$ is nonempty for an arbitrary $\beta\in \mathbb{R}$. To solve the feasibility test problem, it suffices to design a separation oracle that does the following: given an arbitrary point $\vec{y}_0\in \mathbb{R}^d$, either assert that $\vec{y}_0\in\mathcal{F}$, or assert that $\vec{y}_0\not\in\mathcal{F}$ and return a constraint that $\vec{y}_0$ violates.

We design a separation oracle as follows. Given $\vec{y}_0$, we first check whether $\vec{t}\cdot\vec{y}_0\ge \beta$. If this constraint is violated, then we assert $\vec{y}_0\not\in\mathcal{F}$ and return this constraint; otherwise, we compute $\sum_i |\vec{a}_i\cdot \vec{y}_0|$. If it is at most $1$, we assert $\vec{y}_0\in\mathcal{F}$. Otherwise, we assert $\vec{y}_0\not\in\mathcal{F}$, and return the violating constraint $\sum_i \alpha_i\vec{a}_i\cdot \vec{y} \leq 1$, where $\alpha_i = 1$ if $\vec{a}_i \cdot \vec{y}_0 \geq 0$ and $\alpha_i = -1$ otherwise.

It is easy to verify that each call of the separation oracle takes $O(dn)$ time. The number of calls to the separation oracle is $O(d^2\log N)$. Thus,  solving linear program~\eqref{eq:lp} takes $O(d^3n\log N)$ time.

\section{Computing Lattices and Remainders}\label{apx:comp-latt}
\begin{fact}\label{fact:suplatt-details}
    Let $\Gamma, \Gamma' \subseteq \mathbb{Z}^d$ be two full-dimensional lattices with $b_{11}, \ldots, b_{dd}$ and $b'_{11}, \ldots, b'_{dd}$ being the diagonal entries of their basis in Hermite normal form. Suppose that $\Gamma' \supsetneq \Gamma$. Then $b'_{ii}$ is a divisor of $b_{ii}$ for all $i \in \{1, \ldots, d\}$, and at least one $b'_{ii}$ is a proper divisor.
\end{fact}

\lemcomputesinglerem*
\begin{proof}
    We first prove (i).
    Let $B$ be the basis of $\Gamma$ in Hermite normal form. By Definition~\ref{def:hnf}, $B$ is a lower-triangular matrix. Hence, $\vec{v} \bmod \Gamma$ can be computed as follows. For each $j \in \{1, 2, \dots, d\}$, compute 
    \(
        x_j = \left\lfloor \frac{v_j }{b_{jj}} \right\rfloor,
    \)
    and update
    \(
        \vec{v} := \vec{v} - x_j \vec{b}_j.
    \)
    The total time cost is $O(d^2)$.

    Next, we prove (ii). Let $B$ and $B'$ be bases of $\Gamma$ and $\Gamma'$ in Hermite normal form, respectively. Let $s_j : = b_{jj}/b'_{jj}$ for $j \in [d]$. By Fact~\ref{fact:suplatt-details}, $s_j$ is an integer. To compute the representatives of $\Gamma' \pmod \Gamma$, it suffices to consider the remainders (modulo $\Gamma$) of the vectors of the form
    \[
        \vec{v} = \sum_{j=1}^d z_j \vec{b}'_j,
    \]
    where $0 \leq z_j < s_j, z_j \in \mathbb{Z}$, 
    since the number of such vectors is exactly $\prod_{i=1}^d s_i = \frac{\det(\Gamma)}{\det(\Gamma')} = |\Gamma' \bmod \Gamma|$, and one can also verify that no two such vectors are congruent modulo $\Gamma$.  By (i), it takes $\mathcal{O}(d^2)$ time to compute $\vec{v} \bmod \Gamma$ for each $\vec{v}$. Therefore, the total time cost is  $O(\frac{|\det(\Gamma)|}{|\det(\Gamma')|}\cdot d^2)$.

    Next we prove (iii).  Let $B\in \mathbb Z^{d\times d}$ be the basis of $\Gamma$ in Hermite normal form.
    We can first update $\vec{h} := \vec{h} \bmod \Gamma$ in $O(d^2)$ time by (i).
    Let $h_1,\dots,h_d$ be the entries of vector $\vec h$.
    For $i \in \{1, \ldots, d\}$, we iteratively update $\vec{b}_i$ and $\vec{h}$ to make $h_i=0$ while keeping the lattice generated by $\vec b_1,\ldots,\vec b_d,\vec h$ unchanged.

    We use the extended Euclidean algorithm to compute $g_i=\gcd(b_{ii},h_i)$ and integers $x,y$ such that
    \[
        xb_{ii}+yh_i=g_i.
    \]
    This step takes $O(\log(\max\{b_{ii},h_i\}))=O(\log \ell)$ time. We then update $\vec{b}_i$ and $\vec{h}$ as follows:
    \[
        \vec b_i'=x\vec b_i+y\vec{h}, \qquad
        \vec h'=-\frac{h_i}{g_i}\vec b_i+\frac{b_{ii}}{g_i}\vec h.
    \]
    By updating $\vec{b}_i := \vec{b}_i'$ and $\vec{h} := \vec{h}'$, we have $b_{ii}=g_i$ and $h_i=0$, and the lattice generated by $\vec b_i$ and $\vec h$ is unchanged. Each such update uses $O(d)$ time.

    Repeating this for $i=1,\dots,d$ yields a basis in Hermite normal form for $\Gamma'$, and the total time cost is $O(d(d+\log \ell))$.
\end{proof}

\bibliographystyle{alphaurl}
\bibliography{main}
\end{document}